\documentclass[a4paper, ariel, 11pt]{amsart}
\usepackage{amscd,amsthm,amsfonts,latexsym,amssymb}
\usepackage[dvips]{graphicx}

\usepackage{bez123,calc,curves,ebezier,epic,eepic,graphicx,multiply,rotating}

\newcommand{\ds}{\displaystyle}
\newtheorem{theorem}{Theorem} [section]
\newtheorem{lemma}[theorem]{Lemma}
\newtheorem{corollary}[theorem]{Corollary}

\newtheorem{proposition}[theorem]{Proposition}
\newtheorem{example}[theorem]{Example}

\newtheorem{definition}[theorem]{Definition}
\newtheorem{remark}[theorem]{Remark}

\newcommand{\inn}[1]{\langle#1\rangle} 

\newcommand{\dd}{{\mathrm d}}  

\newcommand{\RR}{{\mathbb R}}  
\newcommand{\CC}{{\mathbb C}}  

\newcommand{\NN}{{\mathbb N}}  

\newcommand{\Ee}{{\mathcal E}}
\newcommand{\Ff}{{\mathcal F}}
\newcommand{\Pp}{{\mathcal P}}

\newcommand{\CP}{{\mathbb C}P} 

\newcommand{\ii}{{\rm i}}  

\newcommand{\ra}{\rightarrow}

\newcommand{\pa}{\partial}

\renewcommand{\phi}{\varphi}
\newcommand{\la}{\lambda}
\newcommand{\La}{\Lambda}
\newcommand{\na}{\nabla}
\newcommand{\al}{\alpha}
\newcommand{\be}{\beta}
\newcommand{\ga}{\gamma}
\newcommand{\Ga}{\Gamma}
\newcommand{\ve}{\varepsilon}
\newcommand{\si}{\sigma}
\newcommand{\Si}{\Sigma}

\newcommand{\om}{\omega}

\newcommand{\ta}{\theta}

\newcommand{\wt}{\widetilde} 
\newcommand{\ov}{\overline} 
\newcommand{\wh}{\widehat} 

\begin{document}

\title{A class of quadratic difference equations on a finite graph}
\author{Paul Baird}
\thanks{The author is grateful for support provided by the Australian
Research Council and to the Mathematical
Sciences Institute at the Australian National University for its support and
hospitality; he also thanks Mike Eastwood for many helpful conversations in relation to this work.}

\address{D\'epartement de Math\'ematiques \\
Universit\'e de Bretagne Occidentale \\
6 av.\ Victor Le Gorgeu -- CS 93837 \\
29238 Brest Cedex, France}

\email{Paul.Baird@univ-brest.fr}

\begin{abstract} 
We study a class of complex polynomial equations on a finite graph with a view to understanding how holistic phenomena emerge from combinatorial structure.  Particular solutions arise from orthogonal projections of regular polytopes, invariant frameworks and cyclic sequences.  A set of discrete parameters for which there exist non-trivial solutions leads to the construction of a polynomial invariant and the notion of a geometric spectrum.  Geometry then emerges, notably dimension, distance and curvature, from purely combinatorial properties of the graph.  
\end{abstract}



\maketitle

\thispagestyle{empty}

\tableofcontents

\section{Introduction} Given a graph $\Ga = (V, E)$, with vertices $V$ and edges $E$ together with a function $\phi : V \ra \CC$, the equations we wish to study all have the form
\begin{equation} \label{one}
\frac{\ga (x)}{n(x)} \Big( \sum_{y\sim x} \big(\phi (y) - \phi (x)\big)\Big)^2 =  \sum_{y\sim x} \big(\phi (x) - \phi (y)\big)^2\,,
\end{equation}
at each vertex $x$, where $y\sim x$ means $y$ is adjacent to $x$, $n(x)$ is the degree of $\Ga$ at $x$ (the number of vertices adjacent to $x$) and where $\ga : V \ra \RR$ is a \emph{real}-valued function, which in many situations we will suppose constant; then we call the solution \emph{regular}.  A special case is when $\ga \equiv0$ and we have the equation:
\begin{equation} \label{holo}
0= \sum_{y\sim x} \big(\phi (x) - \phi (y)\big)^2\,,
\end{equation}
for all $x\in V$, solutions of which we refer to as \emph{holomorphic functions}.  
Note that the equations are invariant by the replacement of $\phi$ by $\wt{\phi} = \la\phi + \mu$ where $\la$ and $\mu$ are complex numbers, as well as with respect to complex conjugation $\phi \mapsto \ov{\phi}$. 

In what follows, it is convenient to write $\Delta \phi (x) =  \frac{1}{n(x)}\sum_{y\sim x} \big(\phi (y) - \phi (x)\big)$ (the Laplacian) and $(\dd \phi )^2(x) = \frac{1}{n(x)}\sum_{y\sim x} \big(\phi (y) - \phi (x)\big)^2$ (the symmetic square of the derivative), whereby equation (\ref{one}) becomes:
$$
\ga (x) \Delta \phi (x)^2 = (\dd\phi)^2(x)\,.
$$
All our graphs will be \emph{simple}, that is we do not allow loops or multiple edges, although much of what we discuss can be generalized to non-simple graphs.

For a given graph $\Ga = (V,E)$, we are particularly interested in the set $\Si$ of constant values of $\ga$ that can occur for which (\ref{one}) has non-trivial solutions:
$$
\Si:= \{ \ga\in \RR : \exists \ {\rm non-constant} \  \phi : V \ra \CC \ {\rm such \ that} \  \ga\Delta \phi^2 = (\dd\phi)^2 \}\,.
$$ 
We will call this set the \emph{geometric spectrum of $\Ga$}, which we regard as an intrinsic object associated to the graph, somewhat akin to the spectrum of the Laplacian.  We will see that this ``spectrum" reflects geometric properties, as well as combinatorial properties of $\Ga$.  Note that on a finite graph, by the maximum principle, $\phi$ non-constant implies $\Delta \phi \not\equiv 0$, so that $\Si$ is always well-defined.  

Particle physics provides some motivation to study these equations.  As a lattice model, equation (\ref{one}) expresses the influence exerted on a particular vertex by each of its neighbours.  In the case of a graph of degree three, in \cite{Ba-We} it is shown that equation (\ref{holo}) provides a discrete analogue of the equation for a shear-free ray congruence on space-time.  In the smooth setting such congruences can be used to generate solutions to the zero rest-mass field equations.  In Section \ref{sec:particles}, we refer to a connected graph $\Ga$ as a \emph{particle} and an equivalence class of solutions to (\ref{one}) on $\Ga$ as a \emph{state} of that particle.  States which have $\ga$ constant will be referred to as \emph{isostates}.  These are significant, since empirical evidence suggests that isostates occur as extremal states of a natural energy functional derived in Section \ref{sec:distance}.  We take the view that (i) at a fundamental level, the description of the world should be combinatorial and (ii) the only physical quantities that have meaning are relative. 

The idea that combinatorial structures should be at the basis of the description of matter, was put forward by R. Penrose in a well-known paper of 1971 \cite{Pe}.  In this article, he introduced the idea of a spin network, which is a graph of degree three whose edges are labeled by integers which represent twice the angular momentum.  The notion has been generalized in recent years by C. Rovelli and L. Smolin in their development of loop quantum gravity, see \cite{Ro}. Another key idea of Rovelli, is the relational interpretation of quantum mechanics, which puts forward the thesis that the only quantities that have meaning are relational. Thus no particle exists in an absolute state; the only meaningful way to describe a particle is how it correlates with other particles, or systems \cite{Ro-0}.

This philosophy is consistent with the idea that an isolated system be described by a graph coupled to a complex field.  Firstly, a graph is a combinatorial structure that expresses a binary relation between its vertices.  We do not need to imagine the graph as lying in some ambient space; the only relation that matters is whether one vertex is joined to another or not.  Of course, one may envisage more complicated combinatorial relations in nature, so perhaps edges should be labeled by an integer, or by a group representation, as in loop quantum gravity.  Secondly, the introduction of the field $\phi$ on the graph, represents an additional structure from which geometry emerges.  The equations which govern this field should enjoy the invariance $\phi \mapsto \la \phi + \mu$, so that only relative values of $\phi$ are significant. 

The condition that $\ga$ be real in equation (\ref{one}) is justified by Theorem \ref{thm:lift}.  This expresses the fact that under reasonable hypotheses, there is a unique realization of the solution about a vertex as a regular star in a Euclidean space; this enables us to develop geometry, in particular curvature and distance.  The constancy of $\ga$ then implicates global properties of the graph: it is a uniformizing parameter.  It is the case that the framework of a regular polytope, when projected orthogonally into the complex plane, determines a solution to (\ref{one}) with $\ga$ constant (Theorem \ref{thm:reg-polytope}).  

The factor $1/n(x)$ on the left-hand side of (\ref{one}) is explained by equation (\ref{star-gamma}) of Corollary \ref{cor:configured-star}; this means that $\ga$ depends only on local geometric invariants and not on the degree.  This choice also provides a consistent boundary, once more independent of the degree, between real solutions and complex solutions (see Lemma \ref{lem:ngaleq1}); there are also heuristic arguments that arise from the functional analytic perspective given in Appendix \ref{sec:lin}. 

A notion of \emph{holomorphic function} somewhat similar to ours has been introduced by S. Barr\'e \cite{Ba}.  However, in addition to (\ref{holo}), Barr\'e requires that $\phi$ be harmonic.  This is quite restrictive and in particular, by the maximum principle, with this definition no non-constant holomorphic function could exist on a finite graph.

We begin our paper with some elementary remarks concerning graph colourings.  The next three sections then provide examples of solutions to equation (\ref{one}) which will serve as a reference for subsequent developments.  We give a complete characterization on a cyclic graph;  for $\ga$ constant, real solutions can be described in terms of polynomial equations with integer coefficients (Theorem \ref{thm:cyclic}).  Complex solutions correspond to polygonal chains (or planar bar frameworks) whose edges all have the same length.  In Appendix \ref{sec:trig}, we characterize the regular polygons or star polygons as those solutions having $\ga$ constant.  Indeed the energy functional that we derive in Section \ref{sec:distance} (see below) defines a gradient flow on the configuration space of polygonal chains which we conjecture determines an evolution towards the most compact regular configuration.  

Following work of Eastwood and Penrose \cite{Ea-Pe}, in Section \ref{sec:orthographic} we discuss how solutions arise from orthogonal projections of regular polytopes.  An important step is Corollary \ref{cor:configured-star}, which shows that a certain class of star graphs in $\RR^N$ have invariant properties.  In particular, their projections to $\CC$ satisfy (\ref{one}) at the central vertex with $\ga$ independent of the position of the star. By extending the class of invariant star graphs in Section \ref{sec:inv-polytopes}, we describe a more general family of \emph{invariant frameworks} which define solutions to (\ref{one}).  For example, we show that it is possible to place a complete graph on any number of vertices in $\RR^3$ in an invariant way (Corollary \ref{cor:inv-complete-graph}).  By this we mean that the projections to $\CC$ of the vertices satisfy (\ref{one}) with $\ga$ constant, independently of any rigid motion or dilation of the framework.  More sophisticated invariant structures are described in Appendix \ref{sec:invariant-structures}.  When we change our perspective and begin with a combinatorial structure together with a spectral value $\ga$, these examples take on added significance as ``geometric objects". 

In Section \ref{sec:distance} we consider the converse question of how to lift a solution to (\ref{one}) at each vertex to an embedded star.  This will enable us to associate both distance and curvature to a graph $\Ga$ coupled to a field $\phi$, by considering at each vertex the best-fit polytope with regular vertex figure.  We call this the \emph{lifting problem}; its solution, expressed by Theorem \ref{thm:lift}, is at the heart of our study.  It transpires that when we lift to dimension three and only to this dimension, the solution is the unique minimum of a natural functional determined by the field $\phi$.  This \emph{energy} functional plays an important role in the formulation of our elementary universe in Section \ref{sec:particles} (Definition \ref{def:energy}).   

Curvature is defined in Section \ref{sec:curvature} by analogy with the well-known notion for polytopes embedded in a Euclidean space, in terms of angular deficiency.  A different, \emph{edge-curvature}, measures the angular tilt between the axes of adjacent vertex figures.  This leads to analogues in a discrete setting of sectional and Ricci curvatures.   

An important construction is that of a polynomial invariant associated to a finite connected graph $\Ga$. This is defined to be the least degree univariate polynomial $p_{\ga}$ in the spectral parameter $\ga$ lying in the ideal generated by the equations (\ref{one}) with $\ga$ now constant and complex, after we have taken into account the normalization $\phi \mapsto \la \phi +\mu$.  The geometric spectrum occurs as real roots of this polynomial, however, there may be other real roots that are not in the spectrum.  In Section \ref{sec:spectrum}, we use Gr\"obner bases to find $p_{\ga}$ for a number of examples. 

Our ultimate objective is to describe an elementary universe populated entirely by graphs, from which geometry and dynamics emerge.  For an individual graph, its geometry is implicit in the geometric spectrum.  However, the nature of this geometry (which spectral value applies) should only become manifest upon correlation with another graph:  it is a relative concept.  Dynamics corresponds to change which occurs when two graphs combine, a graph mutates, or vertices and edges are created or annihilated.  With an appropriate definition of thermal time, the universe is then endowed with local geometry and time.

\section{Elementary observations on graph colourings} \label{sec:colourings}
If a vertex $x$ has degree one, that is, it has precisely one neighbour $y$, then if $\phi$ is any function with $\phi (y) \neq \phi (x)$, equation (\ref{one}) is satisfied at $x$, with $\ga (x) = 1$.  More generally, if we can colour the vertices of the graph with two colours in such a way that every vertex is adjacent to at most one vertex of a different colour, then once more (\ref{one}) is satisfied with $\ga (x) = n(x)$ for all $x\in V$.  To see this, we just associate the values $0$ and $1$ to the two colours.  This is easily generalised to the following.

\begin{lemma} \label{lem:col-1} Let $\Ga = (V,E)$ be a graph whose vertices are assigned two colours, $0$ and $1$ say, and let $\phi (x) =$ \emph{colour of vertex} $x$, for each $x\in V$.  Suppose $\Ga$ is coloured in such a way that for some function $k(x)$, \emph{either} each vertex $x$ is adjacent to precisely $k(x)$ of a different colour, \emph{or} each vertex is adjacent to no vertex of a different colour.  Then $\phi$ satisfies {\rm (\ref{one})} with $\ga (x) = n(x)/k(x)$ for all $x\in V$.
\end{lemma}  

So solutions to (\ref{one}) arise from particular $2$-colourings of the graph.  But what about $3$-colourings, or $m$-colourings?  

Consider the three colours $\{ 0,1, \frac{1}{2} + \frac{\sqrt{3}}{2}\ii\}$, corresponding to the positions of the vertices of an equilateral triangle in the plane.  Now, whatever vertex $x$ we take, if it is adjacent to precisely one of each of the other two values, equation (\ref{one}) is satisfied with $\ga = n(x)/3$.  Similarly, a colouring using four colours with this property (with $\ga$ replaced by $n(x)/4$) is given by the four complex numbers $\{ 0, 1, \ii , 1 + \ii\}$.  These points correspond to the projections of the vertices of a regular tetrahedron in $\RR^3$ onto the complex plane.  

A result of Eastwood and Penrose \cite{Ea-Pe}, affirms that the co\"ordinates of any orthogonal projection of the vertices of a regular simplex in $\RR^N$ into $\CC$, provide $N+1$ colours that can be used to generate solutions to (\ref{one}).  Specifically, let $\{ v_1, v_2, \ldots , v_{N+1}\}$ be the vertices of a regular $N$-simplex in $\RR^N$ and let $P:\RR^N \ra \CC$ be any orthogonal projection.  Let $S:=\{ z_1, z_2, \ldots , z_{N+1}\}$, where $z_k = P(v_k)$ ($k = 1, \ldots , N+1$).

\begin{lemma} \label{lem:col-1}  Let $\Ga = (V,E)$ be a graph whose vertices are assigned $N+1$ colours taken from the set $S$.  Let $\phi (x) =$ \emph{colour of vertex} $x$.   Suppose $\Ga$ is coloured in such a way that 
\emph{either} each vertex is adjacent to precisely one of each of the other $N$ colours, \emph{or} each vertex is adjacent to no vertex of a different colour.  Then $\phi$ satisfies {\rm (\ref{one})} with $\ga (x) = n(x)/(N+1)$ for all $x\in V$.
\end{lemma}

\begin{example} \label{ex:bipartite} {\rm A bipartite graph $K_{m,n}$ is a graph with $m+n$ vertices which can be separated into two sets: $\{x_1, \ldots , x_m\}$ and $\{y_1, \ldots , y_n\}$, with the property that for each $j = 1, \ldots , m$, we have $x_j\sim y_r$ for all $r = 1, \ldots , n$; $x_j\not\sim x_k$ for all $j\neq k$; $y_r\not\sim y_s$ for all $r\neq s$.  The bipartite graph $K_{n,n}$ admits solutions deriving from both of the above lemmas.  In the first case, we can colour the vertices $x_1, \ldots , x_n$ with $0$ and the vertices $y_1, \ldots , y_n$ with $1$.  Then the corresponding function $\phi$ solves (\ref{one}) with $\ga = 1$.  In the second case, we let both $x_1, \ldots , x_n$ and $y_1, \ldots , y_n$ to be the orthogonal projections $z_1, \ldots , z_n$ of the vertices of a regular $(n-1)$-simplex in $\RR^{n-1}$, in any order.  Once more (\ref{one}) is satisfied with $\ga = 1$. 
\medskip
\begin{center}
\setlength{\unitlength}{0.254mm}
\begin{picture}(149,92)(115,-156)
        \allinethickness{0.254mm}\path(115,-140)(115,-80) 
        \allinethickness{0.254mm}\path(175,-80)(175,-140) 
        \allinethickness{0.254mm}\path(175,-140)(115,-80) 
        \allinethickness{0.254mm}\path(115,-140)(140,-115) 
        \allinethickness{0.254mm}\path(150,-105)(175,-80) 
        \allinethickness{0.254mm}\path(235,-140)(235,-80) 
        \allinethickness{0.254mm}\path(175,-80)(235,-140) 
        \allinethickness{0.254mm}\path(175,-140)(200,-115) 
        \allinethickness{0.254mm}\path(210,-105)(235,-80) 
        \allinethickness{0.254mm}\path(115,-140)(145,-125) 
        \allinethickness{0.254mm}\path(235,-80)(205,-95) 
        \allinethickness{0.254mm}\path(190,-105)(180,-110) 
        \allinethickness{0.254mm}\path(170,-115)(160,-120) 
        \allinethickness{0.254mm}\path(115,-80)(150,-100) 
        \allinethickness{0.254mm}\path(235,-140)(205,-125) 
        \allinethickness{0.254mm}\path(190,-120)(180,-115) 
        \allinethickness{0.254mm}\path(170,-110)(160,-105) 
        \put(115,-156){\shortstack{$x_1$}} 
        \put(175,-156){\shortstack{$x_2$}} 
        \put(235,-156){\shortstack{$x_3$}} 
        \put(115,-72){\shortstack{$y_1$}} 
        \put(175,-72){\shortstack{$y_2$}} 
        \put(235,-72){\shortstack{$y_3$}} 
\end{picture}
\end{center}
\medskip

We can explicitly find all solutions to (\ref{one}) on the bipartite graph $K_{33}$ illustrated above when $\ga = 1$.  If we first normalize so that $x_1=0$ and $y_1=1$, then the solutions are given by $x_2=x_3=0$, $y_3=\la$ an arbitrary real parameter, $y_2= \frac{1\pm \ii\sqrt{3}}{2} + \la \left(\frac{1\mp \ii\sqrt{3}}{2}\right)$.  As $\la$ varies from $1$ to $0$, the solution interpolates between $y_1=y_2=y_3=1$ and $y_1=1, y_2=\frac{1\pm \ii \sqrt{3}}{2}, y_3=0$, which are the vertices of an equilateral triangle.

If we now normalize so that $x_1=0$ and $x_2 = 1$, then necessarily $x_3 = \frac{1\pm \ii \sqrt{3}}{2}$.  We then have a $2$-parameter family of solutions given by arbitrarily prescribing $y_2=\la$, $y_3= \mu$ and then
$y_1=\la\left(\frac{1\pm \ii \sqrt{3}}{2}\right) + \mu\left(\frac{1\mp \ii \sqrt{3}}{2}\right)$.  In the case when $\mu = 0$ and $\la = 1$, we once more have the vertices of an equilateral triangle, corresponding to the hypotheses of Lemma \ref{lem:col-1}.  It turns out that for $K_{33}$,  $\ga = 1$ is the only possible constant value of $\ga$ for which (\ref{one}) has non-constant solutions. This will be justified in Section \ref{sec:spectrum}. }
\end{example}

\section{Regular cyclic sequences} \label{sec:cyclic}
We shall call a solution to (\ref{one}) on a cyclic graph with $\ga$ constant, a \emph{regular cyclic sequence}.  Real regular cyclic sequences have a particularly simple construction from polynomial equations with integer coefficients, as we explain below.  First we deal with the trivial case when $\ga =2$.

Let $(x_0,x_1,x_2,\ldots ,x_{N-1}, x_N=x_0)$ be a regular cyclic sequence which solves (\ref{one}) with $\ga =2$.  Consider a particular segment of three successive terms $x_{k-1}, x_k, x_{k+1}$.  On applying (\ref{one}) at the vertex $x_k$, we obtain the equation:
\begin{eqnarray*}
 & & (x_{k+1}+x_{k-1}-2x_k)^2 = (x_{k-1}-x_k)^2 + (x_{k+1}-x_k)^2 \\
 & \Leftrightarrow &(x_k-x_{k-1})(x_k-x_{k+1})=0\,,
\end{eqnarray*}
so that necessarily, $x_k$ is equal to one of its neighbours.  Conversely, as described in the previous section, if every vertex is adjacent to at most one of a different value, then the sequence solves (\ref{one}) with $\ga =2$.  Thus any cyclic graph of order $\geq 4$ admits such sequences:  \emph{the cyclic graph is coloured by numbers (real or complex) in such a way that at each vertex, at least one of its neighbours carries the same colour.}  This amounts to colouring the graph so that connected segments of order at least two have the same colour.  We now consider the general case.  In what follows, we refer to \emph{normalization} as the freedom $\phi \mapsto \la\phi + \mu \ (\la , \mu \in \CC)$.  By a \emph{real} regular cyclic sequence, we mean one in which every term is real under some normalization. 

\begin{theorem} \label{thm:cyclic} {\rm (Construction of real regular cyclic sequences)}:  Take any polynomial $a_nx^n+a_{n-1}x^{n-1} + \cdots + a_1x + a_0$ with integer coefficients all strictly positive.  Multiply the polynomial by $x+1$ to obtain the new polynomial 
\begin{eqnarray*}
p(x) & : = & b_{n+1}x^{n+1} + b_nx^n + \cdots + b_1x + b_0 \\
 & = & a_nx^{n+1} + (a_n+a_{n-1})x^n + \cdots + (a_1+a_0)x + a_0\,.
\end{eqnarray*}
  Let $x = y$ be any real root of $p(x)$.  Then a cyclic sequence 
$$
(x_0, x_1, x_2, \ldots , x_{N-1}, x_N=x_0)
$$ 
of order $N = 2\sum_ka_k$ is constructed by arbitrarily prescribing $x_0$ and then requiring increments $y_{\ell} = x_{\ell}-x_{\ell -1}$ of successive terms to be taken from the set $\{ 1, y, y^2, \ldots , y^{n+1}\}$ in such a way that each increment $y^k$ occurs precisely $b_k$ times and any two adjacent increments have powers that differ by precisely one.  This is always possible and up to these constraints, the ordering is arbitrary.  The constant $\ga$ in {\rm (\ref{one})} is given by $\ga = 2(1+y^2)/(1-y)^2$.  

Conversely, up to a multiple, addition of a constant and cyclic permutations, any real regular cyclic sequence with $\ga\neq 2$ or $1$\ arises this way.  The real regular cyclic sequences with $\ga = 2$ are characterized as those made up of connected segments of order $\geq 2$ on which the sequence is constant; those with $\ga = 1$ oscillate and up to normalization are equivalent to $(0,1,0,1, \ldots , 0,1)$.
\end{theorem}

We refer to the increment $y$ in the above theorem, as the \emph{fundamental increment associated to the real cyclic sequence}.   

\begin{remark}  {\rm Since any root $y$ must be strictly negative and adjacent powers differ by one, it follows that a real regular cyclic sequence must oscillate.  The length of the sequence is given by $\sum_kb_k = 2\sum_ka_k$, so that a non-trivial sequence can only occur on a cyclic graph of even order (which is also a consequence of oscillation).  }
\end{remark}

\begin{example} \label{ex:real-hexagon}  {\rm If we take for our starting polynomial $x+2$, then multiplication by $x+1$ gives the polynomial $x^2+3x+2$ with real root $x = -2$.  We can now arrange the powers of this root with appropriate multipicity to give the sequence of increments $(1,y,1,y,y^2,y) = (1,-2,1, -2, 4, -2)$.  We construct a real regular cyclic sequence of order $6$ by first setting $x_0 = 0$ and then proceeding so that $x_1-x_0 = 1, x_2 - x_1 = - 2$ and so on.  We thereby obtain the sequence $(0,1,-1,0, -2, 2)$ on a cyclic graph of order $6$.  Since the sequence is only defined up to multiple, addition of a constant and cyclic permutations, we can normalize the sequence in such a way that the minimum value is $0$ and that this occurs for the first term: $(0, 4, 2, 3, 1, 2)$. 

\medskip
\begin{center}
\setlength{\unitlength}{0.254mm}
\begin{picture}(143,117)(45,-156)
        \allinethickness{0.254mm}\path(85,-140)(140,-140) 
        \allinethickness{0.254mm}\path(85,-55)(140,-55) 
        \allinethickness{0.254mm}\path(140,-55)(165,-100) 
        \allinethickness{0.254mm}\path(165,-100)(140,-140) 
        \allinethickness{0.254mm}\path(85,-140)(60,-100) 
        \allinethickness{0.254mm}\path(60,-100)(85,-55) 
        \allinethickness{0.254mm}\special{sh 0.3}\put(85,-140){\ellipse{4}{4}} 
        \allinethickness{0.254mm}\special{sh 0.3}\put(140,-140){\ellipse{4}{4}} 
        \allinethickness{0.254mm}\special{sh 0.3}\put(165,-100){\ellipse{4}{4}} 
        \allinethickness{0.254mm}\special{sh 0.3}\put(140,-55){\ellipse{4}{4}} 
        \allinethickness{0.254mm}\special{sh 0.3}\put(85,-55){\ellipse{4}{4}} 
        \allinethickness{0.254mm}\special{sh 0.3}\put(60,-100){\ellipse{4}{4}} 
        \put(80,-156){\shortstack{$0$}} 
        \put(140,-156){\shortstack{$4$}} 
        \put(170,-106){\shortstack{$2$}} 
        \put(140,-51){\shortstack{$3$}} 
        \put(85,-51){\shortstack{$1$}} 
        \put(45,-106){\shortstack{$2$}} 
\end{picture}
\end{center}
\medskip
}
\end{example} 

\begin{example} {\rm  Irrational sequences arise from irrational roots.  For example, let us start with the polynomial $x^2+4x+1$, with root $x = - 2 + \sqrt{3}$.  On multiplying by $x+1$ we obtain the polynomial $x^3+5x^2+ 5x+1$.  A suitable sequence of increments is given by $(1,y,y^2,y,y^2,y,y^2,y,y^2,y^3,y^2,y)$ with $y = - 2 + \sqrt{3}$.  On calculating, we can now construct a real regular cyclic sequence of order $12$; explicitly, it is given by $(0,1,-1+\sqrt{3}, 6-3\sqrt{3}, 4-2\sqrt{3}, 11-6\sqrt{3}, 9-5\sqrt{3}, 16 - 9\sqrt{3}, -10 + 6\sqrt{3}, -3 + 2\sqrt{3}, -5 + 3\sqrt{3}, 2-\sqrt{3})$.  Since the absolute value of the root is $<1$, all terms of this oscillating sequence lie in the interval $[0,1]$.  The value of the constant $\ga$ in (\ref{one}) is given by $\ga = 4/3$. }
\end{example}

In order to prove Theorem \ref{thm:cyclic} we first of all establish a recurrence relation that determines a subsequent term of the sequence in terms of three previous terms.

\begin{lemma} \label{lem:recurrence}  Let $(x_0, x_1, \ldots , x_{N-1}, x_N=x_0)$ be a non-constant regular cyclic sequence satisfying {\rm (\ref{one})} with $\ga\neq 2$, then the increments $y_k = x_{k+1} - x_k$ satisfy the recurrence relation:
$$
y_k = \left\{ \begin{array}{rl} {\rm either} \ & y_{k-1}{}^2/y_{k-2} \\
{\rm or} \ & y_{k-2}\,.
\end{array} \right.
$$
Conversely, any sequence of increments satisfying these relations determines a real regular cyclic sequence.
\end{lemma}

\begin{proof}
Let $(x_0, x_1, \ldots , x_{N-1}, x_N=x_0)$ be a non-constant regular cyclic sequence satisfying {\rm (\ref{one})} with $\ga\neq 2$.  Consider a particular segment of the sequence consisting of four consecutive vertices: $(x_{k-2}, x_{k-1}, x_k, x_{k+1})$.  On normalising, we can suppose this segment equivalent to $(k,0,x,y)$, for some real numbers $k,x,y$.  Note that $k\neq 0$, for otherwise $x$ would have to be zero since if not, we would have $\ga =2$ at vertex $k-1$.  But then proceeding along the cycle, we would eventually encounter a non-zero value at a vertex, which would then imply $\ga =2$ at that vertex; a contradiction.  
On evaluating equation (\ref{one}) at the vertex $x_{k-1}$, we obtain:
\begin{equation} \label{a}
\ga = \frac{2(k^2+x^2)}{(k+x)^2}\,.
\end{equation}

Now evaluate (\ref{one}) at vertex $x_k$:
$$
\ga (y-2x)^2 = 2\big((y-x)^2+x^2\big)\,.
$$    
On eliminating $\ga$, we obtain the quadratic equation in $y$:
$$
ky^2+(k-x)^2y-x(k-x)^2 = 0\,.
$$
This gives the two possible values $x(k-x)/k$ and $x-k$ for $y$.  To recover the general case, we set $k = k_{k-2}-x_{k-1}$, $x=x_k-x_{k-1}$, $y=x_{k+1}-x_{k-1}$.  This gives the two values:
$$
x_{k+1} = \left\{ \begin{array}{l}\ds\frac{x_{k-1}(x_k-x_{k-1})+x_k(x_{k-2}-x_k)}{x_{k-2}-x_{k-1}} \\
x_k+x_{k-1}-x_{k-2}
\end{array}
\right.
$$
On subtracting $x_k$ from both sides, we obtain the recurrence relation for the increments in the statement of the lemma.  The converse can be shown by direct computation.
\end{proof}

\noindent \emph{Proof of Theorem {\rm \ref{thm:cyclic}}}.  Let $(x_0, x_1, \ldots , x_{N-1}, x_N=x_0)$ be a non-constant real regular cyclic sequence satisfying {\rm (\ref{one})} with $\ga\neq 2$.  From Lemma \ref{lem:recurrence}, the sequence of increments $(y_k=x_{k+1}-x_k)$ must alternate in sign, for otherwise, if we have two consecutive increments of the same sign, then all subsequent increments would have the same sign and the sequence $(x_k)$ would be monotone increasing or decreasing, which is impossible.  

First normalise so that $x_1-x_0 = 1, x_2-x_1 = x<0$ and consider the sequence of possible increments:
$$
\Big( 1,x,\left\{ \begin{array}{ll} x^2, & \left\{ \begin{array}{l} x^3, \\ x, \end{array} \right. \\ 1,  & \left\{ \begin{array}{l} 1/x, \\ x, \end{array} \right. \end{array} \right. \cdots \Big)
$$
Suppose first that $x=-1$, then we obtain the sequence of increments $(1,-1,1,-1,\ldots , 1, -1)$ which corresponds to the real regular cyclic sequence $(0,1,0,1, \ldots , 0,1)$ with $\ga  = 1$.  Furthermore, any non-constant regular cyclic sequence with $\ga  = 1$ must have this form, since if we take a segment $x_{k-1}\sim x_k\sim x_{k+1}$, then
$$
(x_k-x_{k-1}+x_k-x_{k+1})^2 = 2(x_k-x_{k-1})^2 + 2(x_k-x_{k+1})^2\ \Leftrightarrow \ (x_{k+1}-x_{k-1})^2=0\,,
$$
so that $x_{k+1} = x_{k-1}$.  Henceforth, suppose that $x\neq -1$.  In particular, since $x$ is real, we cannot have $x^k=1$ for any power $k\neq 0$.

Note that every term in the sequence of increments must be a power of $x$.  Now multiply through by the highest negative power of $x$ to obtain $1 = x^0$ in some position, with all other powers of $x$ greater than or equal to zero:
$$
(x^{r_0}, x^{r_1}, \ldots , 1 , \ldots )\,,
$$
where each $r_k\geq 0$.  Now cyclically permute the sequence to obtain $1$ in the first entry:
$$
(1,x^k, \ldots )
$$
where $k$ is necessarily an odd positive integer.  The recurrence relation then implies that all terms must have the form $(x^k)^r$ for some integer $r\geq 0$, so we now set $y = x^k$, to obtain the sequence:
$$
\Big( 1,y, \left\{ \begin{array}{l} y^2, \\ 1, \end{array} \right. \ldots , y \Big)\,.
$$
On applying the recurrence relation once more, we see that an occurrence of $y^{\ell}$ must be followed by either $y^{\ell +1}$ or $y^{\ell -1}$.  But $\sum_{k=0}^{n-1}y_k = \sum_{k=0}^{n-1}(x_{k+1}-x_k) = 0$, which implies that $y$ satisfies a polynomial equation of the form: 
\begin{eqnarray}
p(x) & := &  \al_{r+1}x^{r+1} + (\al_{r+1}+1)x^r + 2x^{r-1}+2x^{r-2} + \cdots \nonumber \\
  & & \qquad \qquad \cdots  + 2 x^2 + (\al_0+1)x + \al_0 + \sum_{s = 2}^r \be_s(x^s+x^{s-1}) \label{real-poly}\\
 & = & (x+1)(\al_{r+1}x^r + x^{r-1} + x^{r-2} + \cdots + x + \al_0) \nonumber \\
 & &  \qquad \qquad + (x+1)\sum_{s = 2}^r \be_s x^{s-1}\,, \nonumber
 \end{eqnarray}
 where $\al_{r+1}, \al_0$ are strictly positive integers and $\be_s$ are integers that are $\geq 0$. Thus for each of the $\al_{r+1}$ occurrences of the maximum power $y^{r+1}$, we must have at least one more occurrence of $y^r$.  Similarly for each of the $\al_0$ occurrences of the minimum power $y^0$.  Since by assumption $y\neq 1$ so that $y^k\neq 1$ for any $k\neq 1$, all intermediate powers must occur at least twice.  However, we may have further oscillations between powers of $y^s$ and $y^{s-1}$ for $s = 2, \ldots , r$, which are given by the coefficients $\be_s$.  Then the polynomial $p(x)$ has the form of the statement of the theorem. 
 
 Conversely, given a polynomial $a_nx^n+a_{n-1}x^{n-1} + \cdots + a_1x + a_0$, as in the statement of the theorem, it can be written uniquely in the form: 
 $$
 (\al_{r+1}x^r + x^{r-1} + x^{r-2} + \cdots + x + \al_0) + \sum_{s = 2}^r \be_s x^{s-1}\,.
 $$
 The expression for $\ga$ is deduced from (\ref{a}).  In that expression $x$ and $-k$ are successive increments, so we must have $x/k=-y$ or $-y^{-1}$.  But the expression for $\ga$ is invariant under $y \mapsto y^{-1}$.   
\hfill $\Box$

\medskip
 
We now consider \emph{complex regular cyclic sequences}, that is solutions to (\ref{one}) with $\ga$ constant which in any normalization have at least one non-real value.   We first prove a lemma in a more general context, which we will require in later sections. 

\begin{lemma} \label{lem:ngaleq1}  Suppose the equation:
$$
\frac{\ga}{n} \left( \sum_{\ell = 1}^nz_{\ell}\right)^2 = \sum_{\ell =1}^nz_{\ell}{}^2\,,
$$
is satisfied for $\ga$ real and for $z_{\ell}$ real and not all zero.  Then $\ga \geq 1$.  In particular, if $\phi$ solves {\rm (\ref{one})} at a vertex $x$ of degree $n$, if $n\geq 2$ and $\ga <1$, then in any normalization which has $\phi (x) = 0$, at least one of $\phi (y)$ ($y\sim x$) must be complex.
\end{lemma}

\begin{proof}  The Cauchy-Schwarz inequality shows that for any set $\{ a_1, \ldots , a_n\}$ of complex numbers, one has the inequality
\begin{equation} \label{inductive-inequality}
n\sum_{\ell = 1}^na_{\ell} \ov{a}_{\ell} \geq \left(\sum_{\ell}^na_{\ell}\right) \left(\sum_{\ell}^n\ov{a}_{\ell}\right)
\end{equation}
with equality if and only if $a_1=a_2=\cdots = a_n$.
Then for $z_{\ell}$ real and not all zero satisfying (\ref{one}), we have:
$$
\sum_{\ell}z_{\ell}{}^2 = \frac{\ga}{n} \left( \sum_{\ell}z_{\ell}\right)^2 \leq \ga \sum_{\ell}z_{\ell}{}^2\,.
$$
\end{proof}

\begin{corollary} \label{cor:norm-real-complex}  Let $(x_0, x_1, x_2, \ldots , x_{N-1}, x_N=x_0)$ be a regular cyclic sequence.  Then \emph{either} there is some normalization under which every term is real, \emph{or} whatever normalization is taken, every consecutive triple $(x_{k-1}, x_k, x_{k+1})$ contains at least one complex term.  The real regular cyclic sequences are characterized by the property $\ga \geq 1$.  
\end{corollary}

Let us explore in more detail the import of the above corollary.  Let $(x_0, x_1, x_2, \ldots , x_{N-1}, x_N=x_0)$ be a regular complex cyclic sequence.  Then for each term $x_k$, whatever the normalization, not all of $(x_{k-1}, x_k, x_{k+1})$ are real.  We can normalize so that $x_k=0$ and its two neighbours take values as indicated in the diagram: 
\medskip
\begin{center}
\setlength{\unitlength}{0.254mm}
\begin{picture}(172,87)(15,-96)
        \allinethickness{0.254mm}\path(80,-80)(30,-30) 
        \allinethickness{0.254mm}\path(80,-80)(170,-80) 
        \allinethickness{0.254mm}\special{sh 0.3}\put(80,-80){\ellipse{4}{4}} 
        \allinethickness{0.254mm}\special{sh 0.3}\put(30,-30){\ellipse{4}{4}} 
        \allinethickness{0.254mm}\special{sh 0.3}\put(170,-80){\ellipse{4}{4}} 
        \allinethickness{0.254mm}\dashline{5}[5](80,-80)(15,-80) 
        \put(170,-96){\shortstack{$s$}} 
        \put(80,-96){\shortstack{$0$}} 
        \put(30,-21){\shortstack{$re^{i(\pi-\theta )}$}} 
        \put(50,-76){\shortstack{$\theta$}} 
\end{picture}
\end{center}
\medskip
where we suppose $r,s>0$ and $\theta \in (-\pi , \pi )\setminus \{ 0\}$. 
On substituting into (\ref{one}), we obtain
$$
\ga = 2 + \frac{4rs}{r^2e^{-\ii\ta} + s^2e^{\ii\ta} - 2rs}\,,
$$
 This is real if and only if 
\begin{equation} \label{spec-value}
r=s \quad {\rm in\ which\ case} \quad \ga = \frac{2\cos \theta}{\cos \theta - 1}\,,
\end{equation}
which confirms the inequality $\ga < 1$.  The case when $\theta = \pm \pi /2$ corresponds to $\ga=0$, that is holomorphicity at the vertex $x$.  We underline the fact that the two cases, real triple or complex triple, are mutually exclusive, their nature determined by the value of $\ga$. 

Thus the terms of any complex regular cyclic sequence determine a closed walk in the plane, such that at each step \emph{the walk progresses along an edge of fixed length in such a way that the angle between successive edges is $\pm \theta$, for some fixed $\theta$}; the latter property being a consequence of the constancy of $\cos \theta$ in (\ref{spec-value}).  In particular, any regular polygon satisfies these criteria, with $\theta$ the (constant) exterior angle.  

In principle it is possible to find all solutions:  if $\theta$ is the exterior angle and we perform $k$ rotations through $+\theta$ and $\ell$ through $-\theta$, then $(k-\ell )\theta$ must be a multiple of $2\pi$  with $k+\ell = N$, the order of the graph.  Thus for each $N$, up to normalization and conjugation, there can only be a finite number of possibilities.  In particular, taking into account the real cyclic sequences, the geometric spectrum of a cyclic graph is finite. The regular cyclic sequences on the cyclic graphs of order six and five, respectively, are given as follows.  

\begin{example} \label{ex:hexagon} {\rm 
Consider the cyclic graph on six vertices.  Then we have already encountered a real regular cyclic sequence in Example \ref{ex:real-hexagon}: $(0, 4, 2, 3, 1, 2)$, with corresponding invariant $\ga  =10/9$. Other complex solutions are indicated in the figure below.  For the left-hand hexagon, we have $\ga  = 1$ which corresponds to a $2$-colouring of the graph; for the middle one, $\ga  = 2/3$ which corresponds to two circuits of a triangle; and for the right-hand one, $\ga  = -2$ which corresponds to the position function of a regular hexagon.  

\medskip

\begin{center}
\setlength{\unitlength}{0.254mm}
\begin{picture}(572,132)(35,-136)
        \allinethickness{0.254mm}\path(80,-120)(140,-120) 
        \allinethickness{0.254mm}\path(140,-120)(170,-70) 
        \allinethickness{0.254mm}\path(50,-70)(80,-120) 
        \allinethickness{0.254mm}\path(240,-120)(300,-120) 
        \allinethickness{0.254mm}\path(300,-120)(330,-70) 
        \allinethickness{0.254mm}\path(210,-70)(240,-120) 
        \allinethickness{0.254mm}\path(170,-70)(140,-20) 
        \allinethickness{0.254mm}\path(140,-20)(80,-20) 
        \allinethickness{0.254mm}\path(80,-20)(50,-70) 
        \allinethickness{0.254mm}\path(330,-70)(300,-20) 
        \allinethickness{0.254mm}\path(300,-20)(240,-20) 
        \allinethickness{0.254mm}\path(240,-20)(210,-70) 
        \allinethickness{0.254mm}\path(400,-120)(460,-120) 
        \allinethickness{0.254mm}\path(460,-120)(490,-70) 
        \allinethickness{0.254mm}\path(490,-70)(460,-20) 
        \allinethickness{0.254mm}\path(460,-20)(400,-20) 
        \allinethickness{0.254mm}\path(400,-20)(370,-70) 
        \allinethickness{0.254mm}\path(370,-70)(400,-120) 
        \allinethickness{0.254mm}\special{sh 0.3}\put(140,-120){\ellipse{4}{4}} 
        \allinethickness{0.254mm}\special{sh 0.3}\put(170,-70){\ellipse{4}{4}} 
        \allinethickness{0.254mm}\special{sh 0.3}\put(140,-20){\ellipse{4}{4}} 
        \allinethickness{0.254mm}\special{sh 0.3}\put(80,-20){\ellipse{4}{4}} 
        \allinethickness{0.254mm}\special{sh 0.3}\put(50,-70){\ellipse{4}{4}} 
        \allinethickness{0.254mm}\special{sh 0.3}\put(80,-120){\ellipse{4}{4}} 
        \allinethickness{0.254mm}\special{sh 0.3}\put(240,-120){\ellipse{4}{4}} 
        \allinethickness{0.254mm}\special{sh 0.3}\put(300,-120){\ellipse{4}{4}} 
        \allinethickness{0.254mm}\special{sh 0.3}\put(330,-70){\ellipse{4}{4}} 
        \allinethickness{0.254mm}\special{sh 0.3}\put(300,-20){\ellipse{4}{4}} 
        \allinethickness{0.254mm}\special{sh 0.3}\put(240,-20){\ellipse{4}{4}} 
        \allinethickness{0.254mm}\special{sh 0.3}\put(210,-70){\ellipse{4}{4}} 
        \allinethickness{0.254mm}\special{sh 0.3}\put(400,-120){\ellipse{4}{4}} 
        \allinethickness{0.254mm}\special{sh 0.3}\put(460,-120){\ellipse{4}{4}} 
        \allinethickness{0.254mm}\special{sh 0.3}\put(490,-70){\ellipse{4}{4}} 
        \allinethickness{0.254mm}\special{sh 0.3}\put(460,-20){\ellipse{4}{4}} 
        \allinethickness{0.254mm}\special{sh 0.3}\put(400,-20){\ellipse{4}{4}} 
        \allinethickness{0.254mm}\special{sh 0.3}\put(370,-70){\ellipse{4}{4}} 
        \put(80,-136){\shortstack{$0$}} 
        \put(140,-136){\shortstack{$1$}} 
        \put(175,-81){\shortstack{$0$}} 
        \put(145,-21){\shortstack{$1$}} 
        \put(65,-16){\shortstack{$0$}} 
        \put(35,-76){\shortstack{$1$}} 
        \put(235,-136){\shortstack{$0$}} 
        \put(300,-136){\shortstack{$1$}} 
        \put(500,-76){\shortstack{$\frac{1+\sqrt{3}i}{2}$}} 
        \put(305,-21){\shortstack{$0$}} 
        \put(230,-16){\shortstack{$1$}} 
        \put(395,-136){\shortstack{$0$}} 
        \put(460,-136){\shortstack{$1$}} 
        \put(380,-76){\shortstack{$\frac{-1+\sqrt{3}i}{2}$}} 
        \put(465,-21){\shortstack{$1+\sqrt{3}i$}} 
        \put(390,-16){\shortstack{$\sqrt{3}i$}} 
        \put(215,-76){\shortstack{$\frac{1+\sqrt{3}i}{2}$}} 
        \put(285,-76){\shortstack{$\frac{1+\sqrt{3}i}{2}$}} 
\end{picture}

\end{center}
\medskip

There is the trivial solution with $\ga  = 2$, when the function $\phi$ has value $0$ on any connected set of either three or two vertices of the hexagon and $1$ on the connected complementary set of vertices, so for each vertex $x$, there is at most one neighbouring vertex $y$ with $\phi (y) - \phi (x) \neq 0$ (see Section \ref{sec:colourings}).  This exhausts all possible values of $\ga$, so the \emph{geometric spectrum} (the possible values of $\ga$) equals the set $\{ -2,2/3,1,10/9,2\}$ (cf. Section \ref{sec:spectrum}).  There is another path we have not considered, namely $0\ra 1\ra \frac{1}{2} + \frac{\sqrt{3}}{2}\ii \ra 0 \ra - \frac{1}{2} + \frac{\sqrt{3}}{2}\ii \ra \frac{1}{2} + \frac{\sqrt{3}}{2}\ii \ra 0$, with exterior angles $2\pi/3, 2\pi /3, 2\pi /3, - 2\pi /3, - 2\pi /3, - 2\pi /3$, respectively.  However, this gives the same value $\ga  = 2/3$, corresponding to all exterior angles equal to $2\pi /3$. }
\end{example} 

\begin{example} \label{ex:cyclic-5vertices} {\rm Now consider a cyclic graph on five vertices.  Once more, there is the trivial solution with $\ga =2$, where we colour the graph with two colours on complementary connected components of three and two vertices.  There are also two more solutions as indicated in the figure below.
\medskip
\begin{center}
\setlength{\unitlength}{0.254mm}
\begin{picture}(249,74)(63,-142)
        \allinethickness{0.254mm}\special{sh 0.3}\put(110,-70){\ellipse{4}{4}} 
        \allinethickness{0.254mm}\special{sh 0.3}\put(65,-95){\ellipse{4}{4}} 
        \allinethickness{0.254mm}\special{sh 0.3}\put(155,-95){\ellipse{4}{4}} 
        \allinethickness{0.254mm}\special{sh 0.3}\put(135,-140){\ellipse{4}{4}} 
        \allinethickness{0.254mm}\special{sh 0.3}\put(85,-140){\ellipse{4}{4}} 
        \allinethickness{0.254mm}\special{sh 0.3}\put(220,-95){\ellipse{4}{4}} 
        \allinethickness{0.254mm}\special{sh 0.3}\put(265,-70){\ellipse{4}{4}} 
        \allinethickness{0.254mm}\special{sh 0.3}\put(310,-95){\ellipse{4}{4}} 
        \allinethickness{0.254mm}\special{sh 0.3}\put(240,-140){\ellipse{4}{4}} 
        \allinethickness{0.254mm}\special{sh 0.3}\put(290,-140){\ellipse{4}{4}} 
        \allinethickness{0.254mm}\path(110,-70)(155,-95) 
        \allinethickness{0.254mm}\path(155,-95)(135,-140) 
        \allinethickness{0.254mm}\path(135,-140)(85,-140) 
        \allinethickness{0.254mm}\path(85,-140)(65,-95) 
        \allinethickness{0.254mm}\path(65,-95)(110,-70) 
        \allinethickness{0.254mm}\path(240,-140)(310,-95) 
        \allinethickness{0.254mm}\path(310,-95)(220,-95) 
        \allinethickness{0.254mm}\path(220,-95)(260,-120) 
        \allinethickness{0.254mm}\path(275,-130)(290,-140) 
        \allinethickness{0.254mm}\path(265,-70)(270,-85) 
        \allinethickness{0.254mm}\path(285,-125)(290,-140) 
        \allinethickness{0.254mm}\path(275,-100)(280,-115) 
        \allinethickness{0.254mm}\path(265,-70)(260,-85) 
        \allinethickness{0.254mm}\path(255,-100)(250,-115) 
        \allinethickness{0.254mm}\path(240,-140)(245,-125) 
\end{picture}
\end{center}
\medskip 
In these two examples, the angle $\theta$ is given by $2\pi /5$ for the regular pentagon on the left, and by $4\pi /5$ for the regular star pentagon on the right.  Since $\cos (2\pi /5) = (\sqrt{5}-1)/4$ and $\cos (4\pi /5) = -(\sqrt{5}+1)/4$, by (\ref{spec-value}), this gives the values $\ga = -2/\sqrt{5}$ and $\ga = +2/\sqrt{5}$, respectively.  These two values together with $\ga=2$ as above, exhaust the geometric spectrum for the cyclic graph on five vertices.  }
\end{example}
    
Consider a planar \emph{polygonal chain} on $N$ vertices $\{ x_1, x_2, \ldots , x_N\}$ with edges all of the same length, such that edges can rotate freely about adjacent vertices.  In a more general context, where edges have fixed length that are not necessarily equal, such an object is sometimes referred to either as a \emph{planar linkage} or as a \emph{planar polygonal bar-and-joint framework}; they are studied notably in robot arm motion planning, see for example \cite{Str}.  

A well known problem in the study of planar polygonal chains is to construct algorithms to either straighten the chain if it is open ended (the \emph{Carpenter's Rule Problem}), or, for a closed chain, to deform it into one that is convex.  If on allows edge crossings, then this was solved by Sallee in 1973 \cite{Sa}.  To do this without edge-crossings (for an initial configuration without crossings) proved more elusive and was solved by Connelly, Demaine and Rote in 2003 \cite{Co-De-Ro}. 

In Section \ref{sec:spectrum} we introduce an energy functional in a more general setting, which, for polygonal chains is given by 
$$
\Ee = \sum_{k = 1}^N(1+\cos\ta_k)\,.
$$
When $N$ is even, the absolute minimum of $\Ee$ is zero which is achieved when all exterior angles are $\pm \pi$, so the polygonal chain is in its most compact form, the edges superimposed along an interval.  In Appendix \ref{sec:trig}, we show how the configurations of a polygonal chain on $N$ vertices can be parametrized by $N-3$ parameters and prove that the regular configurations (with exterior angle constant) are critical for this functional.  The functional $\Ee$ should define a gradient flow on this configuration space whereby a polygonal chain evolves into a regular configuration.  

\section{Orthogonal projections of regular polytopes} \label{sec:orthographic}  In \cite{Ea-Pe}, the authors consider the projection of the vertices of the Platonic solids in $\RR^3$ onto the complex plane, as well as more general orthogonal projections $\RR^N\ra \RR^M$.  As a particular case, they establish that if $z_1, z_2, \ldots , z_{N+1}$ are the orthogonal projections to $\CC$ of the vertices of a regular simplex in $\RR^N$, then 
$$
(z_1+\cdots + z_{N+1})^2 = (N+1)(z_1{}^2+ \cdots + z_{N+1}{}^2)\,.
$$
In particular, it follows that if we view the $1$-skeleton of the simplex as a graph, then the function which associates the values $z_1, \ldots , z_{N+1}$ to the corresponding vertices satisfies (\ref{one}) with $\ga = N/(N+1)$.  
We now prove a corresponding result that shows that, as a graph, the projection of the $1$-skeleton of a regular polytope satisfies equation (\ref{one}), with $\ga$ constant.  We also develop some notions that enable us later to define distance and curvature.

A bipartite graph $K_{1,n}$, otherwise known as a star, consists of an internal vertex $\vec{x}_0$ connected by edges to $n$ external vertices $\vec{x}_1, \ldots , \vec{x}_n$; there are no other connections.  In what follows, we will represent a star embedded in $\RR^N$ with internal vertex located at the origin, by an $(N\times n)$--matrix
$$
W = (\vec{x}_1|\cdots |\vec{x}_n)\,,
$$
whose columns are the components of the external vertices. An \emph{invariant} of the star is a quantity which is invariant under orthogonal transformation of the ambient Euclidean space $\RR^N$.  For example the quantity
$$
\frac{1}{n}{\rm trace}\, (WW^t)=\frac{1}{n}{\rm trace}\,(W^tW) = \frac{1}{n}\big( ||\vec{x}_1||^2 + \cdots ||\vec{x}_n||^2\big)\,
$$
is invariant, where $W^t$ denotes the transpose of $W$; it corresponds to the mean of the squares of the Euclidean lengths of the external vertices.  We consider a bipartite graph $K_{1,n}$ embedded in $\RR^N$ in the following symmetric way.  

Let $(y_1, \ldots , y_N)$ be standard coordinates for $\RR^N$ ($N\geq 2$); write vectors as columns for the purpose of matrix multiplication.  Let $\{ \vec{e}_1, \ldots , \vec{e}_N\}$ be the canonical basis and write $I_N$ for the $N\times N$-identity matrix.  The internal vertex $\vec{x}_0$ is located at the origin, while the external vertices $\vec{x}_1, \ldots , \vec{x}_n$ are situated at distinct points in the hyperplane $y_N=c$ (constant): 
\begin{equation} \label{star-standard}
\vec{x}_{\ell} = \left( \begin{array}{c} \vec{v}_{\ell} \\ c \end{array} \right) \qquad (\ell = 1, \ldots , n)\,,
\end{equation}
We require further that the $(N-1)\times n$-matrix $U = (\vec{v}_1|\vec{v}_2|\cdots |\vec{v}_n)$ with columns the components $v_{\ell j}$ of $\vec{v}_{\ell}$ $(j=1, \ldots , N-1;\, \ell = 1, \ldots , n)$, satisfies:
\begin{equation} \label{star-conditions}
 UU^t= \rho I_{N-1}\,, \qquad  \sum_{\ell = 1}^n\vec{v}_{\ell} = \vec{0}\,,
\end{equation}
for some non-zero constant $\rho$ (necessarily positive), where $\vec{0}$ denotes the zero vector in $\RR^{N-1}$ and $U^t$ denotes the transpose of $U$.  

Any star, which, up to orthogonal transformation of $\RR^N$, is embedded in this way, we will call a \emph{configured star}.  We shall also say that the vectors $\{ \vec{v}_1, \ldots , \vec{v}_n\}$ form a \emph{configuration in $\RR^{N-1}$}, call $U$ the associated \emph{configuration matrix} and $\rho$ the \emph{configuration invariant}.  Provided the star does not lie in any proper linear subspace, we say that the star is \emph{full}.  If further, $||\vec{x}_{\ell}||=r$ (constant) for $\ell = 1, \ldots , n$, we refer to the star as \emph{regular of radius $r$}.  An embedding given by (\ref{star-standard}) and (\ref{star-conditions}) is referred to as a \emph{standard position} of the configured star.

\begin{lemma} \label{lem:configured-star}  Consider a configured star in $\RR^N$ ($N\geq 2$) with internal vertex the origin connected to $n$ external vertices $\{ \vec{x}_1, \ldots , \vec{x}_n\}$ $(n\geq N)$.  Let $W = (\vec{x}_1|\vec{x}_2|\cdots |\vec{x}_n)$ be the $N\times n$-matrix whose columns are the components $x_{\ell j}$ of $\vec{x}_{\ell}$ ($j = 1, \ldots , N; \  \ell = 1, \ldots , n$).  Then
\begin{equation} \label{configured-star-conds}
WW^t=\rho I_N + \si \vec{u}\vec{u}^t, \qquad \sum_{\ell = 1}^n\vec{x}_{\ell} = \sqrt{n(\si +\rho )}\,\vec{u}\,,
\end{equation}
where $\vec{u}\in \RR^N$ is a unit vector, $\rho >0$ and $\rho + \si >0$.  The quantities $n, \rho , \si$ are all invariants of the star; the vector $\vec{u}$ is normal to the affine plane containing $\vec{x}_1, \ldots , \vec{x}_n$.

Conversely, any matrix $W = (\vec{x}_1|\vec{x}_2|\cdots |\vec{x}_n)$ satisfying {\rm (\ref{configured-star-conds})} determines a configured star with central vertex the origin and external vertices $\vec{x}_1, \ldots , \vec{x}_n$.
\end{lemma}

\begin{proof}  Consider a configured star in standard position given by (\ref{star-standard}) and (\ref{star-conditions}).  Set
$$
V = \left( \begin{array}{c|c|c|c} \vec{v}_1 & \vec{v}_2 & \cdots & \vec{v}_n \\ c & c & \cdots & c \end{array} \right) 
$$
and let $A : \RR^N \ra \RR^N$ be an orthogonal transformation; set $\vec{x}_n = A\left( \begin{array}{c} \vec{v}_n \\ c\end{array}\right)$.  Then $W = (\vec{x}_1|\vec{x}_2|\cdots |\vec{x}_n)=AV$ and 
$$
WW^t=AVV^tA^t=\rho I_N+\si (A\vec{e}_N)(A\vec{e}_N)^t\,,
$$
where 
\begin{equation} \label{sigma-c}
\si = nc^2-\rho\,.
\end{equation}  
Furthermore $\sum_{\ell =1}^n\vec{x}_{\ell} = ncA\vec{e}_N$, which gives the form (\ref{configured-star-conds}) with $\vec{u} = A\vec{e}_N$.  The independence of the quantities $n, \rho , \si$ of the orthogonal transformation $A$ is clear.

Conversely, suppose we are given an $N\times n$-matrix $W = (\vec{x}_1|\vec{x}_2|\cdots |\vec{x}_n)$ satisfying (\ref{configured-star-conds}).  Let $A$ be an orthogonal transformation such that $A\vec{u} = \vec{e}_N$ and let $V = AW$.  Write 
$$
V = \left( \begin{array}{c|c|c|c} \vec{v}_1 & \vec{v}_2 & \cdots & \vec{v}_n \\ y_{1N} & y_{2N} & \cdots & y_{nN} \end{array} \right)\,. 
$$
Then 
\begin{equation} \label{VVt}
VV^t = \rho I_N+\si \vec{e}_N\vec{e}_N{}^t \qquad {\rm and} \qquad \sum_{\ell} \left( \begin{array}{c} \vec{v}_{\ell} \\ y_{\ell N}\end{array}\right) = \sqrt{n(\si + \rho )}\, \vec{e}_N\,,
\end{equation}
so that $\sum_{\ell}\vec{v}_{\ell} = 0$ and $\sum_{\ell} y_{\ell N} = \sqrt{n(\si + \rho )}$.
Furthermore, (\ref{VVt}) implies that $\sum_{\ell} y_{\ell N}{}^2 = \rho + \si$.  In particular
$$
n\sum_{\ell}^n y_{\ell N}{}^2 = \left(\sum_{\ell}^n y_{\ell N}\right)^2\,.
$$  
But then (\ref{inductive-inequality}) implies that $y_{1N} = y_{2N} = \cdots = y_{nN} = \sqrt{(\si + \rho )/n}$.  
\end{proof}
 
It now follows that the function which assigns the values after projection of the vertices of a configured star to the complex plane satisfies (\ref{one}) at the internal vertex, independently of the position of the star.

\begin{corollary} \label{cor:configured-star}  Let $W = (\vec{x}_1|\vec{x}_2|\cdots |\vec{x}_n)$ define a configured star and let $P:\RR^N \ra \CC$ be orthgonal projection $P(y_1, \ldots , y_N)=y_1 + \ii y_N$.  Then if $z_{\ell} = P(\vec{x}_{\ell}) = x_{\ell 1} + \ii x_{\ell 2}$, we have
\begin{equation} \label{star-gamma}
\frac{\si}{n(\si + \rho )} \left( \sum_{\ell = 1}^nz_{\ell}\right)^2 = \sum_{\ell =1}^n z_{\ell}{}^2\,,
\end{equation}
where $\rho$ and $\si$ are given by {\rm (\ref{configured-star-conds})}.  In particular, $\ga = \si /(\si + \rho )$ is real and depends only on the star invariants.
\end{corollary}
\begin{proof}  Let $\vec{u} = (u_1, \ldots , u_N)$ be the unit normal to the plane of the star.  Then for each $j = 1, \ldots , N$, we have
$$
\sum_{\ell = 1}^n x_{\ell j} = \sqrt{n(\si + \rho )} \, u_j\,.
$$
Thus
\begin{eqnarray*}
\left( \sum_{\ell = 1}^nz_{\ell}\right)^2 & = & \sum_{k,\ell = 1}^n (x_{k1} x_{\ell 1} - x_{k2} x_{\ell 2} + 2\ii x_{k1}x_{\ell 2}) \\
 & = & n(\si \rho )(u_1{}^2-u_2{}^2 + 2\ii u_iu_2) \\
  & = & n(\si \rho )(u_1+\ii u_2)^2\,,
  \end{eqnarray*}
whereas
\begin{eqnarray*}
\sum_{\ell =1}^n z_{\ell}{}^2 & = & \sum_{\ell = 1}^n (x_{\ell 1}{}^2 - x_{\ell 2}{}^2 + 2\ii x_{\ell 1}x_{\ell 2}) \\
 & = & (WW^t)_{11} - (WW^t)_{22} + 2\ii (WW^t)_{12} \\
  & = & \si (u_1+ \ii u_2)^2\,.
  \end{eqnarray*}
  The formula now follows.
  \end{proof}
  
  In Section \ref{sec:distance} we will consider the problem of establishing a converse to this corollary. 

Examples of configurations of points $\vec{v}_1, \ldots , \vec{v}_n \in \RR^{N-1}$ which satisfy the criteria of (\ref{star-conditions}) are as follows.  In $\RR$, any set of points not all zero distributed along the real line with centre of mass the origin form such a configuration.  In $\RR^2 \simeq \CC$, we have the $n$ roots of unity:
\begin{equation} \label{star-R2}
\vec{v}_{\ell} = e^{2\pi \ii\ell /n} \qquad (\ell = 1, \ldots n)\,.
\end{equation}
For $n\geq 3$, we have $\rho = \sum_{\ell = 1}^n\cos^2(2\pi \ell /n) = \sum_{\ell = 1}^n\sin^2(2\pi \ell /n) = n/2$ and the coefficient $\ga$ determined by (\ref{star-gamma}) is given by
\begin{equation} \label{star-gamma-cyclic}
\ga = \frac{2c^2-1}{2c^2}\,.
\end{equation}
This configuration is regular.  We note the following elementary fact.
\begin{lemma}  In \ $\RR^2$, every configuration of three points is regular.
\end{lemma}

\begin{proof}  Consider a configuration of three points in $\RR^2$ given by (\ref{star-conditions}).  Write the entries of $U$ as follows:
$$
U = \left( \begin{array}{ccc} u_1 & u_2 & u_3 \\ v_1 & v_2 & v_3 \end{array} \right)\,.
$$
By a rotation, we can suppose the further normalization: $u_1 = 1,\ v_1 = 0$.  This gives the equations;
$$
\begin{array}{rcl}
1+u_2{}^2 + u_3{}^2 & = & v_2{}^2+v_3{}^2 \\
u_2v_2+u_3v_3 & = & 0 \end{array}
$$
Then this has the solution, unique up to the sign of the square root,
$$
U = \left( \begin{array}{ccc} 1 & -\frac{1}{2} & -\frac{1}{2} \\ 0 & \frac{\sqrt{3}}{2} & - \frac{\sqrt{3}}{2} \end{array} \right)\,.
$$
Since the lengths of the columns of this matrix are all equal to $1$, then the configuration is regular.
\end{proof}

However, there exist non-regular configurations in $\RR^2$.  Consider the following set of four points, expressed as the columns of the corresponding configuration matrix:
\begin{equation} \label{non-reg-config}
\left( \begin{array}{rrrr} -1 & 1 & -1 & 1 \\ 
\la & \mu & - \la & - \mu \end{array} \right) ,
\end{equation}
where $\la$ and $\mu$ are real non-zero constants.  Then the set forms a configuration if and only if $\la^2+\mu^2=2$.  This is regular only when $\la^2=\mu^2$.  A non-regular star determined by this configuration will appear later in Example \ref{ex:lift}.  In the examples of configurations that follow, we express their vertices as the column vectors of the corresponding configuration matrix.   

In $\RR^3$, the vertices of a tetrahedron placed symmetrically at the points:
\begin{equation} \label{tetrahedron-coords}
\left( \begin{array}{rrrr} -1 & -1 & 1 & 1 \\
-1 & 1 & -1 & 1 \\
-1 & 1 & 1 & -1 
 \end{array}\right)\,,
\end{equation}
satisfy conditions (\ref{star-conditions}), as do the six vertices of an octahedron:
$$
\left( \begin{array}{rrr} \pm 1 & 0 & 0\\
0 & \pm 1 & 0 \\
0 & 0 & \pm 1  \end{array}\right)\,,
$$
and the twelve vertices:
\begin{equation} \label{vertices-icosahedron}
\left( \begin{array}{rrr} 0 & \pm 1 & \pm \la \\
\pm 1 & \pm \la & 0 \\
\pm \la & 0 & \pm 1 \end{array}\right)\,,
\end{equation}
where $\la$ is any real constant.  In the case when $\la = \frac{1+\sqrt{5}}{2}$, these form the vertices of an icosahedron.  In $\RR^3$, there is a further configuration of twenty vertices satisfying (\ref{star-conditions}) given by:
\begin{equation} \label{vertices-dodecahedron}
\left( \begin{array}{cccc} 0 & \pm \la & \pm \la^{-1} & \pm 1 \\
\pm \la^{-1} & 0 & \pm \la & \pm 1 \\
\pm \la & \pm \la^{-1} & 0 & \pm 1  \end{array}\right)\,,
\end{equation}
where $\la$ is any real constant.  In the case when $\la = \frac{1+\sqrt{5}}{2}$, these form the vertices of a dodecahedron (see \cite{Co}, \S 3.8). 

In $\RR^{N-1}$, the configuration of $2(N-1)$ vertices of the cross-polytope:
\begin{equation} \label{cross-poly}
\left( \begin{array}{rrrr} \pm 1 & 0 & \cdots & 0 \\
0 & \pm 1 & \cdots & 0 \\
\vdots & \vdots & \ddots & \vdots \\
0 & 0 & \cdots & \pm 1 \end{array}\right), 
\end{equation}
satisfies (\ref{star-conditions}). 

In order to construct a configuration of vertices of a regular $N$-simplex in $\RR^N$ ($N\geq 3$) satisfying (\ref{star-conditions}), we proceed inductively starting with the configuration (\ref{tetrahedron-coords}) in $\RR^3$.  

Suppose we are given $N$ vectors $\vec{v}_1, \ldots , \vec{v}_N$ in $\RR^{N-1}$ which label the vertices of a regular $(N-1)$-simplex in such a way that, on letting $U$ denote the $(N-1)\times N$-matrix whose columns are formed from the components of $\vec{v}_{\ell}$ $(\ell = 1, \ldots , N)$, the following conditions are satisfied:
$$
\sum_{\ell = 1}^N \vec{v}_{\ell} = \vec{0}, \quad UU^t=\rho I_{N-1},\quad ||\vec{v}_{\ell}||=\si \ \forall \ell\,,
$$
 for constants $\rho$ and $\si$ satisfying $\rho /\si^2=N/(N-1)$ (which is the case for the $3$-simplex of (\ref{tetrahedron-coords})).  Note that the condition that the vectors do indeed correspond to the vertices of a regular $(N-1)$-simplex is that they are of equal length and the angle between them has cosine $- 1/(N-1)$.  Form the following set of $N+1$ vectors in $\RR^N$:
\begin{equation} \label{N-simplex-induction}
 \vec{w}_0=\left( \begin{array}{c} -\frac{N\si}{\sqrt{N^2-1}} \\ \vec{0}  \end{array}\right),\,
\vec{w}_1=  \left( \begin{array}{c} \frac{\si}{\sqrt{N^2-1}} \\ \vec{v}_1  \end{array}\right),\, \cdots ,\, 
\vec{w}_N=      \left( \begin{array}{c} \frac{\si}{\sqrt{N^2-1}} \\ \vec{v}_N  \end{array}\right).
\end{equation}
and let $W$ be the matrix whose columns are the components of $\vec{w}_{\ell}$ for $\ell = 0,1, \ldots ,N$.
Then it is readily checked that 
$$
\sum_{\ell = 0}^N \vec{w}_{\ell} = 0, \quad WW^t=\rho I_N,\quad ||\vec{w}_{\ell}||= \frac{\si N}{\sqrt{N^2-1}} \ \forall \ell\,,
$$
now with $\rho /\left(\frac{\si N}{\sqrt{N^2-1}}\right)^2 = (N+1)/N$, as required.  Furthermore, the angle between any two $\vec{w}_j$ and $\vec{w}_k$ ($j\neq k)$ has cosine $-1/N$.  In particular, the vectors correspond to the vertices of a regular $N$-simplex and satisfy (\ref{star-conditions}).    

To conclude this list of examples of configurations, we note that if $(\vec{v}_1|\vec{v}_2|\cdots |\vec{v}_n)$ is a configuration in $\RR^{N-1}$ with invariant $\rho$, then both
$$
\left( \begin{array}{c|c|c|c|c} 
\vec{v}_1 & \vec{v}_2 & \cdots & \vec{v}_n & \vec{0}  \\
 c & c & \cdots & c & -nc 
 \end{array} \right)
 $$
 where $c = \rho /\sqrt{n(1+n)}$, 
 and
$$
\left( \begin{array}{c|c|c|c|c|c} 
\vec{v}_1 & \vec{v}_2 & \cdots & \vec{v}_n & \vec{0} & \vec{0} \\
 0 & 0 & \cdots & 0 & \sqrt{\frac{\rho}{2}} & -\sqrt{\frac{\rho}{2}}
 \end{array} \right)
 $$
are configurations in $\RR^N$ with the same invariant $\rho$.  The first of these generalizes the inductive construction of the vertex configuration of the regular simplex given by (\ref{N-simplex-induction}). 

Before proceeding, we assemble the information we require concerning regular polytopes, which can be found in the classical text of Coxeter \cite{Co}.  The regular polytopes have a symbolic representation in terms of the Schl\"afli symbol.  This can be defined inductively as follows.  For a regular polygon with $p$ edges, one assigns the symbol $\{ p\}$.  A regular star polygon which winds $m$ times around its center is denoted by the fractional value $\{ p/m\}$, so for example, the second graph of Example \ref{ex:cyclic-5vertices} has Schl\"afli symbol $\{ 5/2\}$.  A regular polyhedron whose faces are polygons of type $\{ p\}$ which has $q$ such faces joining around a vertex has symbol $\{ p,q\}$.  A regular $4$-polytope (or polychoron) with highest dimensional cells polyhedra of type $\{ p,q\}$ having $r$ such cells joining around an edge has symbol $\{ p,q,r\}$, and so on.  

A regular polytope in $\RR^N$ with Schl\"afli symbol $\{ p, q, \ldots , s,t\}$ is characterized by its highest dimensional cells which are regular polytopes with symbol $\{ p,q, \ldots , s\}$, and its \emph{vertex figure}, which is a regular polytope of type $\{ q,r \ldots , t\}$ in $\RR^{N-1}$ obtained by fixing a vertex $\vec{x}_0$ and constructing a polytope in an affine $(N-1)$-plane, whose vertices are points half way along each edge emanating from $\vec{x}_0$.  For example, the dodecahedron has Schl\"afli symbol $\{ 5,3\}$; it is made up of three pentagonal faces around each vertex; its vertex figure is the triangle obtained by moving half way along each edge emanating from a particular vertex and joining these points by edges which traverse each of the three corresponding faces.  The $600$ cell is a $4$-dimensional regular polytope with $120$ vertices, $720$ edges, $1200$ faces and $600$ tetrahedral cells; it has $5$ tetrahedra joining around each of its edges and so has Schl\"afli symbol $\{ 3,3,5\}$.  Its vertex figure is an icosahedron with symbol $\{ 3,5\}$.  

In two dimensions, the regular polytopes are just the regular polygons and star polygons, as discussed in Section \ref{sec:cyclic} (see Example \ref{ex:cyclic-5vertices} for an illustration of a star polygon on five vertices).  A rich variety of regular polytopes exists in three and four dimensions.  In three dimensions, there are five convex regular polytopes (polyhedra): the tetrahedron, octahedron, cube, icosahedron and dodecahedron, often known as the Platonic solids; in addition there are four non-convex regular polytopes called star polyhedra, with Schl\"afli symbols $\{\frac{5}{2},5\}$, $\{ 5, \frac{5}{2}\}$, $\{ \frac{5}{2}, 3\}$, $\{ 3, \frac{5}{2}\}$.  

The star polyhedra can be constructed from the convex polyhedra by a process known as stellating and faceting. For a polygon, stellating consists of maintaining the edges and extending them until they connect in new vertices; faceting on the other hands consists of maintaining the vertices and inserting new edges in an appropriate way.  The procedure for polyhedra is similar: either the faces are extended to create new vertices or the vertices are maintained and new faces are constructed (see \cite{Co}, Chapter 6).  In particular, the ensemble of positions of the vertices of a star polyhedron is always congruent to those of a convex polyhedron.  

In dimension four, there are sixteen regular polytopes, six of them convex.  An important property to note is that it is precisely the three dimensional regular polyhedra that arise as vertex figures of the four dimensional regular polytopes.  

In dimensions five and above, there are just three kinds of regular polytope, the regular simplex with Schl\"afli symbol $\{ 3,3, \ldots , 3\}$, the cross polytope with Schl\"afli symbol $\{ 3, 3, \ldots , 3, 4\}$ and the measure polytope, or hypercube with Schl\"afli symbol $\{ 4,3,\ldots , 3\}$. 

The above discussion leads to the following consequence.

\begin{theorem}  \label{thm:reg-polytope}  Let $\Ga = (V,E)$ be the graph given by the $1$-skeleton of a regular polytope and let $P:\RR^N \ra \CC$ be any orthogonal projection of the ambient Euclidean space.  Let $\phi = P|_V:V\ra \CC$.  Then $\phi$ satisfies {\rm (\ref{one})} with $\ga$ constant.
\end{theorem}

\begin{proof}  It suffices to note, that in each case the vertex figures with respect to a particular vertex $\vec{x}_0$ can be positioned in such a way that the vectors of the corresponding star satisfy the criteria of Corollary \ref{cor:configured-star}.   In the case of a polyhedron, the vertex figure is a polygon whose vertices can be placed as in (\ref{star-R2}).  For a $4$-dimensional regular polytope, the vertex figure is a regular polyhedron and, as the case by case list above shows, their vertices can also be placed in the required way.  Similarly, for the cross-polytope, whose vertex figure is another cross-polytope of one dimension lower; this is dealt with by the vertex placement (\ref{cross-poly}).  The hypercube and the regular $N$-simplex both have vertex figure a regular $N-1$-simplex, whose required placement is given by the inductive construction of the vectors given in (\ref{N-simplex-induction}) (with $N$ replaced by $N-1$).  
\end{proof} 

\section{Invariant frameworks} \label{sec:inv-polytopes} 

A \emph{framework} $\Ff$ in $\RR^N$ is a finite collection of points $\{ \vec{x}_1, \ldots , \vec{x}_n\}$ which correspond to the vertices of a graph whose edges are straight line segments joining the vertices.  We shall view a framework as a graph and continue to use the notation $\Ff = (V,E)$ to distinguish the vertices and edges.  The edges are often called \emph{bars} and such bar frameworks have been much studied in respect of questions about rigidity \cite{Co-Jo-Wh}.  Thus one is interested in whether it is possible to deform the bar framework whilst preserving the length of the edges.  If the only such deformations correspond to rigid motions, then the bar framework is called \emph{rigid}.  The planar linkages of Section \ref{sec:cyclic} (other than the triangle) provide examples of non-rigid bar frameworks. 

We have seen in the last section that the framework corresponding to the $1$-skeleton of a regular polytope has the property that the function which, after an orthogonal projection $\RR^N\ra \CC$ associates the corresponding complex values to each vertex, satisfies (\ref{one}).  Furthermore, it does so in an invariant fashion with respect to orthogonal transformation of the ambient Euclidean space.  In this section, we wish to determine more general frameworks for which these properties hold.  Amongst the examples are ones that are ``flexible" in the sense that they can be deformed continuously through solutions to (\ref{one}) without disturbing the value of $\ga$. 

Let $\Ff = (V,E)$ be a framework in $\RR^N$ and let $P:\RR^N \ra \CC$ be some orthogonal projection of the ambient Euclidean space.  Let $\phi = P|_V:V\ra \CC$.  We say that the framework is \emph{invariant} if (i) $\phi$ satisfies (\ref{one}) for some $\ga : V\ra \RR$; (ii) if $A:\RR^N\ra \RR^N$ is any orthogonal transformation, then $\phi_A = P\circ A|_V:V\ra \CC$ also satisfies (\ref{one}) with corresponding $\ga_A$ satisfying $\ga_A (A\vec{x}) = \ga (\vec{x})$ for all $\vec{x}\in V$.  Note that if (i) is satisfied, this always remains the case if the framework is translated or dilated, since this just corresponds to a normalization $\phi \mapsto \la \phi + \mu$ (now with $\la$ real).  Then condition (ii) implies invariance by a transformation $\vec{x} \mapsto \la A\vec{x} + \vec{b}$, where $\la >0$ is real, $A$ is any orthogonal transformation and $\vec{b}\in \RR^N$ is any vector.

By Theorem \ref{thm:reg-polytope}, any framework given by the $1$-skeleton of a regular polytope is invariant.  But do there exist other invariant frameworks that arise as the $1$-skeleton of non-regular polytopes? It turns out that the answer to this question is yes.  In Section \ref{sec:particles}, when we consider connected graphs which satisfy (\ref{one}) as components of a more complex system, we will view such objects as ``physically" significant.  We now construct families of examples.  In order to do this, we generalize the configured stars of the last section.  

Consider a star in $\RR^N$ with internal vertex the origin and external vertices located at the points $\vec{x}_1, \vec{x}_2, \ldots , \vec{x}_n$.  As in the last section, it is often convenient to represent the star by an $N\times n$ --matrix whose columns are the components of the vectors $\vec{x}_{\ell}$ $(\ell = 1, \ldots , n)$:
$$
W = (\vec{x}_1|\vec{x}_2|\cdots |\vec{x}_n)\,.
$$
Consider the projection $P:\RR^N\ra \CC$ given by $P(y_1, \ldots , y_N) = y_1 + \ii y_2$.  Let $A:\RR^N\ra \RR^N$ be an orthogonal transformation and set $z_{\ell} = P\circ A (\vec{x}_{\ell})$.  We say that the star is \emph{invariant} if it satisfies the equation:
\begin{equation} \label{inv-star}
\frac{\ga}{n} \left( \sum_{\ell = 1}^nz_{\ell}\right)^2 = \sum_{\ell = 1}^nz_{\ell}{}^2\,,
\end{equation}
with $\ga$ real and independent of the transformation $A$.  By Corollary \ref{cor:configured-star}, a configured star is invariant. However, there are more general invariant stars.  One family is given by the following lemma.

\begin{lemma} \label{lem:invariant-star}
The star in $\RR^3$ with $2r$ external vertices represented by the columns of the $3\times (2r)$ --matrix
$$
W = \left( \begin{array}{rrrrrrrr} x_1 & x_2 & \cdots & x_r & x_1 & x_2 & \cdots & x_r \\
s_1 & s_2 & \cdots & s_r & -s_1 & - s_2 & \cdots & - s_r \\
t_1 & t_2 & \cdots & t_r & - t_1 & - t_2 & \cdots & - t_r 
\end{array} \right)\,,
$$
 where the vectors $\vec{s} = (s_1, \ldots , s_r)$ and $\vec{t} = (t_1,\ldots , t_r)$ are orthogonal and of the same length, is invariant.
 \end{lemma}
 
We omit the proof, which runs along similar lines to that of the theorem which follows.

The lemma enables us to construct a family of invariant polytopes in $\RR^3$, as illustrated in the following figure.

\medskip
\begin{center}
\setlength{\unitlength}{0.254mm}
\begin{picture}(135,150)(110,-225)
        \allinethickness{0.254mm}\path(145,-180)(205,-180) 
        \allinethickness{0.254mm}\path(205,-180)(245,-150) 
        \allinethickness{0.254mm}\path(145,-180)(110,-150) 
        \allinethickness{0.254mm}\path(110,-150)(140,-130) 
        \allinethickness{0.254mm}\path(175,-75)(110,-150) 
        \allinethickness{0.254mm}\path(175,-75)(145,-180) 
        \allinethickness{0.254mm}\path(175,-75)(205,-180) 
        \allinethickness{0.254mm}\path(175,-75)(245,-150) 
        \allinethickness{0.254mm}\path(245,-150)(215,-130) 
        \allinethickness{0.254mm}\path(175,-225)(110,-150) 
        \allinethickness{0.254mm}\path(145,-180)(175,-225) 
        \allinethickness{0.254mm}\path(175,-225)(205,-180) 
        \allinethickness{0.254mm}\path(175,-225)(245,-150) 
\end{picture}
\end{center}
\medskip

We take a regular polygon on $n$ vertices in the $y_1y_2$--plane with centre the origin and join each of these vertices to two symmetrically placed points $(0,0, \pm b)$ along the $y_3$ axis, where $b$ is a constant to be determined.  In the first instance, suppose the vertices of the polygon are located at the points $e^{2k\pi \ii /n}$ for $k = 0,1, \ldots , n-1$.  Fix attention on the star centred on the vertex at $(1,0,0)\in \RR^3$.  After translating to the origin, the matrix of the star is given by
$$
\left( \begin{array}{rcrc} -1 & -1 + \cos \frac{2\pi}{n} & -1 & -1 + \cos \frac{2\pi}{n} \\
 0 & \sin \frac{2\pi}{n} & 0 & - \sin \frac{2\pi}{n} \\
 b & 0 & -b & 0 
 \end{array} \right)\,.
 $$
 We then see that the uniquely determined choice $b = \sin \frac{2\pi}{n}$ (up to sign) guarantees that this matrix has the form of Lemma \ref{lem:invariant-star}.  We may calculate $\ga_{\rm lat}$ at this \emph{lateral} vertex, to give:
\begin{equation} \label{ga-lat}
 \ga_{\rm lat} = \frac{2(1-2\cos\frac{2\pi}{n} + 2\cos^2\frac{2\pi}{n})}{(2-\cos\frac{2\pi}{n})^2}\,.
\end{equation}
 On the other hand, the stars at the apexes, are regular configured stars which therefore satisfy (\ref{inv-star}), with $\ga_{\rm apex}$ given by (\ref{star-gamma-cyclic}):
\begin{equation} \label{ga-apex}
 \ga_{\rm apex} = \frac{2\sin^2\frac{2\pi}{n} -1}{2\sin^2\frac{2\pi}{n}}\,.
\end{equation}
 The two values of $\ga$ coincide precisely when $n=4$, in which case $\ga = 1/2$; then the figure is regular and corresponds to the cross-polytope.  It is readily checked that any other choice of $b$ destroys invariance.  We now generalize this construction.
 
 \medskip
 
Consider a framework $\Ga$ in $\RR^N$ given by the $1$-skeleton of a regular polytope.  We suppose the polytope is centred on the origin in $\RR^N$.  Define the \emph{double cone on $\Ga$ of height} $b$ to be the new framework in $\RR^{N+1}$ obtained by adding two vertices symmetrically placed at the points $(0, 0, \ldots , 0, \pm b)\in \RR^{N+1}$ and adding edges joining each of these points to each vertex of the polytope.  The example discussed above represents a double cone on a regular polygon. (Single) cones on frameworks have been considered in respect of rigidity problems by Connelly and Whitely \cite{Co-Wh}.

\begin{theorem} \label{thm:double-cone}  Given the framework $\Ga$ corresponding to the $1$-skeleton of a regular polytope, then there is a unique double cone on $\Ga$ which is invariant.
\end{theorem} 

In particular, if the regular polytope is convex, then the double cone corresponds to the vertex and edge set of a new convex polytope of one dimension greater.  In general this is only regular if the original polytope is the cross-polytope.  Indeed, the degree is no longer constant in general.  For example, the double cone on the dodecahedron has two vertices of degree $20$ and twenty vertices of degree $5$.  These polytopes generalize the regular polytopes, as being invariant by rigid transformation in $\RR^{N+1}$ in respect of equation (\ref{one}).  

\medskip

\noindent \emph{Proof of theorem} :  Consider a regular polytope in $\RR^N$.  Fix attention on one particular vertex $\vec{x}$ which we suppose located at the origin in $\RR^N$.  We may suppose the co\"ordinates chosen so that the vertex figure at $\vec{x}$ is given by the configuration $U = (\vec{v}_1|\vec{v}_2|\cdots |\vec{v}_n)$ in $\RR^{N-1}$, where $\RR^{N-1}$ is embedded in $\RR^N$ as the first $N-1$ co\"ordinates.  In particular, the star on the vertex $x$ is given by the matrix:
$$
W = \left( \begin{array}{c|c|c|c} \vec{v}_1 & \vec{v}_2 & \cdots & \vec{v}_n \\ c & c & \cdots & c 
\end{array} \right)\,, 
$$
for some non-zero constant $c$.  

Now the line from the origin to the centre of the configuration passes through the centre of the polytope, which is at some point $\left( \begin{array}{c} \vec{0} \\ a \end{array} \right)$, where $\vec{0} \in \RR^{N-1}$ is the centre of the configuration given by $U$.  
\medskip
\begin{center}
\setlength{\unitlength}{0.254mm}
\begin{picture}(183,264)(135,-297)
        \allinethickness{0.254mm}\path(155,-200)(280,-35) 
        \allinethickness{0.254mm}\path(155,-200)(235,-165) 
        \allinethickness{0.254mm}\path(155,-200)(260,-220) 
        \allinethickness{0.254mm}\dottedline{5}(155,-200)(280,-185) 
        \allinethickness{0.254mm}\dottedline{5}(280,-35)(280,-295) 
        \allinethickness{0.254mm}\path(155,-200)(280,-295) 
        \allinethickness{0.254mm}\dottedline{5}(235,-165)(260,-220) 
        \allinethickness{0.254mm}\special{sh 0.3}\put(155,-200){\ellipse{4}{4}} 
        \allinethickness{0.254mm}\special{sh 0.3}\put(280,-35){\ellipse{4}{4}} 
        \allinethickness{0.254mm}\special{sh 0.3}\put(280,-295){\ellipse{4}{4}} 
        \allinethickness{0.254mm}\special{sh 0.3}\put(260,-220){\ellipse{4}{4}} 
        \allinethickness{0.254mm}\special{sh 0.3}\put(235,-165){\ellipse{4}{4}} 
        \allinethickness{0.254mm}\path(290,-185)(290,-35)\special{sh 1}\path(290,-35)(288,-41)(290,-41)(292,-41)(290,-35) 
        \put(300,-121){\shortstack{$b$}} 
        \put(230,-155){\shortstack{$\vec{v}_1$}} 
        \put(245,-236){\shortstack{$\vec{v}_2$}} 
        \put(135,-196){\shortstack{$\vec{0}$}} 
\end{picture}
\end{center}
\medskip
In order to construct the double cone, we situate the polytope in the plane $y_{N+1} = 0$ in $\RR^{N+1}$ and add two more vertices at the points $\left( \begin{array}{r} \vec{0} \\ a \\ \pm b \end{array} \right)$.  These two vertices are to be added as new external vertices to the star centred on $\vec{0} \in \RR^{N+1}$, whereby the star matrix now becomes:
\begin{equation} \label{star-double-cone}
S= \left( \begin{array}{c|c|c|c|c|r} \vec{v}_1 & \vec{v}_2 & \cdots & \vec{v}_n & \vec{0} & \vec{0} \\
 c & c & \cdots & c & a & a \\ 0 & 0 & \cdots & 0 & b & -b \end{array} \right)
\end{equation}
This is not in general a configured star.  We wish to confirm its invariance under orthogonal transformation with respect to (\ref{inv-star}) (with $z_1,\ldots , z_n$ replaced by $z_1, \ldots , z_n, z_{n+1}, z_{n+2}$ as defined below). 

Let $A:\RR^{N+1} \ra \RR^{N+1}$ be an orthogonal transformation and as usual set $A = (a_{rs})$, where now $r,s\in \{ 1,2, \ldots , N+1\}$.  Let $j,k$ range over the indices $\{ 1, \ldots , N-1\}$, so that the components of $\vec{v}_{\ell}$ are given by $v_{\ell j}$.  We compute the first two rows of $AS$ to determined the complex numbers $z_{\ell}, z_{n+1}, z_{n+2}$, where $1\leq \ell \leq n$, after projection $P:\RR^{N+1} \ra \CC$ given by $P(y_1, \ldots , y_{N+1}) = y_1 + \ii y_2$:
\begin{eqnarray*}
z_{\ell} & = & c(a_{1N}+\ii a_{2N}) + \sum_{j = 1}^{N-1} (a_{1j} + \ii a_{2j})v_{\ell j} \\
z_{n+1} & = & a(a_{1N} + \ii a_{2N}) + b(a_{1, N+1} + \ii a_{2,N+1}) \\
z_{n+2} & = & a(a_{1N} + \ii a_{2N}) - b(a_{1, N+1} + \ii a_{2,N+1})
\end{eqnarray*}
On recalling that for each $j = 1, \ldots , N-1$, the sum $\sum_{\ell} v_{\ell j} = 0$, we have
\begin{equation} \label{sum-zl}
z_{n+1} + z_{n+2} + \sum_{\ell = 1}^n z_{\ell} = (nc+2a)(a_{1N} + \ii a_{2N})\,.
\end{equation}
Now
\begin{eqnarray*}
z_{\ell}{}^2 & = & c^2(a_{1N}+\ii a_{2N})^2 + 2c(a_{1N} + \ii a_{2N})\sum_{j=1}^{N-1}(a_{1i} + \ii a_{2j})v_{\ell j} \\
& & \qquad + \sum_{j=1}^{N-1} (a_{1j} + \ii a_{2j})^2v_{\ell j}{}^2 + 2\sum_{j<k}(a_{1j} + \ii a_{2j})(a_{1k} + \ii a_{2k})v_{\ell j}v_{\ell k}\,.
\end{eqnarray*}
But since $U = (\vec{v}_1|\vec{v}_2|\cdots |\vec{v}_n)$ is a configuration, from (\ref{star-conditions}), we have
$$
\sum_{\ell = 1}^nv_{\ell j}v_{\ell k} = \rho \delta_{jk}\,,
$$
for all $j,k = 1, \ldots , N-1$.  It now follows that
$$
z_{n+1}{}^2 + z_{n+2}{}^2 + \sum_{\ell = 1}^nz_{\ell}{}^2 = (2a^2+nc^2-\rho )(a_{1N}+ \ii a_{2N})^2 + (2b^2-\rho )(a_{1,N+1}+ \ii a_{2,N+1})^2\,.
$$
On comparing with (\ref{sum-zl}), we see that (\ref{inv-star}) is satisfied with no dependence on $A$ if and only if $b = \sqrt{\rho /2}$.  In this case
$$
\ga_{\rm lat} = \frac{(n+2)(2a^2+nc^2-\rho)}{(nc+2a)^2}\,,
$$
where $\ga_{\rm lat}$ refers to the \emph{lateral} value of $\ga$, namely the value at one of the vertices of the regular polytope.

For the two vertices corresponding to the apexes of the double cone, the invariance if given by Corollary \ref{cor:configured-star}.  Indeed, as observed in the last section, the vertices of any regular polytope $\Pp$, convex or not, form a configuration.  If we let $\rho_{\Pp}$ denote the invariant of this configuration and let $m$ denote the its cardinality, then from (\ref{sigma-c}), we have
$$
\si = mb^2 - \rho_{\Pp} = \frac{m\rho}{2} - \rho_{\Pp}\,.
$$ 
It follows that
\begin{equation} \label{double-cone-ga-apex}
\ga_{\rm apex} = \frac{\si}{\si + \rho_{\Pp}} = \frac{m\rho - 2\rho_{\Pp}}{m\rho}\,.
\end{equation}
This completes the proof of the theorem.
\hfill $\Box$
\medskip

We have computed the two values of $\ga$ in the above proof so as to be able to apply them to examples.
Note that we are at liberty to normalize the parameters $a,c,\rho$ as we wish.  For example, for the double cone on a regular polygon of $n$ sides, $c = 2a\sin^2\frac{\pi}{n}$ and $\rho = 2a^2\sin^2\frac{2\pi}{n}$.  This confirms the values of $\ga$ computed in the example preceeding the theorem.

If our initial polytope is the cross-polytope in $\RR^N$, then we have $m = 2N$ and $n = 2(N-1)$.  If we normalize so that $a=1$, say, then $c=1$ and $\rho = \rho_{\Pp} = 2$.  It follows that the unique value of $b$ which gives an invariant double cone is $b = \sqrt{\rho /2} = 1$, as required, since the double cone on the cross-polytope is another cross-polytope of one dimension higher.  One easily checks that $\ga_{\rm apex} = \ga_{\rm lat}$ in this case.

Before moving on, we establish one more instance of invariant frameworks which has implications for questions of ``rigidity".

\begin{proposition} \label{prop:double-cone-on-complete}  Let $(\vec{v}_1|\vec{v}_2|\cdots |\vec{v}_n)$ be a configuration in $\RR^N$.  Consider the framework $\Ff$ given by the complete graph on this configuration.  Then there is a unique double cone on $\Ff$ which is invariant.
\end{proposition}

\begin{proof}  We use the same notation as in the proof of the above theorem, except that now each vector $\vec{v}_{\ell}$ lies in $\RR^N$ rather than in $\RR^{N-1}$, so the indices $j,k$ range over $1, \ldots , N$.  The invariance at the apexes is guaranteed by Corollary \ref{cor:configured-star}.  We therefore fix attention on one of the lateral vertices, say $\vec{v}_1$.  Perform a translation so that $\vec{v}_1$ is located at the origin.  If we suppose the height of the double cone is given by $b$, then the star matrix at $\vec{v}_1$ is given by
$$
S:= \left( \begin{array}{ccccrr} v_{21}-v_{11} & v_{31} - v_{11} & \cdots & v_{n1}-v_{11} & -v_{11} & - v_{11} \\
v_{22} - v_{12} & v_{32} - v_{12} & \cdots & v_{n2} - v_{12} & - v_{12} & - v_{12} \\
\vdots & \vdots & \ddots & \vdots & \vdots & \vdots \\
v_{2N} - v_{1N} & v_{3N}-v_{1N} & \cdots & v_{nN} - v_{1N} & - v_{1N} & - v_{1N} \\
 0 & 0 & \cdots & 0 & b & -b 
 \end{array} \right)
$$
Now consider the effect of an arbitrary orthogonal transformation $A = (a_{ij}):\RR^{N+1} \ra \RR^{N+1}$ on the star followed by projection to the first two coordinates.  This gives the corresponding complex numbers:
\begin{eqnarray*}
z_{\ell} & = & \sum_{j = 1}^N(a_{1j} + \ii a_{2j})(v_{\ell j} - v_{1j}) \qquad (\ell = 2, \ldots , n) \\
z_{n+1} & = &  \ b(a_{1,N+1} + \ii a_{2,N+1})- \sum_{j=1}^N(a_{1j} + \ii a_{2j})v_{1j} \\ 
z_{n+2} & = &  -b(a_{1,N+1} + \ii a_{2,N+1})- \sum_{j=1}^N(a_{1j} + \ii a_{2j})v_{1j}
\end{eqnarray*}
But since the vectors $\vec{v}_{\ell}$ form a configuration, we have $\sum_{\ell = 2}^nv_{\ell j} = - v_{1j}$ for all $j = 1, \ldots , N$.  Thus
$$
z_{n+1} + z_{n+2} + \sum_{\ell = 2}^nz_{\ell} = - (n+2) \sum_{j = 1}^N(a_{1j} + \ii a_{2j})v_{1j}\,.
$$

For the sum of squares, we obtain
$$
\begin{array}{l} 
z_{n+1}{}^2 + z_{n+2}{}^2 + \sum_{\ell = 2}^n z_{\ell}{}^2 \\
= 2\sum_{j=1}^N(a_{1j} + \ii a_{2j})^2(v_{1j}{}^2 - b^2) + 4\sum_{j<k}(a_{1j} + \ii a_{2j})(a_{1k}+ \ii a_{2k}v_{1j})v_{1k} \\
\qquad + \sum_{\ell = 2}^n\sum_{j = 1}^N(a_{1j} + \ii a_{2j})^2(v_{\ell j} - v_{1j})^2 \\
\qquad \qquad + 2\sum_{\ell = 2}^n\sum_{j<k}(a_{1j} + \ii a_{2j})(a_{1k} + \ii a_{2k})(v_{\ell j} - v_{1j})(v_{\ell k}- v_{1k})\,.
\end{array}
$$
If we let $\rho$ denote the configuration invariant, so that $\sum_{\ell = 1}^nv_{\ell j}v_{\ell k} = \rho \delta_{jk}$, then this gives
\begin{eqnarray*}
z_{n+1}{}^2 + z_{n+2}{}^2 + \sum_{\ell = 2}^n z_{\ell}{}^2 & = & \sum_{j = 1}^N(a_{1j} + \ii a_{2j})^2[(n+2)v_{1j}{}^2+ \rho - 2b^2] \\
 & & + 2(n+2)\sum_{j<k}(a_{1j} + \ii a_{2j})(a_{1k} + \ii a_{2k})v_{1j}v_{1k}\,.
\end{eqnarray*}
 On the other hand
 $$
 \begin{array}{l}
 \Big( z_{n+1} + z_{n+2} + \sum_{\ell = 2}^nz_{\ell}\Big)^2 = \\
 (n+2)^2\left\{ \sum_{j=1}^N(a_{1j}+ \ii a_{2j})^2v_{1j}{}^2 + 2\sum_{j<k}(a_{1j} + \ii a_{2j})(a_{1k}+ \ii a_{2k})v_{1j}v_{1k}\right\} \,.
\end{array}
$$ 
It now follows that equation (\ref{inv-star}) is satisfied (with $z_1, \ldots , z_n$ replaced by $z_2,\ldots , z_n,z_{n+1},z_{n+2}$) with $\ga$ independent of $A$ if and only if $2b^2 = \rho$, in which case:
$$
\frac{1}{n+2}\Big( z_{n+1} + z_{n+2} + \sum_{\ell = 2}^nz_{\ell}\Big)^2 = z_{n+1}{}^2 + z_{n+2}{}^2 + \sum_{\ell = 2}^n z_{\ell}{}^2\,.
$$
This completes the proof.

\end{proof}

The significance of this proposition is that it provides examples of frameworks that are ``flexible" in respect of equation (\ref{one}), by which we mean the edge lengths may change, but the framework can be deformed continuously through solutions to (\ref{one}) while maintaining invariance.  Configurations for which this is the case are given by the examples (\ref{non-reg-config}), (\ref{vertices-icosahedron}) and (\ref{vertices-dodecahedron}) of Section \ref{sec:orthographic}.  In each case, there is a $1$-parameter family of configurations which yield invariant double cones.  

The last formula of the proof shows that the lateral value of $\ga$ is given by
$$
\ga_{\rm lat} = \frac{n+1}{n+2}\,.
$$
At an apex, the value of $\ga$ is obtained from (\ref{sigma-c}).  In fact
$$
\si = nb^2-\rho = (n-2)b^2\,,
$$
so that
$$
\ga_{\rm apex} = \frac{\si}{\si + \rho } = \frac{n-2}{n}\,.
$$
We note also that $\ga_{\rm lat}$ and $\ga_{\rm apex}$ depend only on the cardinality $n$ of the configuration.  Although these two values are never equal, remarkably, when $\vec{v}_1, \ldots , \vec{v}_n$ are the vertices of a regular polygon in the plane, the addition of one more edge joining the two apexes yields an invariant framework in $\RR^3$ with $\ga =(n+1)/(n+2)$ constant.  The addition of this extra edge makes the framework into a complete graph on $n+2$ vertices.

Say that a framework is \emph{immersed} if all vertices are distinct; say that it is \emph{embedded} if it is immersed and no two edges intersect except at their endpoints.  Then the complete graph on $n+2$ vertices can always be embedded as an invariant framework in $\RR^{n+1}$, its vertices and edges corresponding to the $1$-skeleton of a regular $n+1$--simplex.

\medskip
\begin{center}
\setlength{\unitlength}{0.254mm}
\begin{picture}(120,120)(120,-170)
        \allinethickness{0.254mm}\path(120,-120)(240,-120) 
        \allinethickness{0.254mm}\path(180,-115)(180,-50) 
        \allinethickness{0.254mm}\path(180,-125)(180,-170) 
        \allinethickness{0.254mm}\path(180,-50)(120,-120) 
        \allinethickness{0.254mm}\path(180,-50)(240,-120) 
        \allinethickness{0.254mm}\path(240,-120)(180,-170) 
        \allinethickness{0.254mm}\path(180,-170)(120,-120) 
        \allinethickness{0.254mm}\path(155,-100)(120,-120) 
        \allinethickness{0.254mm}\path(155,-100)(175,-105) 
        \allinethickness{0.254mm}\path(240,-120)(195,-110) 
        \allinethickness{0.254mm}\path(155,-100)(180,-50) 
        \allinethickness{0.254mm}\path(180,-170)(165,-130) 
        \allinethickness{0.254mm}\path(155,-100)(160,-115) 
\end{picture}
\end{center}
\medskip

\begin{corollary} \label{cor:inv-complete-graph}  There exists an invariant immersed framework in $\RR^3$ with $\ga$ constant corresponding to the complete graph on $n+2$ ($n\geq 2$) vertices.  There exists an invariant embedded framework in $\RR^n$ with $\ga$ constant corresponding to the complete graph on $n+2$ ($n\geq 2$) vertices.
\end{corollary}

\begin{proof}  Let $\vec{v}_1, \ldots , \vec{v}_n$ be unit length vertices of a regular polygon in the plane and construct the unique invariant double cone of the above proposition.  Since $\rho = n/2$ (Section \ref{sec:orthographic}), its height is given by $b = \sqrt{n}/2$.  Lemma \ref{lem:extension-star} of Appendix \ref{sec:invariant-structures} shows that the addition of another vertex at a distance $x$ ($\neq -nb$) along the axis of symmetry of the cone connected to the apex, yields another invariant framework, with new $\ga$ at the apex given by (\ref{ga-ext}), that is by
$$
\wt{\ga}_{\rm apex} = \frac{(n+1)(x^2+nb^2\ga_{\rm apex} )}{(x+nb)^2}\,,
$$
where $\ga_{\rm apex} = (n-2)/n$ is the value of $\ga$ prior to the addition of the new vertex.  But now substitution of the values $b = \sqrt{n}/2$ and $x=\sqrt{n}$ yields precisely the value $\wt{\ga}_{\rm apex} = (n+1)/(n+2)$, which is the value at a lateral vertex.  Furthermore, the new edges connects the two apexes.

For the second part of the corollary, when $n\geq 3$ one proceeds as above but now with $\vec{v}_1, \ldots , \vec{v}_n$ the vertices of a regular $(n-1)$--simplex in $\RR^{n-1}$.  Indeed, as noted after the above proposition, the value of $\ga_{\rm lat}$ depends only on the cardinality of the configuration.  The resulting double cone is clearly embedded, since when $n\geq 3$, the centre of mass of the regular simplex is not contained in any of its edges.  But the edge joining the two apexes intersects the subspace $\RR^{n-1}$ precisely at this point.  When $n=2$, the framework can be taken to be the projection of the $1$-skeleton of a tetrahedron as indicated in the right-hand figure below.
\end{proof} 

In Appendix \ref{sec:invariant-structures} we show how to build more intricate frameworks, called \emph{invariant structures}, from the elementary ones constructed above.   Examples with $\ga$ constant such as those given by the above corollary, are particularly significant in respect of our model of an elementary universe developed in Section \ref{sec:particles}, where we view them as stable particles.

Given a solution $\phi$ to (\ref{one}) on a connected graph $\Ga$, the \emph{global embedding problem} is to realise the graph as an invariant embedded framework $\Ff = (V,E)$ in a Euclidean space $\RR^N$ in such a way that $\phi (\vec{x}) = P(\vec{x})$ for all $\vec{x}\in V$.  The above corollary shows that the solution corresponding to the projection of the framework of a regular $n$-simplex (in $\RR^n$) can be embedded in $\RR^{n-1}$ for $n\geq 3$.  Clearly, if $\Ga$ can be realised as a planar graph with non-intersecting edges in such a way that $\phi$ is the position function associated to the vertices, then the solution $(\Ga , \phi )$ is globally embedded in $\RR^2$.  This is the case for the left-hand solution below, for which $\ga = 1/3$ (see Example \ref{ex:lift}):

\medskip
\begin{center}
\setlength{\unitlength}{0.254mm}
\begin{picture}(312,107)(25,-151)
        \allinethickness{0.254mm}\path(80,-60)(80,-140) 
        \allinethickness{0.254mm}\path(40,-100)(120,-100) 
        \allinethickness{0.254mm}\path(80,-60)(120,-100) 
        \allinethickness{0.254mm}\path(120,-100)(80,-140) 
        \allinethickness{0.254mm}\path(80,-140)(40,-100) 
        \allinethickness{0.254mm}\path(40,-100)(80,-60) 
        \allinethickness{0.254mm}\special{sh 0.3}\put(80,-60){\ellipse{4}{4}} 
        \allinethickness{0.254mm}\special{sh 0.3}\put(80,-100){\ellipse{4}{4}} 
        \allinethickness{0.254mm}\special{sh 0.3}\put(40,-100){\ellipse{4}{4}} 
        \allinethickness{0.254mm}\special{sh 0.3}\put(120,-100){\ellipse{4}{4}} 
        \allinethickness{0.254mm}\special{sh 0.3}\put(80,-140){\ellipse{4}{4}} 
        \put(85,-96){\shortstack{$0$}} 
        \put(85,-56){\shortstack{$i$}} 
        \put(125,-96){\shortstack{$1$}} 
        \put(22,-96){\shortstack{$-1$}} 
        \put(80,-151){\shortstack{$-i$}} 
        \allinethickness{0.254mm}\path(240,-140)(335,-140) 
        \allinethickness{0.254mm}\path(335,-140)(285,-70) 
        \allinethickness{0.254mm}\path(285,-70)(240,-140) 
        \allinethickness{0.254mm}\path(285,-115)(335,-140) 
        \allinethickness{0.254mm}\path(285,-115)(285,-70) 
        \allinethickness{0.254mm}\path(285,-115)(240,-140) 
        \allinethickness{0.254mm}\special{sh 0.3}\put(285,-70){\ellipse{4}{4}} 
        \allinethickness{0.254mm}\special{sh 0.3}\put(285,-115){\ellipse{4}{4}} 
        \allinethickness{0.254mm}\special{sh 0.3}\put(240,-140){\ellipse{4}{4}} 
        \allinethickness{0.254mm}\special{sh 0.3}\put(335,-140){\ellipse{4}{4}} 
\end{picture}
\end{center}
\medskip 
In fact, by direct calculation, one can show that this solution admits no global embedding in $\RR^3$ and by the lemma below, it cannot admit any global embedding in $\RR^N$ for $N>3$, hence $\RR^2$ is the unique Euclidean space into which it embeds. 

\begin{lemma} \label{lem:embed}  Let $\phi$ be a solution to {\rm (\ref{one})} on a connected graph $\Ga = (V,E)$ such that $\ga (x)<1$ for all $x\in V$.  Suppose that $n_0$ is the smallest degree of the vertices of $\Ga$.  Then the solution cannot be realised as an invariant embedded framework in any $\RR^N$ with $N>n_0$.
\end{lemma}

\begin{proof}  Suppose on the contrary that there exists an invariant embedding in $\RR^N$ with $N>n_0$.   Let $\vec{x}$ be a vertex of degree $n_0$ and consider the affine subspace $S$ spanned by the edges emanating from $\vec{x}$.  Let $P: \RR^N \ra \CC$ be an orthogonal projection.  Then since $\dim S<N$, there exists a Euclidean motion $A : \RR^N \ra \RR^N$ such that $P(A(S))$ is entirely real.  In particular each neighbour $\vec{y}$ of $\vec{x}$ has $P(A(\vec{y}))$ real.  But by Lemma \ref{lem:ngaleq1}, this contradicts the hypothesis that $\ga <1$.
\end{proof} 

The right-hand figure above shows a projection of the tetrahedron.  This can be realised either as an embedded framework in $\RR^2$, or in $\RR^3$.  On the other hand, a solution corresponding to the projection of a cube cannot be embedded in the plane, since edges must cross; $\RR^3$ is the unique space into which it embeds.  Corollary \ref{cor:inv-complete-graph} shows that the complete graph on five vertices admits solutions to (\ref{one}) with $\ga = 4/5$ which correspond to embeddings in $\RR^3$ and $\RR^4$, and only into these spaces.  The star graph of Example \ref{ex:cyclic-5vertices} shows that there exist examples with $\ga <1$ which cannot be embedded into any Euclidean space:  since the degree of each vertex is $2$, by the above lemma it can only embed into $\RR^2$, but then edges must cross.  In the next section, we address the \emph{local} embedding problem, or \emph{lifting problem} and show that, for a given solution to (\ref{one}), we can always embed a vertex and its neighbours in an invariant way.

\section{The lifting problem} \label{sec:distance}
Given a solution $\phi$ to (\ref{one}), at each vertex $x$, our aim is to construct a configured star in some Euclidean space $\RR^N$ whose external vertices project to the points $\phi (y)-\phi (x)$ ($y\sim x$) of the complex plane.  To do this, we establish a converse to Corollary \ref{cor:configured-star}.  We shall refer to the problem of constructing such a star as \emph{the lifting problem}.  At a vertex of degree three with $\phi$ holomorphic, this is the Theorem of Axonometry of Gauss \cite{Ga}.  It turns out that provided $\ga <1$ ($n=$ vertex degree), the lifting problem can always be solved.  Remarkably, when $N=3$, the configured star is unique up to a sign ambiguity and arises as the minimizer of a natural functional determined by $\phi$.  Relative edge-length and so relative distance on the graph, can now be defined in terms of virtual configured stars.

Suppose we are given $z_1, \ldots , z_n\in \CC$ ($n\geq 2$) not all zero satisfying
  \begin{equation} \label{gz}
\frac{\ga}{n} \left( \sum_{\ell = 1}^nz_{\ell}\right)^2 = \sum_{\ell =1}^n z_{\ell}{}^2 \quad (\ga \in \RR )\,.
\end{equation}
For a given $N$ with $2\leq N\leq n$, we wish to construct a configured star $W = (\vec{x}_1|\vec{x}_2|\cdots |\vec{x}_n)$ in $\RR^N$ with $z_{\ell}$ the orthogonal projection of $\vec{x}_{\ell}$.  For convenience, write $z_{\ell} = x_{\ell 1} + \ii x_{\ell 2} = \al_{\ell} + \ii \be_{\ell}$, so that
$$
W = \left( \begin{array}{c|c|c|c} 
\al_1 & \al_2 & \cdots & \al_n \\
\be_1 & \be_2 & \cdots & \be_n \\
x_{13} & x_{23} & \cdots & x_{n3} \\
 \vdots & \vdots & \vdots & \vdots \\
  x_{1N} & x_{2N} & \cdots & x_{nN} \end{array} \right)\,.
  $$
  For $N\geq 3$, we are required to solve the system:
\begin{equation} \label{system}
  WW^t = \rho I_N+\si \vec{u}\vec{u}^t, \qquad \sum_{\ell =1}^n \vec{x}_{\ell} = \sqrt{n(\si + \rho )} \, \vec{u}\,,
\end{equation}
  for $x_{\ell j}$ $(\ell = 1, \ldots , n;\ j = 3, \ldots , N)$, $\rho >0$, $\si$ such that $\rho + \si >0$ and $\vec{u} \in \RR^N$ unit, with $\ga = \si / (\si + \rho )$.  If we require the star to be regular, then we have the further condition:
$$
||\vec{x}_{\ell} ||=r \qquad \forall \ell = 1, \ldots , n, \quad {\rm for\ some} \ r>0\,.
$$ 
If $N=2$, then $W$ is determined and for $n\geq 3$, the system (\ref{system}) is not in general satisfied; if $N=n=2$ it is always satisfied. 
  
 Note also the normalizing freedom; namely, (\ref{gz}) is invariant by the replacement $z_{\ell} \mapsto \la z_{\ell}$ $(\la \in \CC )$.  We will discuss the effect of this on the parameters of the problem below.  

\begin{theorem}  \label{thm:lift} Given a non-zero solution $\{ z_1, \ldots , z_n ; \ga\}$ to {\rm (\ref{gz})} with $n\geq 3$ satisfying $\ga <1$; then for each $3\leq N\leq n$, there is a full configured star $W = (\vec{x}_1|\vec{x}_2|\cdots |\vec{x}_n)$ in $\RR^N$ such that $z_{\ell}$ is the orthogonal projection of $\vec{x}_{\ell}$\,.  There are two cases:

{\rm (i)}  If $N=3$ there is a $3\times n$--real matrix $A$ and a vector \ $\vec{b}\in \RR^3$, both of which depend only on \ $z_1, \ldots , z_n$, such that the vertices of the star are determined by the solutions \ $\vec{t}\in \RR^n$ to the equation $A\vec{t}=\vec{b}$; provided the $z_{\ell}$ do \emph{not} satisfy the identity
\begin{equation} \label{lift-condition}
n\sum_{\ell}|z_{\ell}|^2 + (\ga - 2)\Big|\sum_{\ell}z_{\ell}\Big|^2 = 0\,,
\end{equation}
then the star corresponds to the minimal Euclidean norm solution to this equation and is unique up to a sign ambiguity. 

{\rm (ii)}  If $N>3$, the star is no longer unique, and although it corresponds to a solution $X$ to a more general system $AX=B$, where $X$ is a variable $n\times (N-2)$ --matrix and $B$ is an $3\times (N-2)$ --matrix which depends only on \ $z_1, \ldots , z_n$, it is not in general minimizing amongst solutions.
\end{theorem}

In the case when $N=3$, except for special cases, the theorem says that the configured star arises as the absolute minimum of a functional determined by the solution to (\ref{gz}).  The special solutions which satisfy (\ref{lift-condition}) are characterized in Lemma \ref{lem:lin-comb} below.  Note that we do not insist that $\sum_{\ell} z_{\ell} \neq 0$, which is necessary to determine $\ga$; in fact in this case we have a family of stars for different choices of $\ga$ subject to $\ga <1$.  In order to prove the theorem, we first of all establish some preliminary identities.   
  
\begin{lemma} If the system {\rm (\ref{system})} is satisfied, then we have the following equalities:
\begin{equation} \label{identity-1}
u_1 = \frac{1}{\sqrt{n(\si + \rho )}} \sum_{\ell =1}^n \al_{\ell} , \quad  u_2 = \frac{1}{\sqrt{n(\si + \rho )}} \sum_{\ell =1}^n \be_{\ell}\,;
\end{equation}
\begin{equation} \label{identity-2}
\left\{ \begin{array}{rcl}
\ds \frac{\ga}{n} \big(\sum_{\ell} \al_{\ell}\big)^2 + \rho & = & \ds \sum_{\ell}\al_{\ell}{}^2 \\
\ds \frac{\ga}{n} \big(\sum_{\ell} \be_{\ell}\big)^2 + \rho & = & \ds \sum_{\ell}\be_{\ell}{}^2 
 \end{array}
\right.
\end{equation}
In particular, we have the identity:
\begin{eqnarray*} 
\rho & = & \frac{1}{2}\sum_{\ell} (\al_{\ell}{}^2 + \be_{\ell}{}^2) - \frac{\ga}{2n}\left(\Big(\sum_{\ell}\al_{\ell}\Big)^2 + \Big(\sum_{\ell}\be_{\ell}\big)^2\right) \nonumber \\
 & = & \frac{1}{2}\sum_{\ell} z_{\ell} \ov{z}_{\ell} - \frac{\ga}{2n} \Big( \sum_{\ell}z_{\ell}\Big)\Big( \sum_{\ell} \ov{z}_{\ell}\Big)\,. 
 \end{eqnarray*}
Furthermore, provided $\ga <1$, then $\rho >0$. 
\end{lemma}

\begin{proof}
Identity (\ref{identity-1}) follows from the second equation of (\ref{system}).  Then the first equation of (\ref{system}) implies that
$$
\sum_{\ell}\al_{\ell}{}^2 = \rho + \si u_1{}^2 = \rho + \frac{\si}{n(\si + \rho )}\Big(\sum_{\ell}\al_{\ell}\Big)^2 = 
\rho + \frac{\ga}{n}\Big(\sum_{\ell}\al_{\ell}\Big)^2\,,
$$
which gives the first identity of (\ref{identity-2}).  The second identity follows similarly.  The inequality $\rho >0$ is an immediate consequence of (\ref{inductive-inequality}).
\end{proof}

These are necessary conditions that must be satisfied by any configured star that projects to a solution of (\ref{gz}). 

Write vectors in column form and set $\vec{\al} = \left( \begin{array}{c} \al_1 \\ \vdots \\ \al_n \end{array}\right) ,\ \vec{\be} = \left( \begin{array}{c} \be_1 \\ \vdots \\ \be_n \end{array}\right)$ and ${\bf 1}\in \RR^n$ to be the vector having $1$ in each of its entries.  

In the context of equation (\ref{gz}), by \emph{normalization} we mean the freedom to multiply each $z_{\ell}$ by a non-zero constant $\la \in \CC$; this leaves the equation invariant.  So we no longer allow the more general freedom $\phi \mapsto \la \phi + \mu$ associated to equation (\ref{one}) since that disturbs the requirement that $(z_1, \ldots , z_n)$ represents a solution which is zero at the internal vertex of a star.
The following lemma concerns normalizations that are useful in the proof of the theorem.
 
 \begin{lemma}  \label{lem:lin-comb} {\rm (i)}  If $\sum_{\ell}z_{\ell} \neq 0$, then there is a normalization with respect to which $\sum_{\ell}z_{\ell} = 1$, in which case:
\begin{equation} \label{normal-1}
\sum_{\ell} \al_{\ell} = 1, \quad \sum_{\ell} \be_{\ell} = 0, \quad \sum_{\ell} \al_{\ell}\be_{\ell} = 0\,.
\end{equation}

{\rm (ii)}  For a non-trivial solution to {\rm (\ref{gz})} with $\ga <1$, suppose there exists a linear combination of the form \ $a\vec{\al} + b \vec{\be} + c{\bf 1}=0$ for $a,b,c \in \RR$ not all zero.  Then there exists a normalization factor $\la$ such that $\la z_{\ell} = \al_{\ell} + \ii \be$, with $\al_{\ell}\in \RR$ ($\ell = 1, \ldots , n)$ and with $\be \in \RR$ constant.  Furthermore $\be \neq 0$ and the identities: 
\begin{equation} \label{normal-2}
\sum_{\ell} \al_{\ell} = 0 \qquad {\rm and} \qquad \ga = 1-\frac{1}{n\be^2}\sum_{\ell}\al_{\ell}{}^2\,,
\end{equation}
are satisfied.  In particular {\rm (\ref{lift-condition})} holds.

Conversely, if {\rm (\ref{lift-condition})} holds, then there is a linear combination of the form \ $a\vec{\al} + b \vec{\be} + c{\bf 1}=0$ for $a,b,c \in \RR$ not all zero. 
 \end{lemma}
 
 \begin{proof}  (i) This is trivial; if $\sum_{\ell}z_{\ell} = \mu \neq 0$, then we multiply each $z_{\ell}$ by $\la = \mu^{-1}$.  Then after normalization
 $$
 \ga = n\sum_{\ell = 1}^n (\al_{\ell} + \ii \be_{\ell})^2 = n\sum_{\ell = 1}^n(\al_{\ell}{}^2 - \be_{\ell}{}^2 + 2\ii \al_{\ell}\be_{\ell})\,,
 $$
 so that $\sum_{\ell}\al_{\ell}\be_{\ell} = 0$.
 
(ii)  Suppose there exists a linear combination of the form $a\vec{\al} + b \vec{\be} + c{\bf 1}=0$ for $a,b,c \in \RR$ not all zero.  If $b\neq 0$, then
 $$
 z_{\ell} = \left( 1 - \ii \frac{a}{b}\right)\al_{\ell} - \ii \frac{c}{b}\,,
 $$
 so that 
 $$
 b(b+ \ii a)z_{\ell} = \{ (a^2+b^2)\al_{\ell} + ac\}  - \ii cb\,,
 $$
 which is of the required form.  If on the other hand $b = 0$, then since $\ga <1$, by Lemma \ref{lem:ngaleq1}, $a\neq 0$ and $-\ii z_{\ell} = \be_{\ell} + \ii (c/a)$, which is once more of the required form.
 
 Suppose now that $z_{\ell}$ has the form $z_{\ell} = \al_{\ell} + \ii \be$ with $\be \in \RR$ constant.  First note that $\be \neq 0$, otherwise we would contradict Lemma \ref{lem:ngaleq1} once more.  The identities (\ref{normal-2}) now follow by taking the real and imaginary part of (\ref{gz}).  It is straightforward to check that these imply (\ref{lift-condition}).  
 
 Conversely, suppose that (\ref{lift-condition}) holds.  If\ $\sum_{\ell}z_{\ell} = 0$, then from (\ref{lift-condition}),\ $\sum_{\ell} |z_{\ell}|^2=0$ which implies each $z_{\ell}$ vanishes, contradicting our hypothesis.  Hence $\sum_{\ell} z_{\ell} \neq 0$.  Now normalize so that $\sum_{\ell} z_{\ell} = \ii$.  Then
 $$
 \sum_{\ell} \al_{\ell =1}^n = 0,\quad \sum_{\ell = 1}^n \be_{\ell} = 1,\quad \sum_{\ell = 1}^n \al_{\ell}\be_{\ell} = 0\,,
 $$
and
 $$
 \ga = -n\sum_{\ell = 1}^n(\al_{\ell}{}^2 - \be_{\ell}{}^2)\,.
 $$
 But now (\ref{lift-condition}) implies that
 $$
 n\sum_{\ell = 1}^n \be_{\ell}{}^2 - 1 = 0 \quad \Leftrightarrow \quad n\sum_{\ell = 1}^n \be_{\ell}{}^2 - \left(\sum_{\ell = 1}^n\be_{\ell}\right)^2 = 0\,,
 $$
 which, by the Cauchy-Schwarz inequality implies that $\be_{\ell} = \be$ constant.  Since $\ga <1$, we must have $\be \neq 0$ and there exists a non-trivial linear combination: $\vec{\be} - \be {\bf 1}=0$.
 \end{proof} 

\medskip

\noindent \emph{Proof of Theorem} \ref{thm:lift}:  Let $\{ z_1, \ldots , z_n ; \ga\}$ be a non-trivial solution to (\ref{gz}) satisfying $\ga <1$.  Set  
\begin{equation} \label{rho}
\rho = \frac{1}{2}\sum_{\ell} z_{\ell} \ov{z}_{\ell} - \frac{\ga}{2n} \Big( \sum_{\ell}z_{\ell}\Big)\Big( \sum_{\ell} \ov{z}_{\ell}\Big) >0 \,, 
\end{equation}
and 
\begin{equation} \label{si}
\si = \frac{\ga \rho }{1 - \ga} \quad (\Rightarrow \si + \rho = \rho /(1-\ga )>0)\,.
\end{equation} 
Then from (\ref{gz}) we have
$$
\frac{\ga}{n} \Big(\sum_{\ell}\al_{\ell}\Big)^2 - \sum_{\ell}\al_{\ell}{}^2 = \frac{\ga}{n} \Big(\sum_{\ell}\be_{\ell}\Big)^2 - \sum_{\ell}\be_{\ell}{}^2
$$
and
$$ 
\frac{\ga}{n} \Big(\sum_{\ell}\al_{\ell}\Big)\Big(\sum_{\ell}\be_{\ell}\Big)=\sum_{\ell}\al_{\ell}\be_{\ell}\,,
$$
so that the identities (\ref{identity-2}) hold:
$$ 
\rho = \sum_{\ell}\al_{\ell}{}^2 - \frac{\ga}{n} \Big(\sum_{\ell}\al_{\ell}\Big)^2 = \sum_{\ell}\be_{\ell}{}^2 - \frac{\ga}{n} \Big(\sum_{\ell}\be_{\ell}\Big)^2\,.
$$
Define $u_1$ and $u_2$ by (\ref{identity-1}).  We are required to solve (\ref{system}) for $x_{\ell j}$ $(\ell = 1, \ldots , n;\ j = 3, \ldots , N)$ and for $u_3, \ldots , u_N$ satisfying $u_1{}^2+u_2{}^2 + u_3{}^2 + \cdots + u_N{}^2 = 1$.  It is convenient to rewrite (\ref{system}) as follows.

Set 
$$
A:= \left( \begin{array}{cccc} \al_1 & \al_2 & \cdots & \al_n \\
\be_1 & \be_2 & \cdots & \be_n \\
1 & 1 & \cdots & 1 \end{array} \right) , \qquad 
X:= \left( \begin{array}{cccc}
x_{13} & x_{14} & \cdots & x_{1N} \\
x_{23} & x_{24} & \cdots & x_{2N} \\
\vdots & \vdots & \ddots & \vdots \\
x_{n3} & x_{n4} & \cdots & x_{nN} 
\end{array} \right)\,.
$$
Then 
$$ \label{system-2}
AX  =  \left( \begin{array}{cccc} 
\inn{\vec{\al} , \vec{X}_3} & \inn{\vec{\al} , \vec{X}_4} & \cdots & \inn{\vec{\al} , \vec{X}_N} \\
\inn{\vec{\be} , \vec{X}_3} & \inn{\vec{\be} , \vec{X}_4} & \cdots & \inn{\vec{\be} , \vec{X}_N} \\
\sum_{\ell} x_{\ell 3} & \sum_{\ell} x_{\ell 4} & \cdots & \sum_{\ell} x_{\ell N} 
\end{array} \right)\,, 
$$
where $\inn{\vec{\al} , \vec{X}_j}$ denotes the Euclidean inner product of the vectors $\vec{\al} = \left( \begin{array}{c} \al_1 \\ \vdots \\ \al_n \end{array} \right)$ and $\vec{X}_j = \left( \begin{array}{c} x_{1j} \\ \vdots \\ x_{nj} \end{array} \right)$, for $j = 3, \ldots , N$ (similarly for $\inn{\vec{\be} , \vec{X}_j}$).
Then solving (\ref{system}) is equivalent to solving
\begin{eqnarray} 
 AX =  B:& = &  \left( \begin{array}{cccc}
\si u_1u_3 & \si u_1u_4 & \cdots & \si u_1u_N \\
\si u_2u_3 & \si u_2u_4 & \cdots & \si u_2u_N \\
\sqrt{n(\si + \rho )}\,u_3 & \sqrt{n(\si + \rho )}\,u_4 & \cdots & \sqrt{n(\si + \rho )}\,u_N
\end{array} \right) \nonumber \\
 & = & \left( \begin{array}{c} \si u_1 \\ \si u_2 \\ \sqrt{n(\si + \rho )} \end{array} \right)
\left( \begin{array}{cccc} u_3 & u_4 & \cdots & u_N \end{array} \right)\,, \label{system-2}
\end{eqnarray}
subject to the constraint: 
\begin{eqnarray} \label{system-2-constraint}
X^tX & = & \left( \begin{array}{cccc} 
\inn{\vec{X}_3 , \vec{X}_3} & \inn{\vec{X}_3 , \vec{X}_4} & \cdots & \inn{\vec{X}_3 , \vec{X}_N} \\
\inn{\vec{X}_4 , \vec{X}_3} & \inn{\vec{X}_4 , \vec{X}_4} & \cdots & \inn{\vec{X}_4 , \vec{X}_N} \\
\vdots & \vdots & \ddots & \vdots \\
\inn{\vec{X}_N , \vec{X}_3} & \inn{\vec{X}_N , \vec{X}_4} & \cdots & \inn{\vec{X}_N , \vec{X}_N}
\end{array} \right) \\
& = & 
\left( \begin{array}{cccc}
\rho + \si u_3{}^2 & \si u_3u_4 & \cdots & \si u_3u_N \\
\si u_4u_3 & \rho + \si u_4{}^2 & \cdots & \si u_4u_N \\
 \vdots & \vdots & \ddots & \vdots \\
 \si u_Nu_3  & \si u_Nu_4 & \cdots & \rho + \si u_N{}^2 
 \end{array} \right) \nonumber \\
  & = & \rho I_{N-2} + \si \left( \begin{array}{c} u_3 \\ u_4 \\ \vdots \\ u_N \end{array} \right) (u_3 \ u_4 \ \cdots \ u_N) \nonumber
 \end{eqnarray}
 It is important to note the sign ambiguity: the equations are invariant under the simultaneous replacement of $u_j$ by $-u_j$ for $j = 3, \ldots , N$ and of $X$ by $-X$.  This ambiguity represents two choices for the configured star.
 In order to study solutions to (\ref{system-2}), we will employ the Moore--Penrose pseudo-inverse of the matrix $A$ \cite{Pe-0}. 

In general, if $A$ is an $m\times n$ --matrix with $m\leq n$ such that the rows of $A$ are linearly independent, then $AA^t$ is invertible and the \emph{Moore-Penrose pseudo-inverse} is the $n\times m$ --matrix
$$
A^+ : = A^t(AA^t)^{-1}\,.
$$
If $B$ is an $m\times p$ --matrix and $X$ a variable $n\times p$ --matrix, then $Z = A^+B$ is a solution to the equation
$$
AX=B\,;
$$
furthermore, it satisfies $||Z||_F \leq ||X||_F$ for all solutions $X$, where 
$$
||X||_F^2 = \sum_{i = 1}^m\sum_{j=1}^n|x_{ij}|^2 = {\rm trace}\, (X^tX)\,,
$$ is the square of the Frobenius norm of $X$.  The general solution is given by 
$$
X = A^+B+(I_n-A^+A)V\,,
$$
where $V$ is an $n\times p$ --matrix which can take on arbitrary values.  We now return to the system (\ref{system-2}).

Suppose first that the rows of $A$ are linearly independent.  Set $Z = A^+B$.  Then $AZ=B$.  In order to study the constraint (\ref{system-2-constraint}), we calculate $Z^tZ$:
\begin{eqnarray*}
Z^tZ & = & (A^+B)^tA^+B = B^t(A^+)^tA^+B \\
 & = & B^t(A^t(AA^t)^{-1})^tA^t(AA^t)^{-1}B \\
 & = & B^t((AA^t)^{-1})^tAA^t(AA^t)^{-1}B \\
  & = & B^t(AA^t)^{-1}B\,.
\end{eqnarray*}
Then from (\ref{identity-1}) and (\ref{identity-2}),
$$
AA^t=\left( \begin{array}{ccc}
\rho + \si u_1{}^2 & \si u_1u_2 & \sqrt{n(\si + \rho )}\, u_1 \\
\si u_1u_2 & \rho + \si u_2{}^2 & \sqrt{n(\si + \rho )}\, u_2 \\
\sqrt{n(\si + \rho )}\, u_1 & \sqrt{n(\si + \rho )} \, u_2 & n
\end{array} \right)\,,
$$
with determinant:
\begin{equation} \label{det-AAt}
\det (AA^t) = n\rho^2(1-u_1{}^2-u_2{}^2)\,.
\end{equation}
The inverse is given by:
$$ \begin{array}{l}
 (AA^t)^{-1} =  \\
\ds  \frac{1}{n\rho (1-u_1{}^2-u_2{}^2)} \left( \begin{array}{ccc}
n(1 - u_2{}^2) & nu_1u_2 & - \sqrt{n(\si + \rho )}\,u_1\\
nu_1u_2 & n(1-u_1{}^2) & - \sqrt{n(\si + \rho )}\,u_2 \\
-\sqrt{n(\si + \rho )}\,u_1 & - \sqrt{n(\si + \rho )}\,u_2 & \rho + \si (u_1{}^2 + u_2{}^2)
\end{array} \right)
\end{array}
$$
We then compute to find
\begin{equation} \label{ZtZ} 
Z^tZ = \left( \frac{\rho}{1-u_1{}^2-u_2{}^2} + \si \right) \left( \begin{array}{c} u_3 \\ u_4 \\ \vdots \\ u_N \end{array} \right) (u_3 \ u_4 \ \cdots \ u_N)\,.
\end{equation}
On comparing (\ref{ZtZ}) with (\ref{system-2-constraint}), we see that we always have equality of the right-hand sides when $N=3$.  In particular, the solution $X=Z$ which satisfies the constraint minimizes $||X||_F$ among solutions $X$ to $AX=B$ and so is unique up to the sign ambiguity referred to above.  (If also $n=3$, then there is one and only one solution to $AX=B$; this solution automatically satisfies the constraint).  On the other hand, for $N>3$, the solution $X=Z$ never satisfies the constraint and we are required to look at a more general set of solutions. 

If an $n\times (N-2)$ --matrix $Y$ satisfies the homogeneous equation $AY=\vec{0}$, where $\vec{0}$ is the $3\times (N-2)$ zero matrix, then $X=Z+Y$ solves the equation (\ref{system-2}).  We now look at the homogeneous equation.   Write the column vectors of $Y$ as $\vec{Y}_j$ ($j = 3, \ldots , N$).  

Suppose first that $\sum_{\ell}z_{\ell} \neq 0$.  Then by Lemma \ref{lem:lin-comb}, there is a normalization with respect to which the conditions (\ref{normal-1}) are satisfied.  Note that in this normalization $u_2=0$ and $u_1 = 1/\sqrt{n(\si + \rho )} = \sqrt{\ga / n\si}$.  Consider $\RR^n$ with coordinates $(t_1, \ldots , t_n)$.  Let $\Pi$ be the $(n-1)$-dimensional hyperplane $t_1+\cdots + t_n=0$ with inner product induced from the standard one on $\RR^n$.  Then $\vec{\be} \in \Pi$.  Let $\pi : \RR^n \ra \Pi$ be orthogonal projection and consider $\pi (\vec{\al} )$; this vector (which could vanish) is orthogonal to $\vec{\be}$.  Then a necessary and sufficient condition that $Y$ satisfy $AY=\vec{0}$, is that the column vectors $\vec{Y}_j \ (j = 3, \ldots , N)$ lie in $\Pi$ and are orthogonal to both $\vec{\be}$ and $\pi (\vec{\al})$.  Let us now consider the constraint (\ref{system-2-constraint}) for $X=Z+Y$. 

Then 
$$
X^tX = Z^tZ + Y^tZ + Z^tY + Y^tY\,.
$$
Given (\ref{ZtZ}), we see that $X$ satisfies (\ref{system-2-constraint}) if and only if
\begin{equation} \label{constraint-3}
\rho I_{N-2} =  \frac{\rho}{u_3{}^2 + \cdots + u_N{}^2}  \left( \begin{array}{c} u_3 \\ u_4 \\ \vdots \\ u_N \end{array} \right) (u_3 \ u_4 \ \cdots \ u_N) + Y^tZ + Z^tY + Y^tY\,.
\end{equation}
Now a calculation shows that
\begin{eqnarray} \label{Z}
& & \qquad Z  =  A^+B \\
 & &  =  \frac{1}{n(1-u_1{}^2-u_2{}^2)}\left( \begin{array}{ccc} \al_1 & \be_1 & 1 \\ \al_2 & \be_2 & 1 \\ \vdots & \vdots & \vdots \\ \al_n & \be_n & 1 \end{array} \right) \left( \begin{array}{c} - nu_1 \\ - n u_2 \\ \sqrt{n(\si + \rho )} \end{array} \right) \left( \begin{array}{ccc} u_3 & \cdots & u_N \end{array} \right)\,. \nonumber
\end{eqnarray}

Recall that in $\RR^N$ we use co\"ordinates $(y_1, \ldots , y_N)$.  By a rotation of the $(N-2)$--dimensional subspace spanned by $(y_3, \ldots , y_N)$, we may suppose that the vector $(u_3, \ldots , u_N)$ is directed along $\pa / \pa y_3$; thus $u_4=u_5=\cdots =u_N=0$ and $u_3{}^2 = 1 - u_1{}^2 - u_2{}^2$.  We now construct $Y$ in a judicious way in order to satisfy the constraint. 

Let $Y = (\vec{0}, \vec{Y}_4, \ldots , \vec{Y}_N)$ be formed by setting the first vector $\vec{Y}_3$ equal to zero, and, in addition to $\vec{Y}_j \ (j = 4, \ldots , N-2)$ all lying in $\Pi$, we require that they are mutually orthogonal of length $\sqrt{\rho}$ and are orthogonal to both $\vec{\be}$ and $\pi (\vec{\al})$.  This is possible since $N\leq n$ and can be achieved by Gram-Schmidt orthonormalization.  Then by construction, $Y^tZ=\vec{0}$ since $\inn{\vec{Y}_j, \vec{\al}} = \inn{\vec{Y}_j, \vec{\be}} = \inn{\vec{Y}_j, {\bf 1}} = 0$ for $j = 3, \ldots , N$, and 
$$
Y^tY = \rho I_{N-2} - \rho \left( \begin{array}{c} 1 \\ 0 \\ \vdots \\ 0\end{array} \right) \left( \begin{array}{cccc} 1 & 0 & \cdots & 0 \end{array} \right)\,.
$$
But this is exactly what is required to satisfy the constraint (\ref{constraint-3}) when $u_4 = \cdots = u_N = 0$.  This solves the lifting problem in the case when the rows of $A$ are linearly independent and $\sum_{\ell}z_{\ell} \neq 0$.  Note that provided $\pi (\vec{\al}) \neq 0$, the solution given by $Y$ is unique up to orthogonal transformation of the subspace of $\Pi$ orthogonal to $\pi (\vec{\al})$ and $\vec{\be}$.  If $\pi (\vec{\al})=0$ there is a further freedom in the construction of $Y$. 

If now the rows of $A$ are linearly independent and $\sum_{\ell} z_{\ell} = 0$, then we proceed as above, replacing the condition $\sum_{\ell} \al_{\ell} = 1$, by $\sum_{\ell} \al_{\ell} = 0$.  This simply means that the vector $\vec{\al}$ automatically lies in the plane $\Pi$ and there is no need to project via $\pi$.  Note that the conditions $\sum_{\ell} \be_{\ell} = \sum_{\ell} \al_{\ell}\be_{\ell} = 0$ are still satisfied.  Then the construction of the configured star is identical.

Suppose now that the rows of $A$ are dependent, equivalently, the matrix $AA^t$ is no longer invertible.  From (\ref{det-AAt}), this is equivalent to $u_1{}^2+u_2{}^2=1$, that is, to $u_3=u_4= \cdots = u_N=0$.  Then (\ref{identity-1}) and (\ref{identity-2}) show that this is equivalent to the condition:
$$
n\sum_{\ell}|z_{\ell}|^2 + (\ga - 2)\Big|\sum_{\ell}z_{\ell}\Big|^2 = 0\,.
$$
Solutions $\{ z_1, \ldots , z_n\}$ to (\ref{gz}) which arise in this case are made explicit in Lemma \ref{lem:lin-comb}.  

We are now required to solve the system $AX={\bf 0}$, where ${\bf 0}$ is the $3\times (N-2)$ --matrix having zero in each of its entries, subject to the constraint $X^tX=\rho I_{N-2}$.  By Lemma \ref{lem:lin-comb}, we can choose a normalization for $\{ z_1, \ldots , z_n\}$ such that $\sum_{\ell} \al_{\ell}=0$ and $\be_{\ell} = \be$ is a non-zero constant.  We therefore have the system:
$$
\left( \begin{array}{ccc} \al_1 & \cdots & \al_n \\ \be & \cdots & \be \\ 1 & \cdots & 1 \end{array}\right) \left( \begin{array}{ccc} x_{13} & \cdots & x_{1N} \\ \vdots & \ddots & \vdots \\ x_{n3} & \cdots & x_{nN} \end{array} \right) = {\bf 0}\,.
$$
where now $\sum_{\ell} \al_{\ell} = 0$.  In particular, we require $\sum_{\ell} x_{\ell j} = 0$ for each $j = 3, \ldots , N-2$ and the constraint requires that the vectors $\vec{X}_j$ be orthogonal to $\vec{\al}$, to each other and have length $\sqrt{\rho}$.  We now proceed in a way similar to the non-degenerate case. 

Consider $\RR^n$ with coordinates $(t_1, \ldots , t_n)$.  Let $\Pi$ be the $(n-1)$-dimensional hyperplane $t_1+\cdots + t_n=0$ as above.  Then we are required to find $N-2$ vectors of length $\sqrt{\rho}$ in $\Pi$ orthogonal to $\vec{\al}$.  The dimension of the space orthogonal to $\vec{\al}$ in $\Pi$ is $n-2$, so since by hypothesis $N\leq n$, this can be achieved by Gram-Schmidt orthonormalization. 

The configured stars obtained by the above constructions are full provided the configuration matrix $W$ has maximal rank $N$.  This is the case if and only if the matrix $WW^t$ given in (\ref{system}) is invertible.  But a simple inductive argument shows that 
$$
|WW^t| = \left| \begin{array}{cccc}
\rho + \si u_1{}^2 & \si u_1u_2 & \cdots & \si u_1u_N \\
\si u_2u_1 & \rho + \si u_2{}^2 & \cdots & \si u_2u_N \\
 \vdots & \vdots & \ddots & \vdots \\
 \si u_Nu_1  & \si u_Nu_2 & \cdots & \rho + \si u_N{}^2 
 \end{array} \right| = \rho^N+\rho^{N-1}\si\,.
 $$
Recall that since $\ga <1$, by Lemma \ref{lem:ngaleq1}, we have $\rho >0$.  Since also $\rho + \si > 0$ (equation (\ref{si})), it follows that $|WW^t|>0$ and $WW^t$ is indeed invertible.
This completes the proof of the theorem.
\hfill $\Box$

\medskip

\begin{example} \label{ex:lift} {\rm We consider the graph on five vertices below, with solutions $\phi$ to (\ref{one}) normalized so as to take the value $0$ at the central vertex and $1$ on one of the other vertices.  The symmetry of the figure means that this determines the most general non-constant solution.

\medskip
\begin{center}
\setlength{\unitlength}{0.254mm}
\begin{picture}(102,112)(35,-136)
        \allinethickness{0.254mm}\path(40,-40)(40,-120) 
        \allinethickness{0.254mm}\path(40,-120)(120,-120) 
        \allinethickness{0.254mm}\path(120,-120)(120,-40) 
        \allinethickness{0.254mm}\path(40,-40)(120,-40) 
        \allinethickness{0.254mm}\path(120,-40)(40,-120) 
        \allinethickness{0.254mm}\path(120,-120)(40,-40) 
        \allinethickness{0.254mm}\special{sh 0.3}\put(40,-40){\ellipse{4}{4}} 
        \allinethickness{0.254mm}\special{sh 0.3}\put(120,-40){\ellipse{4}{4}} 
        \allinethickness{0.254mm}\special{sh 0.3}\put(80,-80){\ellipse{4}{4}} 
        \allinethickness{0.254mm}\special{sh 0.3}\put(40,-120){\ellipse{4}{4}} 
        \allinethickness{0.254mm}\special{sh 0.3}\put(120,-120){\ellipse{4}{4}} 
        \put(35,-136){\shortstack{$1$}} 
        \put(120,-136){\shortstack{$x$}} 
        \put(120,-33){\shortstack{$y$}} 
        \put(35,-33){\shortstack{$z$}} 
        \put(75,-71){\shortstack{$0$}} 
\end{picture}
\end{center}
\medskip
There are two solutions to (\ref{one}) with $\ga$ constant, namely:
$$
\begin{array}{ll}
\ga = 1/3\,; & \quad x = \pm \,\ii \,, \ y = - 1\,, \ z = \mp \,\ii \,; \\
\ga = 1\,; & \quad x = yz\,, \ y = \frac{1}{2} \pm \frac{\sqrt{3}}{2}\ii\,, \ z = 3 \pm 2\sqrt{2}\,.
\end{array}
$$

We reject the latter solution, since the inequality $\ga <1$ is violated.  Consider the solution with $\ga = 1/3$.  Let us construct the lift at the bottom left-hand vertex.  First, we normalize so that the solution takes on the value $0$ at this vertex:

\medskip
\begin{center}
\setlength{\unitlength}{0.254mm}
\begin{picture}(486,117)(20,-136)
        \allinethickness{0.254mm}\path(40,-40)(40,-120) 
        \allinethickness{0.254mm}\path(40,-120)(120,-120) 
        \allinethickness{0.254mm}\path(40,-120)(80,-80) 
        \allinethickness{0.254mm}\special{sh 0.3}\put(40,-120){\ellipse{4}{4}} 
        \allinethickness{0.254mm}\special{sh 0.3}\put(120,-120){\ellipse{4}{4}} 
        \allinethickness{0.254mm}\special{sh 0.3}\put(80,-80){\ellipse{4}{4}} 
        \allinethickness{0.254mm}\special{sh 0.3}\put(40,-40){\ellipse{4}{4}} 
        \put(35,-136){\shortstack{$1$}} 
        \put(85,-76){\shortstack{$0$}} 
        \put(115,-136){\shortstack{$i$}} 
        \put(25,-31){\shortstack{$-i$}} 
        \allinethickness{0.254mm}\path(175,-90)(230,-90)\special{sh 1}\path(230,-90)(224,-88)(224,-90)(224,-92)(230,-90) 
        \allinethickness{0.254mm}\path(280,-40)(280,-120) 
        \allinethickness{0.254mm}\path(280,-120)(360,-120) 
        \allinethickness{0.254mm}\path(280,-120)(320,-80) 
        \allinethickness{0.254mm}\special{sh 0.3}\put(280,-120){\ellipse{4}{4}} 
        \allinethickness{0.254mm}\special{sh 0.3}\put(320,-80){\ellipse{4}{4}} 
        \allinethickness{0.254mm}\special{sh 0.3}\put(280,-40){\ellipse{4}{4}} 
        \allinethickness{0.254mm}\special{sh 0.3}\put(360,-120){\ellipse{4}{4}} 
        \put(280,-136){\shortstack{$0$}} 
        \put(325,-76){\shortstack{$-1$}} 
        \put(345,-136){\shortstack{$-1 + i$}} 
        \put(255,-31){\shortstack{$-1 - i$}} 
\end{picture}
\end{center}
\medskip
The choice $n=N=3$ is determined and from (\ref{rho}) and (\ref{si}), we obtain $\rho = 2$ and $\si = 1$.  From (\ref{identity-1}) we find that $u_1{}^2+u_2{}^2=1$ so that $u_3=0$.  Then the $3\times 1$ --matrix $B$ vanishes and system (\ref{system-2}) has general solution $\vec{X}_3 = (2\la , -\la , - \la )$.  The constraint (\ref{system-2-constraint}) requires that $\inn{\vec{X}_3, \vec{X}_3}=\rho = 2$, so that $\la = \pm 1/\sqrt{3}$.  The star matrix $W$ (whose columns give the positions of the external star vertices) is given by:
$$
W = \left( \begin{array}{rrr} - 1 & -1 & - 1 \\
0 & 1 & -1 \\ \pm \frac{2}{\sqrt{3}} & \mp \frac{1}{\sqrt{3}} & \mp \frac{1}{\sqrt{3}} \end{array}
\right)
$$

We can proceed similarly with the central vertex of degree $4$.  Now we can choose $N = 3$ or $N = 4$.  In either case, $u_1=u_2=0$, $\rho = 2$ and $\si = 1$, so that, for $N=3$ we must have $u_3=\pm 1$.  Then
$$
A = \left( \begin{array}{rrrr} 1 & 0 & -1 & 0 \\ 0 & 1 & 0 & -1 \\ 1 & 1 & 1 & 1 \end{array} \right) , \qquad  B = \left( \begin{array}{c} 0 \\ 0 \\ 2\sqrt{3} \end{array} \right)\,
$$
and the unique (minimizing) solution is given by 
$$
Z = A^+B = \left( \begin{array}{c} \frac{\sqrt{3}}{2} \\ \frac{\sqrt{3}}{2} \\ \frac{\sqrt{3}}{2} \\ \frac{\sqrt{3}}{2} \end{array} \right)\,.
$$
The star matrix $W$ is given by
$$
W = \left( \begin{array}{rrrr} 1 & 0 & -1 & 0 \\ 0 & 1 & 0 & -1 \\ \frac{\sqrt{3}}{2} & \frac{\sqrt{3}}{2} & \frac{\sqrt{3}}{2} & \frac{\sqrt{3}}{2} \end{array} \right)\,,
$$
where the last line is only defined up to sign.
 
At the same vertex, we can also take $N=4$.  Then 
$$
B = \left(\begin{array}{cc} 0 & 0 \\ 0 & 0 \\ 2\sqrt{3}\, u_3 & 2\sqrt{3}\,u_4 \end{array} \right) \,.
 $$
 We proceed as in the proof of Theorem \ref{thm:lift}.  Take $u_3=1$ and $u_4=0$.  The solution is then given by the $(2\times 4)$--matrix $X = Z+Y$, where
 $$
 Z = A^+B = \left( \begin{array}{cc} \frac{\sqrt{3}}{2} & 0 \\
 \frac{\sqrt{3}}{2}  & 0 \\
\frac{\sqrt{3}}{2} & 0 \\
 \frac{\sqrt{3}}{2} & 0
 \end{array} \right) \,,
 $$
 and $Y = (\vec{0}, \vec{Y}_4)$, with $\vec{Y}_4\in \RR^4$ a vector of length $\sqrt{\rho} = \sqrt{2}$ in the plane $t_1+t_2+t_3+t_4=0$ orthogonal to both $\vec{\al} = (1,0,-1,0)$ and $\vec{\be} = (0,1,0,-1)$.  Up to sign, this is given by $\vec{Y}_4 = (1/\sqrt{2}, -1/\sqrt{2}, 1/\sqrt{2}, - 1/\sqrt{2})$.  We then obtain four possible lifts corresponding to the different choices of square root, with star matrix:
$$
W = \left( \begin{array}{rrrr}
1 & 0 & -1 & 0 \\
0 & 1 & 0 & -1 \\
\frac{\sqrt{3}}{2} & \frac{\sqrt{3}}{2} & \frac{\sqrt{3}}{2} & \frac{\sqrt{3}}{2} \\
\frac{1}{\sqrt{2}} & -\frac{1}{\sqrt{2}} & \frac{1}{\sqrt{2}} & -\frac{1}{\sqrt{2}}
\end{array} \right)\,.
$$
}
\end{example} 

\medskip

What additional information is required in order to make a unique choice for the lifted star?  This question is important when we come to consider \emph{edge curvature} in Section \ref{sec:curvature}.  For example, given a solution to (\ref{one}) corresponding to one of the regular polyhedra, we would like the lifted star to be exactly the one that arises from its canonical embedding in Euclidean space.  One way to do this is to define a notion of \emph{orientation} as follows.

In \cite{Ba-We}, a notion of orientation was considered on a regular graph of degree $n$, say, whereby the graph is endowed with an edge colouring of the $n$ colours $\{ 1, 2, \ldots , n\}$.  Thus each edge is coloured in such a way that no two edges of the same colour are incident at a vertex.  This enables one to uniquely label the edges at each vertex to give an ordering.  One could then attempt to apply a right-hand rule say, in order to make a choice of lift.  However, although this can be done with the solution corresponding to the framework of the tetrahedron in a way consistent with its embedding, it turns out to be impossible for the cube and the dodecahedron.  In the latter examples, any edge colouring with three colours leads to at least one of the two choices of lifted stars to be directed in the opposite way required.  We therefore proceed to define an orientation in terms of an edge colouring together with an $n$-form at each vertex of degree $n$.

\begin{definition} \label{def:orientation}  Let $\Ga = (V,E)$ be a graph with largest vertex degree equal to $M$.  Then an \emph{edge colouring} of $\Ga$ is an association of one of the colours $\{ 1, 2, \ldots , M, M+1\}$ to each edge so that no two same colours are incident at any vertex.  By a theorem of Vizing, any graph can be coloured with either $M$ or $M+1$ colours (see \cite{Di}).  We make the convention to choose the minimum $M$ colours when possible.  Given an edge colouring of $\Ga$, at each vertex, a \emph{volume form} is an alternating mapping $\theta$ of the edges which takes on the value $+1$ or $-1$.  Thus if $x$ is a vertex with $n$ incident edges $e_1, \ldots , e_n$ arranged so that \emph{colour}($e_j$)\,$<$\,\emph{colour}($e_k$) for $j<k$, then $\theta_x(e_1, \ldots , e_n) = \pm 1$ with $\theta_x(e_{\si (1)}, \ldots e_{\si (n)}) = {\rm sign}\, (\si )\theta (e_1, \ldots , e_n)$, for any permutation $\si$ of $\{ 1, \ldots , n\}$.  An \emph{orientation} of $\Ga$ is given by an edge colouring together with a volume form at each vertex.
\end{definition}

In order to apply this notion of orientation to make a choice of lifted star, in the notation of Theorem \ref{thm:lift} and its proof, we suppose that $N=3$ and the matrix $A$ of equation (\ref{system-2}) is of maximal rank $3$.  There is now either a unique lifted configured star in the case when $u_3=0$, or two choices if $u_3\neq0$ depending on the sign chosen for $u_3 = \pm \sqrt{u_1{}^2+u_2{}^2}$.  The star matrix is now given by
$$
W = \left( \begin{array}{cccc} \al_1 & \al_2 & \cdots & \al_n \\
\be_1 & \be_2 & \cdots & \be_n \\
x_{13} & x_{23} & \cdots & x_{n3} 
\end{array} \right)\,,
$$
where the last row is only defined up to sign.  Suppose that the vertex in question has an orientation according to Definition \ref{def:orientation}.  Without loss of generality, we can suppose that the edges are coloured with the colours $\{ 1, 2, \ldots , n\}$, in such a way that the external vertex $\vec{x}_{\ell}$ is joined to the internal vertex by the edge with colour $\ell$.  Suppose that the volume form satisfies $\theta (e_1, \ldots , e_n) = \ve$, where $\ve \in \{ +1, -1\}$.  Then provided the determinant of the $3\times 3$--minor given by the first three columns of $W$ is non-zero, we choose the sign of the third row so that 
$$
\left| \begin{array}{ccc} \al_1 & \al_2 & \al_3 \\ \be_1 & \be_2 & \be_3 \\ x_{13} & x_{23} & x_{33} 
\end{array} \right| = \ve \delta,
$$
where $\delta >0$.  If on the other hand this determinant vanishes, then we proceed in a lexicographic ordering, to choose next the minor formed from columns $1$, $2$ and $4$ and so on, until we encounter a non-zero determinant and apply the above rule. 

\begin{example} {\rm Consider the framework of a regular octagon with vertices placed at the points $(\pm 1,0,0)$, $(0,\pm 1, 0)$, $(0,0,\pm 1)$.  Then this can be edge-coloured as indicated.  Then there is a volume form which gives the lifts that correspond to the standard embedding in $\RR^3$.  However, in order to do this at the lateral vertices, we have to impose an additional condition that the star be regular.  This is because at these vertices $u_3=0$ and we do not satisfy the conditions of the discussion above.  The corresponding solution to (\ref{one}) has $\ga = 1/2$ and $\rho = \si = 2$.
\medskip
\begin{center}
\setlength{\unitlength}{0.200mm}
\begin{picture}(304,282)(35,-341)
        \allinethickness{0.254mm}\path(185,-75)(125,-245) 
        \allinethickness{0.254mm}\path(125,-245)(185,-335) 
        \allinethickness{0.254mm}\path(125,-245)(305,-195) 
        \allinethickness{0.254mm}\path(305,-195)(185,-75) 
        \allinethickness{0.254mm}\path(305,-195)(185,-335) 
        \allinethickness{0.254mm}\path(125,-245)(50,-195) 
        \allinethickness{0.254mm}\path(50,-195)(185,-75) 
        \allinethickness{0.254mm}\path(50,-195)(185,-335) 
        \allinethickness{0.254mm}\path(305,-195)(220,-150) 
        \allinethickness{0.254mm}\path(50,-195)(145,-170) 
        \allinethickness{0.254mm}\path(220,-150)(165,-165) 
        \allinethickness{0.254mm}\path(185,-75)(220,-150) 
        \allinethickness{0.254mm}\path(220,-150)(210,-210) 
        \allinethickness{0.254mm}\path(185,-335)(205,-230) 
        \put(180,-71){\shortstack{$x_0$}} 
        \put(35,-186){\shortstack{$x_1$}} 
        \put(140,-231){\shortstack{$x_2$}} 
        \put(310,-196){\shortstack{$x_3$}} 
        \put(225,-151){\shortstack{$x_4$}} 
        \put(195,-341){\shortstack{$x_5$}} 
        \put(155,-136){\shortstack{$1$}} 
        \put(115,-176){\shortstack{$1$}} 
        \put(110,-136){\shortstack{$4$}} 
        \put(195,-126){\shortstack{$3$}} 
        \put(230,-116){\shortstack{$2$}} 
        \put(250,-186){\shortstack{$4$}} 
        \put(220,-231){\shortstack{$3$}} 
        \put(205,-281){\shortstack{$2$}} 
        \put(250,-271){\shortstack{$1$}} 
        \put(100,-271){\shortstack{$3$}} 
        \put(90,-221){\shortstack{$2$}} 
        \put(155,-286){\shortstack{$4$}} 
\end{picture}
\end{center}
\medskip

Consider the vertex $x_0$ and define the volume form $\theta_{x_0}$ by $\theta_{x_0} (1234) = -1$ (for convenience, we write $\theta (1234)$ rather than $\theta (e_1, e_2, e_3, e_4)$).  Then with this edge-colouring, at this vertex $z_1=1,\, z_2=\ii ,\, z_3 = -1$ and $z_4 = -\ii$.  Thus $u_1=u_2=0$ and $u_3 = \pm 1$.  The solution to (\ref{system-2}) is given by $\vec{X}_3 = \pm (1,1,1,1)$.  In order to be consistent with the orientation, we must take the negative sign, to give the lifted star:
$$
W = \left( \begin{array}{rrrr} 1 & 0 & -1 & 0 \\ 0 & 1 & 0 & -1 \\ -1 & -1 & -1 & -1 \end{array} \right)
$$
whose sign of the determinant of the $3\times 3$--minor given by the first three columns is negative, which coincides with the sign of $\theta_{x_0}(1234)$. 

At the vertex $x_1$, we choose $\theta_{x_1}(1234)=+1$.  Then the edge-colouring dictates that $z_1=-1+\ii ,\, z_2=1+\ii ,\, z_3 = \ii , \, z_4 = \ii$, so that $u_1=0$, $u_2=1$ and $u_3=0$.  The solution to (\ref{system-2}) gives a $1$-parameter family of lifted stars:
$$ 
W = \left( \begin{array}{cccc} 1 & 0 & -1 & 0 \\ 0 & 1 & 0 & -1 \\ -\frac{\cos t}{\sqrt{2}} & -\frac{\cos t}{\sqrt{2}} & \frac{\cos t}{\sqrt{2}} + \sin t & \frac{\cos t}{\sqrt{2}} - \sin t \end{array} \right)\,.
$$
If we now impose the condition that the lift must be a regular star, then there are just two solutions given by $t = \pi /2$ or $t = 3\pi /2$.  The choice $t=3\pi /2$ is required in order that the determinant of the $3\times 3$--minor consisting of the first three columns be positive, to coincide with the sign of $\theta_{x_1}(1234)$.  This gives the lift that coincides with the canonical embedding of the octahedron.  We proceed similarly with the other vertices, defining the appropriate volume form, with the proviso that the the stars at the lateral vertices be regular.
}
\end{example}  

\medskip

\noindent \emph{Distance}:  The above analysis enables us to define edge-length and so distance on a graph $\Ga = (V,E)$ admitting a solution to (\ref{one}), provided that at each vertex we have $\ga <1$.  For this, note that if we reverse the order of multiplication of $W$ and $W^t$, we obtain
$$
W^tW = \left( \begin{array}{cccc}
||\vec{x}_1||^2 & \inn{\vec{x}_1,\vec{x}_2} & \cdots  & \inn{\vec{x}_1,\vec{x}_n} \\
\inn{\vec{x}_1,\vec{x}_2} & ||\vec{x}_2||^2 & \cdots  & \inn{\vec{x}_2, \vec{x}_n} \\
 \vdots & \vdots & \ddots   & \vdots \\
\inn{\vec{x}_1,\vec{x}_n} & \inn{\vec{x}_2,\vec{x}_n} & \cdots  & ||\vec{x}_n||^2 
\end{array}
\right)\,,
$$  
 where $\inn{\vec{x}_j,\vec{x}_k}$ denotes the standard Euclidean inner product of $\vec{x}_i$ and $\vec{x}_j$.  But then
 \begin{eqnarray*}
 \sum_{\ell}||\vec{x}_{\ell}||^2 & = & {\rm trace}\, W^tW = {\rm trace}\, WW^t \\
   & = & N\rho + \si ||\vec{u}||^2 = N\rho + \si\,.
   \end{eqnarray*}
 Now this latter quantity can be expressed in terms of $\ga$ and $z_{\ell}$ from (\ref{rho}) and the relation $\ga = \si / (\si + \rho )$ to give the mean of the values $||\vec{x}_{\ell}||$:
\begin{equation} \label{mean-distance}  
 \frac{1}{n}\sum_{\ell}||\vec{x}_{\ell}||^2 = \frac {\big(N+(1-N)\ga\big)}{n(1-\ga )}\rho\,.
\end{equation}
This equation expresses the mean length of the edges of a virtual configured star in $\RR^N$ whose external vertices $\vec{x}_{\ell}$ project to $z_{\ell}$.  This motivates our definition of edge length in a graph.  

Let $\Ga = (V,E)$ be a graph coupled to a solution $\phi : V \ra \CC$ to equation (\ref{one}).  For each $x\in V$, set 
$$
\rho (x) = \frac{1}{2}\left\{\sum_{y\sim x} |\phi (y)-\phi(x)|^2  - \frac{\ga (x)}{n(x)} \Big| \sum_{y\sim x}(\phi (y)-\phi (x))\Big|^2\right\}\,,
$$
where $n(x)$ is the degree of $\Ga$ at $x$.

\begin{definition} \label{def:edge-length}  If $x\in V$ is a vertex of degree $n(x)$ such that $\ga (x)<1$, then we define the \emph{median edge length at $x$ relative to $\phi$} to be the quantity $r(x)$ whose square is given by
$$
r(x)^2 = \frac {\big[N+(1-N)\ga (x)\big]}{n(x)[1-\ga (x)]}\,\rho (x)\,.
$$
If $\ov{xy}\in E$ is an edge which joins $x$ to $y$ such that both $\ga (x)<1$ and $\ga (y) < 1$, then we define the \emph{length of $\ov{xy}$ relative to $\phi$} to be the mean $\ell (\ov{xy})$ of the median edge lengths at $x$ and $y$:
$$
\ell (\ov{xy}) = \frac{r(x)+r(y)}{2}\,.
$$
\end{definition}
As emphasized in the above definition, the lengths so defined are \emph{relative} to the solution $\phi$ of (\ref{one}), which is only defined up to $\phi \mapsto \la \phi + \mu$ for $\la , \mu \in \CC$.  This means that the only meaningful quantities are \emph{relative} lengths, say $\ell (e) / \ell (f)$, for two edges $e,f\in E$.  This is consistent with our relational interpretation of physical quantities as discussed in the Introduction.  In particular, if both $n$ and $\ga$ are constant on the graph, we may take the quantity $2\rho$ defined by (\ref{rho}) as a measure of median edge length at each vertex:
$$
r(x)^2 = 2\rho =  \sum_{y\sim x} |\phi (y)-\phi(x)|^2  - \frac{\ga}{n} \Big| \sum_{y\sim x}(\phi (y)-\phi (x))\Big|^2\,.
$$

We can define an absolute length by normalizing as follows.  Let $\Ga = (V,E)$ be a graph coupled to a solution $\phi : V \ra \CC$ to (\ref{one}).  Then as in Appendix \ref{sec:lin}, define the \emph{square $L^2$-norm of the derivative of $\phi$} to be the quantity:
$$
||\dd \phi ||^2 = \sum_{\ov{xy} \in E} |\dd\phi (\ov{xy})|^2 = \frac{1}{2}\sum_{x,y \in V, x\sim y} |\phi (y)-\phi (x)|^2\,,
$$   
where $\dd\phi (\ov{xy}) = \phi (y)-\phi (x)$ is the discrete derivative with respect to some orientation of the edge $\ov{xy}$ (in this case $x$ is the initial vertex and $y$ the end vertex).  
\begin{definition} \label{def:abs-edge-length}
Let $\Ga = (V,E)$ be a graph coupled to a solution $\phi : V \ra \CC$ to equation {\rm (\ref{one})}.  If $x\in V$ is a vertex such that $\ga (x)<1$, then we define the \emph{absolute median edge length at $x$ relative to $\phi$} to be the quantity $r_{\rm abs}(x)$ whose square is given by
$$
r_{\rm abs}(x)^2 = \frac{r(x)^2}{||\dd\phi ||^2}\,,
$$
where $r(x)$ is the median edge length at $x$ relative to $\phi$.
If $e\in E$ is an edge joining $x$ to $y$ such that both $\ga (x)<1$ and $\ga (y) < 1$, then we define the \emph{absolute length of $e$ relative to $\phi$} to be the mean $\ell_{\rm abs} (e)$ of the absolute median edge lengths at $x$ and $y$:
$$
\ell_{\rm abs}(e) = \frac{r_{\rm abs}(x)+r_{\rm abs}(y)}{2}\,.
$$
\end{definition}
Then both the quantities $r_{\rm abs}(x)$ and $\ell_{\rm abs}(x)$ are independent of the freedom $\phi \mapsto \la \phi + \mu$ ($\la , \mu \in \CC$).  

The median edge length of Definition \ref{def:edge-length} is defined so as to give the length of the edges of a corresponding regular star in $\RR^N$, when such exists.  In particular, if $\Ga = (V,E)$ is the $1$-skeleton of a regular polytope and $\phi :V\ra \CC$ associates to each vertex its value after an orthogonal projection, then the edge-length at each vertex coincides with the lengths of the edges of the regular polytope. More generally, we can interpret the edge length at each vertex as the length of the edges of the ``best fit" polytope at that vertex.  The median edge length then gives the average length at two adjacent vertices.

\begin{example} {\rm 
If we return to Example \ref{ex:lift} and consider the solution corresponding to $\ga = 1/3$, then the (unique) common dimension to define edge length is $N=3$.  At the central vertex the edge length is $\sqrt{7}/2$ and at any of the other vertices, it is $\sqrt{7/3}$.  Thus the median edge length of the edge joining the central vertex to one of the other vertices is $\frac{(\sqrt{7}/2)+\sqrt{7/3}}{2}$, whereas the median edge length of one of the outside edges is $\sqrt{7/3}$.  So, for example, the shortest path joing $x$ to $z$ is given by passing through the central vertex.  Note that, as already remarked, the edge lengths are only defined up to multiple and so only relative edge lengths have meaning.
}
\end{example}

A question we now consider, is whether the notion of distance, either relative or absolute, that we have defined above, endows a graph with the structure of a \emph{path metric space} in the sense of M. Gromov \cite{Gr}.  We first of all note a triangle inequality around complete subgraphs on three vertices. 

Given a function $\phi : V \ra \CC$ and a vertex $x\in V$, we say that \emph{$\phi$ is constant on the star centred on $x$} if the restriction of $\phi$ to $x$ and its neighbours $y\sim x$, is constant. 

\begin{proposition} \label{prop:triangle-ineq}  {\rm (Local triangle inequality)} Let $\Ga = (V,E)$ be a graph coupled to a solution $\phi : V \ra \CC$ to equation {\rm (\ref{one})}.  Suppose $x,y,z\in V$ are three vertices of a complete subgraph: $x\sim y,\, y\sim z,\, z\sim x$, such that the inequality $\ga <1$ is satisfied at each vertex.   Then the triangle inequality is satisfied:
$$
\ell (\ov{xy}) + \ell (\ov{xz}) \geq \ell (\ov{yz})\,.
$$
If further $\phi$ is non-constant on the star centred on $x$, then the inequality is strict.
\end{proposition}

\begin{proof}  This is an immediate consequence of the definition.  Specifically,
$$
\ell (\ov{xy}) + \ell (\ov{xz}) = \frac{1}{2}(r(x) + r(y)) + \frac{1}{2}(r(x)+r(z)) = \ell (\ov{yz}) + r(x) \geq \ell (\ov{yz})\,,
$$
since because of the inequality $\ga (x) < 1$, we have $r(x) \geq 0$.  If further, $\phi$ is non-constant on the star centred on $x$, then $r(x)>0$ and the inequality is strict.  
\end{proof} 

In spite of this local triangle inequality, we may encounter a difficulty in trying to endow a graph coupled to a solution $\phi$ to (\ref{one}) with a metric space structure.  This may arise when, for a given vertex $x$, the function $\phi$ is constant on the star centered on $x$, as well as on the star centred on one of its neighbours $y$.  Then $\ell (\ov{xy}) = 0$.  We can either agree to allow distinct points to have zero distance between them, and so consider rather a \emph{pseudo-metric space structure}, or we can avoid this situation by introducting a notion of \emph{collapsing}.  This is a concept we will return to in Section \ref{sec:particles}.

\begin{definition} \label{def:collapse}  Let $(\Ga , \phi )$ be a graph coupled to a solution to equation {\rm (\ref{one})}.  Then we \emph{collapse} $\Ga$ to a new graph $\wt{\Ga}$ by removing all edges that connect vertices at which $\phi$ takes on identical values; then remove all isolated vertices.
\end{definition}

It is clear that after collapse, if we let $\wt{\phi}$ denote the restriction of $\phi$ to $\wt{\Ga}$, then $\wt{\phi}$ also satisfies (\ref{one}) with $\wt{\ga} = \wt{n}\ga /n$ where $\wt{n}$ is the new degree at each vertex.  Indeed, if we check at a vertex $x$, then if $y$ is a neighbour at which $\phi (y) = \phi (x)$, then since only the difference $\phi (y) - \phi (x)$ occurs in (\ref{one}), removing the edge $\ov{xy}$ only affects the degree.  However, it is to be noted that collapsing may disconnect a graph.

Let $(\Ga , \phi )$ be a graph coupled to a solution to equation (\ref{one}).  Then given a path $\ov{c}:= \ov{x_0x_1x_2\cdots x_p}$ joining $x$ to $y$ (so we have $x=x_0$, $y=x_p$ and $x_j\sim x_{j+1}$ for all $j = 0, \ldots , p-1$), then we define the length $\ell (\ov{c})$ to be the sum:
$$
\sum_{j=0}^{p-1} \ell (\ov{x_jx_{j+1}})\,.
$$
We can now define the distance between two vertices to be the infimum of the lengths of all paths joining the two vertices.  Then provided $\Ga$ is collapsed with respect to $\phi$, it is clear that this notion of distance endows $\Ga$ with the structure of a path metric space.

\section{Curvature} \label{sec:curvature}
Our introduction of curvature on a graph is based on Theorem \ref{thm:lift} and Corollary \ref{cor:configured-star}.  Thus, we consider a graph $\Ga = (V,E)$ coupled to a solution $\phi$ to equation (\ref{one}): $\ga_{\phi}(\Delta \phi )^2 = \dd \phi^2$.  At each vertex $y\in V$, we measure the (in general solid) angular deficit $\delta (y)$ as determined by a regular star placed in $\RR^N$ whose central vertex $x_0$ has the same degree as $y$ in $\Ga$, and which also solves (\ref{one}) at $x_0$ with $\ga_{{\rm star}}(x_0) = \ga_{\phi}(y)$.

Depending on the degree $n$ and the value of $\ga$, $\delta (y)$ will be well-defined.  However, in some situations, there may be different possibilities for the dimension $N$, leading to different possible values for the curvature.  The restriction $\ga\leq 1$, will be a necessary condition.

The aim is to suppose $\ga$ is part of the geometric spectrum (see Section \ref{sec:spectrum}), so that both dimension and curvature arise from purely combinatorial properties of the graph (independent of $\phi$).  For some graphs, these will be uniquely defined.  This is the case for the bipartite graph $K_{33}$, for example, whose geometric spectrum contains the unique value $\ga = 1$ (see Section \ref{sec:spectrum}) and whose vertex degree dictates that it ``lives in " dimension three.

In what follows, we discuss \emph{convex} polytopes, which are by definition, the closed intersection of half-spaces (whether this be in Euclidean space, or in spherical space).  In the case when a polytope is regular (convex or not), its vertices all lie on a sphere called the \emph{circum-sphere} \cite{Co}.  It is useful to use \emph{absolute angle measure} when measuring solid angles (see \cite{Sh, Gr-Sh}).  We will write $H^M(\La )$ for the $M$-dimensional Hausdorff measure of a set $\La$ in these units.  Then, in any dimension, the angle is measured as a fraction of the total angle subtended by a sphere centred at the point in question.  Thus in two dimensions, a right-angle has value $1/4$, whereas in three dimensions, the angle subtended by the vertex figure of a cube has value $1/8$.  Equivalently, $H^2(\La )=1/8$, where $\La \subset S^2$ represents one eighth portion: $x,y,z>0$, of the sphere $x^2+y^2+z^2=1$.  

\begin{definition} \label{def:N-curvature}  {\rm ($N$-dimensional vertex-curvature)} Let $(\Ga , \phi )$ be a pair consisting of a graph $\Ga$ coupled to a solution $\phi$ to {\rm (\ref{one})}.  Let $y$ be a vertex of $\Ga$ and let $n$ be the degree of \, $\Ga$ at $y$.  Then the $N$-dimensional vertex-curvature at $y$ is defined provided $\ga (y)\leq 1$ and there is a regular $N$-polytope with vertex figure $P$ (a regular $(N-1)$-polytope) having $n$ vertices.  Let $P$ be centred on $\vec{0}\in \RR^{N-1}$ with vertices ${\vec v}_1, \ldots , \vec{v}_n$ lying on its circum-sphere of radius $r>0$.  For $\ga <1$, let $U$ be the corresponding configuration matrix with associated constant $\rho >0$ (see {\rm (\ref{star-standard})} and {\rm (\ref{star-conditions})}).  Let 
$$
\vec{x}_{\ell} = \frac{1}{\sqrt{\rho + nr^2(1-\ga )}}\left( \begin{array}{c} \sqrt{n(1-\ga )}\, \vec{v}_{\ell} \\ \sqrt{\rho}\end{array} \right) \in S^{N-1} \quad (\ell = 1, \ldots , n)\,,
$$
be the corresponding vertices of a regular star in $\RR^N$ centred on $\vec{0}$.  Let $\La$ be the convex hull in $S^{N-1}$ of the set $\{ \vec{x}_1, \ldots , \vec{x}_n\} \subset S^{N-1}$.  We define the \emph{$N$-dimensional vertex-curvature of $(\Ga , \phi )$ at $y$} to be the deficit:
$$
\delta_{\ga} (y) = 1 - H^{N-2}(\pa \La )\,,
$$
in absolute angle measure.  In the case when $\ga =1$, then we define the $N$-dimensional vertex-curvature to be the limit $\lim_{\ga \ra 1^-} \delta_{\ga}(y)$, when this exists.
\end{definition}

Note that in the above definition, if we rescale the vertices $\vec{v}_{\ell}$ by $\vec{v}_{\ell} \mapsto \vec{v}_{\ell}^{\ \sim} = \la v^{\ell}$, say, then $r\mapsto r^{\sim} = \la r$, $\rho \mapsto \rho^{\sim} = \la^2\rho$ and both $\vec{x}_j$ and the curvature $\delta_{\ga}(y)$ are well-defined and independent of this scaling.  We write \emph{vertex-curvature} to distinguish it from \emph{edge-curvature} which we define later, but if the context is clear, we shall just write $N$-curvature.

We first of all justify this definition and put it into the context of classical work on topological invariants of polytopes.  In particular, we refine the definition for $3$- and $4$-curvature.  We begin by reviewing a generalization of a theorem of Descartes, by Shephard \cite{Sh}, Gr\"unbaum and Shephard \cite{Gr-Sh} and Ehrensborg \cite{Eh}.  The context is that of elementary polytopes.  

A family $\{ F_1, \ldots , F_r\}$ of $(N-1)$-dimensional convex polytopes in $\RR^N$ form an \emph{elementary polytope} $P$ of dimension $N$ if: (i)  for all $j,k$, $F_j\cap F_k$ is either empty or a face of each of $F_j$ and $F_k$; (ii)  $\cup_{j}F_j$ is an $(N-1)$-dimensional manifold.  

Given an elementary polytope $P\subset \RR^N$, denote by $\Pp$ the \emph{face decomposition of} $P$; thus $\Pp$ is the collection of all (open) faces of all dimension, consisting of the vertices, edges, ... , $(N-1)$-faces, $N$-faces.  For $\vec{x}\in \RR^N,\, \vec{w}\in S^{N-1}$, following Ehrensborg \cite{Eh}, we define the quantity $R(\vec{x},\vec{w}):= \lim_{s\ra 0^+}{\bf 1}_P(\vec{x}+s\cdot \vec{w})$, where ${\bf 1}_P$ is the characteristic function of $P$.  Note that this takes on the value $0$ or $1$.  Then given a face $F\in \Pp$, we have $R(\vec{x},\vec{w})= R(\vec{y},\vec{w})$ for all $\vec{x},\vec{y}\in F$; write $R(F,\vec{w})$ for this and define $\La_F=\{ \vec{w}\in S^{N-1}: R(F,\vec{w}) = 1\}$. 

Let $S(\La_F) = H^{N-2}(\pa \La_F)$ and let $\si_{N} = H^N\, (S^{N})$, so that in absolute angle measure, $\si_{N} = 1$.  Let $F$ be a face of $P$ of dimension $\leq N-3$.  We define the \emph{deficiency at $F$} to be the quantity:
$$
\delta (F):= \si_{N-2} - S(\La_F)\,.
$$
Note that if $Q\subset S^{N-1}$ is spherically convex (that is, it is the intersection of hemispheres), then $S(Q) = H^{N-2} (\pa Q)$ is proportional to the Haar measure of all the great circles which intersect $Q$.  The following theorem generalizes a classical result of Descartes.

\begin{theorem} \label{thm:descartes} {\rm \cite{Sh, Gr-Sh, Eh}}  Let $P$ be an elementary polytope with face decomposition $\Pp$ such that $P$ has only one $N$-dimensional face $P^0$.  Then
$$
\sum_{F\in \Pp , \ \dim F\leq N-3 } \ve (F)\delta (F) = \si_{N-2}[(-1)^N-1]\ve (P^0)\,,
$$
where $\ve (F)$ denotes the Euler characteristic of $F$ given by $(-1)^k$ when $F$ is of dimension $k$.
\end{theorem}

In the case when $N=3$ and $P$ is a convex polyhedron (now using radians for our measure), we obtain the classical theorem of Descartes:
$$
\sum_{\ell} \delta (\vec{v}_{\ell}) = 4\pi\,,
$$  
 where the sum is taken over the vertices of $P$.  Here, the deficiency $\delta (\vec{v}_{\ell})$ is the sum of the internal angles at $\vec{v}_{\ell}$ of the faces which contain $\vec{v}_{\ell}$.  We may view this theorem as a discrete version of the Gauss-Bonnet Theorem for surfaces in the smooth setting.

Let us now consider the different dimensional curvatures that arise from Definition \ref{def:N-curvature}.  By convention, at a vertex of degree $1$, we assign the curvature $\delta = 1$.  It is straightforward to see that at a vertex of degree $2$, the $2$-curvature just measures the exterior angle in absolute angle measure.  In view of identity (\ref{spec-value}), we can state this as follows.

\begin{proposition} \label{prop:2-curvature}  Let $(\Ga , \phi )$ be a pair consisting of a graph $\Ga$ coupled to a solution $\phi$ to {\rm (\ref{one})}.  Let $y$ be a vertex of degree $2$ where $\ga \leq 1$, if such exists.  Then the $2$-dimensional vertex-curvature of $(\Ga , \phi )$ at $y$ is given by the quantity:
$$
\delta_{\ga}(y) = \frac{1}{2\pi} \arccos\left( \frac{\ga}{\ga - 2}\right)\,.
$$
\end{proposition}

Note that $\lim_{\ga \ra 1^-}\delta_{\ga}(y)$ is well-defined and equals $1/2$.  For example, if $\Ga$ is a cyclic graph of even order $2k$ and $\phi$ is a function taking on alternate values at neighbouring vertices.  Then $\phi$ satisfies (\ref{one}) with $\ga = 1$.  The total curvature is then given by $k$.  If $\Ga$ is a regular polygon in the plane and $\phi$ the corresponding position function, then the total $2$-curvature is equal to $1$, or in radians, to $2\pi$, as required. 

We now proceed to higher dimensional curvature; dimension $3$ is of particular interest because of the minimizing property of the solution to the lifting problem that occurs in this case; see Theorem \ref{thm:lift}.

\begin{proposition} \label{prop:3-curvature} {\rm ($3$-dimensional vertex-curvature)} Let $(\Ga , \phi )$ be a pair consisting of a graph $\Ga$ coupled to a solution $\phi$ to {\rm (\ref{one})}.  Let $y$ be a vertex of $\Ga$ and let $n$ be the degree of $\Ga$ at $y$.  Suppose that $n \in \{ 3,4,5\}$ and that $\ga \leq 1$.  Then the $3$-dimensional vertex-curvature of $(\Ga , \phi )$ at $y$ is given by the quantity: 
$$
\delta_{\ga}(y) = 1 - \frac{n}{2\pi}\arccos\left\{ \frac{1+2(1-\ga )\cos \frac{2\pi}{n}}{3-2\ga}\right\}\,.
$$
\end{proposition}
\begin{proof}
The configuration of vertices $\vec{v}_{\ell}$ is given by (\ref{star-R2}), that is $\vec{v}_{\ell} = e^{2\pi \ii \ell / n}$ ($\ell = 1, \ldots , n$).  The boundary of the convex hull $\La$ of the set $\{ \vec{x}_1, \ldots , \vec{x}_n\}$ in $S^2$ is made up of arcs of great circles of length $\al = \arccos (\vec{x}_1\cdot \vec{x}_2)$.  In absolute angle measure, the deficit, or $3$-curvature, is given by $1 - \frac{n}{2\pi}\al$.  Substitution of the expressions for $\vec{x}_{\ell}$ given by Definition \ref{def:N-curvature} gives the required formula.
\end{proof}

Our requirement that $n\in \{ 3,4,5\}$ is a consequence of Definition \ref{def:N-curvature}, for the only polygons that appear as vertex figures of regular polyhedra have these possibilities for their number of sides.  Of course, the expression for $\delta_{\ga}(y)$ above is defined for any $n$ provided $\ga \leq 1$.

Note that $\lim_{\ga \ra 1^-}\delta_{\ga}(y)$ is well-defined and equals $1$.  For example, the bipartite graph $K_{33}$ has the unique value $\ga = 1$ in its geometric spectrum.  Also the degree of each vertex is $n = 3$.  The $3$-curvature at each vertex is then $\delta = 1$ and the total curvature is given by $\sum_{y\in V}\delta (y) = 6$.  On the other hand, the $1$-skeleton of a tetrahedron in $\RR^3$ has $n = 3$ and $\ga = 3/4$, so that $\ga <1$.  It is easily checked that the $3$-curvature at each vertex is given by $\delta = 1/2$, giving a total curvature of $2$.  To obtain the total curvature in radian measure, we multiply by $2\pi$ to give the value $4\pi$, which confirms the theorem of Descartes.

For degree $3$, the $3$-curvature is the only $N$-curvature that can apply, since there is no other vertex figure with three vertices.  This then gives a well-defined curvature for vertices of degree $3$, provided $\ga \leq 1$.

\begin{example} \label{ex:curvature-double-cone} {\rm  Consider the double cone on the triangle discussed in Section \ref{sec:inv-polytopes}.  The corresponding framework satisfies equation (\ref{one}) with $\ga = 4/5$ on the three lateral vertices of degree $4$ and $\ga = 1/3$ at the two apexes of degree $3$.  Then we calculate:
$$
\delta_{\rm apex} = 1- \frac{3}{2\pi}\arccos \frac{1}{7} \qquad {\rm and} \qquad \delta_{\rm lat} = 1 - \frac{4}{2\pi} \arccos \frac{5}{7}\,,
$$
to give total curvature:
$$
\delta_{\rm tot} = 3\delta_{\rm lat} + 2\delta_{\rm apex} = 5 - \frac{3}{\pi}\left( \arccos\frac{1}{7} + 2\arccos\frac{5}{7}\right)\,.
$$
Now
$$
\cos\left( \arccos\frac{1}{7} + 2\arccos\frac{5}{7}\right) = \frac{1}{7^3}(1-240\sqrt{2}) = -0\cdot 98662\,.
$$
to five decimal places, so that
$$
 \arccos\frac{1}{7} + 2\arccos\frac{5}{7} \sim \pi\,.
 $$
 is close, but not equal to $\pi$.  The exact value $\pi$ gives a total curvature of $2$, which is the value that we expect from the theorem of Descartes.  The small difference arises due to the fact that the vertex figures at the lateral vertices are not \emph{configured} stars, whereas the curvature is defined in terms of the deficit that occurs for the unique lifted configured star.  That is, we try to fit a \emph{regular} polytope in the best way possible.  What is remarkable, is how close the two values are.  
}
\end{example}

The above example illustrates one of the problems in defining the curvature.  The advantage of lifting to a \emph{configured} star is that, in dimension $N=3$, the lift is unique and so the curvature is uniquely defined.  However, any expression of the total curvature as an invariant quantity would need to involve some approximation.

For the $4$-curvature, there are some special cases to consider.  We list these in the proposition below. 

\begin{proposition} \label{prop:4-curvature} {\rm ($4$-dimensional vertex-curvature)} Let $(\Ga , \phi )$ be a pair consisting of a graph $\Ga$ coupled to a solution $\phi$ to {\rm (\ref{one})}.  Let $y$ be a vertex of $\Ga$ and let $n$ be the degree of $\Ga$ at $y$.  Suppose that $n \in \{4,6,12,20 \}$ and that $\ga \leq 1$.  Then depending on the degree, the $4$-dimensional vertex-curvature of $(\Ga , \phi )$ at $y$ is given by one of the expressions below:
\begin{center}
\begin{tabular}{ccc}
degree & vertex figure & curvature  \\ \hline
$4$ & tetrahedron & $\ds 2-\frac{3}{\pi}\arccos \left(\frac{\ga}{4-2\ga}\right)$ \\
$6$ & octahedron & $ \ds 3 - \frac{6}{\pi}\arccos \left( \frac{1}{5-3 \ga}\right)$ \\
$12$ & icosahedron & \quad $\ds 6-\frac{15}{\pi}\arccos\left(\frac{6(\sqrt{5}+1)\ga - 11 - 7\sqrt{5}}{2[6(\sqrt{5}+3)\ga - 23 - 7\sqrt{5}]}\right)$ \\
$20$ & dodecahedron & $\ds 10-\frac{15}{\pi}\arccos\left(\frac{5\ga -1-2\sqrt{5}}{2[-5\ga +8-\sqrt{5}]}\right)$ \\ \hline
\end{tabular}
\end{center}
\medskip
\end{proposition}

\begin{proof}
Given two vectors $\vec{u}, \vec{v} \in S^M(r)$ in a sphere of radius $r$, the arc of the great circle joining $\vec{u}$ to $\vec{v}$ is given by
\begin{eqnarray*}
\theta \mapsto & \ds \left( \cos\theta - \frac{(\vec{u}\cdot\vec{v})\sin\theta}{\sqrt{r^4-(\vec{u}\cdot \vec{v})^2}}\right)\vec{u} + \frac{r^2\sin\theta}{\sqrt{r^4-(\vec{u}\cdot \vec{v})^2}}\,\vec{v} \\ & \qquad \qquad  \qquad \qquad (0\leq \theta \leq \arcsin\frac{\sqrt{r^4-(\vec{u}\cdot \vec{v})^2}}{r^2} )\,.
\end{eqnarray*}
When $r=1$, this is unit speed.  Furthermore, the tangent vector to this arc at $\vec{v}$ is given by
\begin{equation} \label{tangent}
\vec{t} = \frac{1}{\sqrt{r^4-(\vec{u}\cdot \vec{v})^2}}(-r^2\vec{u} + (\vec{u}\cdot\vec{v})\vec{v})\,.
\end{equation}
The area of a spherical polygon with $m$ sides and with interior angles $\theta_k$ ($k = 1, \ldots , m$) is given by 
$$
KA = \sum_k\theta_k - (m-2)\pi\,,
$$
where $K$ is the curvature of the sphere.  We are required to calculate the spherical surface area of the boundaries of the various vertex figures in $S^3$.  These are made up of faces lying in great $2$-spheres which are either triangles, or in the case of the dodecahedron, pentagons, whose edges are arcs of great circles.  In order to calculate the interior angles, we calculate the scalar product of the unit tangents to these edges at a vertex.  By symmetry, any vertex will do.  We calculate this for the icosahedron and the dodecahedron, the other cases being similar.

For the dodecahedron, a configuration of vertices is given by (\ref{vertices-dodecahedron}).  With the notation of Definition \ref{def:N-curvature}, three consecutive vertices around one pentagonal face are given by
$$
\vec{x}_1=\frac{1}{\sqrt{8+3a^2}}\left(\begin{array}{c} a\la^{-1} \\ a\la \\ 0  \\ \sqrt{8} \end{array} \right) , \, 
\vec{x}_2=\frac{1}{\sqrt{8+3a^2}}\left(\begin{array}{c} a \\ a \\ a  \\ \sqrt{8} \end{array} \right) , \,
\vec{x}_3=\frac{1}{\sqrt{8+3a^2}}\left(\begin{array}{c} 0 \\ a\la^{-1} \\ a\la \\ \sqrt{8} \end{array} \right) ,
$$
where $a = \sqrt{n(1-\ga )}$ and $\la = (1+\sqrt{5})/2$.  These three vertices determine a great $2$-sphere in $S^3$ which contains the pentagonal face.  With this arrangement, $\vec{x}_2$ is the central vertex joined to $\vec{x}_1$ and $\vec{x}_3$ by arcs of great circles.  In order to calculate the interior angle at each vertex of the pentagon, we calculate the tangent to each of these arcs at $\vec{x}_2$.  To do this we apply (\ref{tangent}).  For the first arc we set $\vec{u} = \vec{x}_1$ and $\vec{v} = \vec{x}_2$, to obtain the tangent vector:
$$
\vec{t}_1 = \frac{1}{2\sqrt{3a^2+8}\sqrt{a^2+12-4\sqrt{5}}}\left( \begin{array}{c} \frac{a^2}{2}(3-\sqrt{5})+4(3-\sqrt{5}) \\ - \frac{a^2}{2}(3+\sqrt{5})+4(1-\sqrt{5}) \\
\sqrt{5}a^2+8 \\ -a\sqrt{8}(3-\sqrt{5}) \end{array} \right) \,.
$$
For the second arc, we set $\vec{u} = \vec{x}_3$ and $\vec{v} = \vec{x}_2$, to obtain:
$$
\vec{t}_2 = \frac{1}{2\sqrt{3a^2+8}\sqrt{a^2+12-4\sqrt{5}}}\left( \begin{array}{c} \sqrt{5}a^2+8 \\ \frac{a^2}{2}(3-\sqrt{5})+4(3-\sqrt{5}) \\ - \frac{a^2}{2}(3+\sqrt{5})+4(1-\sqrt{5}) \\
 -a\sqrt{8}(3-\sqrt{5}) \end{array} \right) \,.
$$
Then
$$
\vec{t}_1\cdot \vec{t}_2 = \frac{-a^2+16-8\sqrt{5}}{2(a^2+12-4\sqrt{5})}\,,
$$
which gives the cosine of the interior angle (it is indeed the interior angle, being greater than $\pi /2$).  Then the area (in absolute angle measure) of each pentagonal face is given by
$$
\frac{1}{4\pi}\left\{ 5\arccos \left(\frac{-a^2+16-8\sqrt{5}}{2(a^2+12-4\sqrt{5})}\right) - 3\pi\right\} \,,
$$
so that the surface area of the spherical dodecahedron is given by twelve times this quantity.  We then obtain the angular deficiency, or $4$-curvature:
$$
\delta = 10 - \frac{15}{\pi} \arccos \left(\frac{-a^2+16-8\sqrt{5}}{2(a^2+12-4\sqrt{5})}\right)\,.
$$ 
On substituting the value of $a$, we obtain the required formula.

For the icosahedron, a configuration of vertices is given by (\ref{vertices-icosahedron}).  Three vertices which form one of the triangular faces are given by:
$$
\vec{v}_1=\left( \begin{array}{c} 0 \\ 1 \\ \la \end{array}\right) , \quad 
\vec{v}_2=\left( \begin{array}{c} 1 \\ \la \\ 0 \end{array}\right) , \quad
\vec{v}_3=\left( \begin{array}{c} \la \\ 0 \\ 1 \end{array}\right) ,
$$ 
where $\la = (1+\sqrt{5})/2$.  Then $\rho = 2+2\la^2 = 5+\sqrt{5}$, $n=12$ and $r^2 = (5+\sqrt{5})/2$.  This gives the corresponding vertices in $S^3$ as:
$$
\vec{x}_{\ell} = \frac{\sqrt{2}}{\sqrt{2+n(1-\ga )}}\left( \begin{array}{c} \sqrt{\frac{n(1-\ga )}{5+\sqrt{5}}}\, \vec{v}_{\ell} \\ 1 \end{array} \right) \quad (\ell = 1,2,3)\,.
$$ 
We then proceed as above for the dodecahedron to calculate the angular deficiency. 
\end{proof}

\begin{example} {\rm The $600$-cell is a convex $4$-dimensional regular polytope made up of $600$ tetrahedral $3$-polytopes.  It has $120$ vertices and $720$ edges.  Its vertex figure is a regular icosahedron.  If we consider an orthogonal projection onto the complex plane and let $\phi$ associate the corresponding value to each vertex, then by Theorem \ref{thm:reg-polytope}, $\phi$ satisfies (\ref{one}) with $\ga$ constant.  We can find the value of $\ga$ as follows.  

Since the edges of the $600$-cell all have the same length, in the notation of the above proof, we must have the distance from the origin to $\vec{x}_1$, that is $1$, equal to the distance between two neighbours of the vertex figure: $||\vec{x}_1-\vec{x}_2||$.  One can readily calculate:
$$
||\vec{x}_1-\vec{x}_2||^2 = \frac{8n(1-\ga )}{(5+\sqrt{5})[2+n(1-\ga )]}\,,
$$
to obtain the negative value:
$$
\ga = \frac{5(1-2\sqrt{5})}{3}\,.
$$
One can now confirm the generalization of the theorem of Descartes (Theorem \ref{thm:descartes}).  

The $600$-cell has $5$ tetrahedra around each edges, each having dihedral angle $\arccos (1/3)$.  Thus the angular deficiency at each edge (in absolute angle measure) is given by:
$$
\delta_e= 1-\frac{5}{2\pi}\arccos \frac{1}{3}\,.
$$
Substitution of the above value of $\ga$ into the third formula of Proposition \ref{prop:4-curvature} gives the deficit, or curvature at each vertex, as
$$
\delta_v = 6-\frac{15}{\pi}\arccos \frac{1}{3}\,.
$$
On then finds that $120\delta_v-720\delta_e = 0$, as required.
}
\end{example}

We can proceed similarly to obtain explicit formulae for higher dimensional $N$-curvature.  This is simplified by the fact that there are just three regular polytopes in dimensions $N\geq 5$, namely the $N$-simplex, the $N$-cube and the cross-polytope, with vertex figures an $(N-1)$-simplex in the first two cases and another cross-polytope in the last case.  To find the $N$-curvature requires the calculation of the $(N-2)$-dimensional measure of $(N-2)$-simplices in great spheres of $S^{N-1}$.  This is a standard, but non-trivial procedure using Schl\"afli's differential equality \cite{Sc}.  See also the expository article of J. Milnor for a nice account and references \cite{Mi}.  The article of J. Murakami provides explicit expressions in the $3$-sphere \cite{Mu}.  We do not attempt to derive these formulae here. 

\medskip

\noindent \emph{Edge-curvature}  Let $\Ga = (V,E)$ be a graph endowed with a solution $\phi$ to (\ref{one}), together with a choice of lift of a configured star into $\RR^N$ at each vertex, where the dimension $N$ is to be fixed over the whole graph.  If we suppose $N=3$, this may be achieved by defining an orientation on $\Ga$, as discussed in Section \ref{sec:distance}.  We suppose further that each star has a well-defined axis defined by a unit vector $\vec{u}(x)\in \RR^N$, for each $x\in V$.  The axis should be directed from the internal vertex of the star towards its centre of mass.

\begin{definition} \label{def:edge-curv}  Given an edge $e=\ov{xy}\in E$, define the \emph{edge-curvature} of $e$ to be the unique angle $\theta (e) := \arccos (\inn{\vec{u}(x), \vec{u}(y)}_{\RR^N})\in [0, \pi ]$.  Given a vertex $x\in V$, define the \emph{mean edge-curvature at $x$} to be the mean of the edge-curvatures of the edges incident with $x$.
\end{definition}

Thus the edge-curvature measures the angle between the axes of adjacent stars.  It is clearly independent of the freedom $\phi \mapsto \la \phi + \mu$ in the solution $\phi$.   By analogy with Riemannian geometry, various other curvatures can now be defined.  If we let $\ell (e)$ denote the length of an edge $e=\ov{xy}$ as given by Definition \ref{def:edge-length}, and $\theta (e)$ its edge-curvature, then the radius of the best-fit circle is given by $r(e) = \ell (e)/\theta (e)$ (by \emph{best-fit circle}, we mean the circle subtending the same arc length $\ell (e)$ for the given angle $\theta (e)$).  The \emph{normal curvature} of $e$ is then the reciprocal $1/r(e) = \theta (e)/\ell (e)$.  The \emph{mean curvature} at a vertex $x$ is the mean of the normal curvatures of the edges incident with $x$.  Since $\ell (e)$ depends on the scaling $\phi \mapsto \la \phi$, this quantity also depends on the scaling; the mean curvature should be thought of as the analogue of the same notion in the smooth setting, wherby we locally embed a Riemannian manifold in a Euclidean space. 

Ricci curvature is one of the most natural curvatures intrinsic to a Riemannian manifold.  Recall that given two unit directions $\vec{X}$ and $\vec{Y}$, the sectional curvature ${\rm Sec}\, (\vec{X}, \vec{Y})$ can be interpreted as the Gaussian curvature of a small geodesic surface generated by the plane $\vec{X}\wedge \vec{Y}$.  This in turn is the product of the principal curvatures which are the extremal values of the normal curvatures.  The Ricci curvature ${\rm Ric}\,(\vec{X},\vec{X})$ is then the sum: $\sum_j{\rm Sec}\,(\vec{X}, \vec{Y}_j)$ taken over an orthonormal frame $\{ \vec{Y}_j\}$ with each $\vec{Y}_j$ orthogonal to $\vec{X}$.   This motivates the following definition. 

\begin{definition}  Given a vertex $x\in V$ and two edges $e_1=\ov{xy_1}$ and $e_2=\ov{xy_2}$ with endpoint $x$, we define the \emph{sectional curvature ${\rm Sec}\,(e_1,e_2)$ determined by $e_1$ and $e_2$} to be the product:
$$
{\rm Sec}_x\,(e_1,e_2)= \theta (e_1)\theta (e_2)\,,
$$
where $\theta (e_j)$ ($j=1,2$) are the edge-curvatures.  For an edge $e=\ov{xy}$, the \emph{Ricci curvature} ${\rm Ric}\,(e,e)$ is the sum
$$
{\rm Ric}_x\,(e,e) = \ell (e)^2\sum_{z\sim x, z\neq y}\theta (\ov{xy})\theta( \ov{xz})\,,
$$ 
and the \emph{scalar curvature at} $x$ is given by
$$
{\rm Scal}_x = \sum_{y\sim x} {\rm Ric}\,(\ov{xy}, \ov{xy})/\ell (\ov{xy})^2.
$$
\end{definition}

The length scaling does not appear in the sectional curvature, since in the smooth setting this quantity depends only on the plane generated by two unit vectors.  On the other hand, the Ricci curvature is bilinear in its arguments and so should depend on the square of the length.  In Riemannian geometry, one usually applies the polarization identity to define ${\rm Ric}\,(\vec{X}, \vec{Y})$, however, there would seem to be no reasonable interpretation for the sum of two edges in our setting.  The dependence on length is once more removed from the scalar curvature, which is the trace of the Ricci curvature.

\section{The geometric spectrum and the $\ga$-polynomial} \label{sec:spectrum}
Recall the \emph{geometric spectrum} of a graph $\Ga = (V,E)$ is the set 
$$
\Si = \{ \ga\in \RR : \exists\ {\rm non-constant} \ \phi : V \ra \CC \ {\rm such \ that} \ \ga (\Delta \phi )^2 = (\dd \phi )^2\}\,.
$$
It is clear that the spectrum only depends on the isomorphism class of a graph.
Section \ref{sec:distance} shows how a particular value in the spectrum may correspond to local Euclidean geometry.  Thus, at a particular vertex, the edges which connect it to its neighbours may be realised as vectors in a Euclidean space; in particular their relative lengths are defined as well as the angles between them.  The edge-curvature as defined in Definition \ref{def:edge-curv}, may then be considered as a measure of how these local Euclidean geometries are pieced together to form a global geometric object.  None of these aspects require the graph to be embedded in an ambient space; they emerge purely from the combinatorial properties of the graph.

There are some obvious questions about the geometric spectrum of a graph: is it discrete?  is it finite? are there bounds?  To compute it, even for simple graphs, is quite challenging.  We deduced in Section \ref{sec:cyclic} the spectrum of some cyclic graphs of low order.  However, once the order increases, then the problem can become difficult.  Real regular cyclic sequences correspond to polynomial equations over the integers with positive coefficients having real roots; complex solutions correspond to closed walks in the plane with each step forming an angle $\pm \theta$ with the previous step, for some fixed $\theta$.  We may attempt an algebraic geometric approach to shed some light on these issues.

Let $\Ga = (V,E)$ be a connected graph.  We are interested in the possible real numbers $\ga$ for which there are non-constant solutions to the equation:
$$
\ga\Delta\phi^2=(\dd\phi )^2\,.
$$
Any solution is invariant by $\phi \mapsto \la\phi + \mu$, for complex constants $\la, \mu$.  Consider first how to parametrize all possible complex fields on the graph under this invariance.  

Label the vertices of the graph $x_1,x_2, \ldots , x_N$ and consider a non-constant complex field $\phi$ that assigns the value $\phi (x_k) = z_k$ to vertex $x_k$.  Then the space of all such fields is identified with the complex space $\CC^N \setminus \{ \mu (1,1,\ldots , 1): \mu \in \CC\}$.  Up to the equivalence $(z_1, \ldots , z_N) \sim (z_1+\mu , \ldots , z_N+\mu )$, we can identify these fields with the set $\Pi\setminus \{ 0\}$, where $\Pi$ is the linear subspace $\Pi = \{ \vec{Z}=(z_1, \ldots , z_N)\in \CC^N: z_1+\cdots + z_n=0\}\subset \CC^N$.  In effect, given any non-constant field $(z_1, \ldots , z_N)$, then $(z_1+\mu, \ldots , z_N+\mu )$ lies in the plane $z_1+\cdots +z_N = 0$, when we set $\mu = -\frac{1}{N}(z_1+\cdots + z_N)$.  By non-constancy, this is non-zero.  Furthermore, it is clear that any two equivalent fields correspond to the same point.  

Now consider the equivalence $\vec{Z}\sim \la \vec{Z}$, for $\la \in \CC\setminus \{ 0\}$.  This defines the \emph{moduli space} of fields up to equivalence, as ${\mathcal Z}:=\CP^{N-2}$.  Specifically, given a point $[z_1,\ldots , z_{N-1}]\in {\mathcal Z}$ in homogeneous co\"ordinates, we define a representative field by $(z_1, \ldots , z_{N-1}, z_N= - \sum_{k=1}^{N-1}z_k)\in \CC^N$.  In practice, we can set a field equal to $0$ and $1$ on two selected vertices $x_0$ and $x_1$, respectively, and label the other vertices arbitrarily.  This is only one chart and we miss those fields which coincide at these two vertices.  

If we consider $\ga$ as an arbitrary \emph{complex} parameter, then (\ref{one}) imposes a constraint at each vertex, so we have $N$ equations in $N-1$ parameters.  In general these are independent so that this is an overdetermined system, which may have no solutions.  The graphs on five and six vertices below have empty geometric spectrum.

\medskip
\begin{center}
\setlength{\unitlength}{0.254mm}
\begin{picture}(244,94)(113,-202)
        \allinethickness{0.254mm}\path(135,-195)(185,-195) 
        \allinethickness{0.254mm}\path(185,-195)(205,-145) 
        \allinethickness{0.254mm}\path(205,-145)(160,-110) 
        \allinethickness{0.254mm}\path(135,-195)(115,-145) 
        \allinethickness{0.254mm}\path(115,-145)(160,-110) 
        \allinethickness{0.254mm}\path(160,-110)(135,-195) 
        \allinethickness{0.254mm}\path(160,-110)(185,-195) 
        \allinethickness{0.254mm}\path(315,-110)(275,-135) 
        \allinethickness{0.254mm}\path(275,-135)(275,-175) 
        \allinethickness{0.254mm}\path(275,-175)(315,-200) 
        \allinethickness{0.254mm}\path(315,-200)(355,-175) 
        \allinethickness{0.254mm}\path(355,-175)(355,-135) 
        \allinethickness{0.254mm}\path(355,-135)(315,-110) 
        \allinethickness{0.254mm}\path(315,-110)(275,-175) 
        \allinethickness{0.254mm}\path(315,-110)(315,-200) 
        \allinethickness{0.254mm}\path(315,-110)(355,-175) 
        \allinethickness{0.254mm}\special{sh 0.3}\put(160,-110){\ellipse{4}{4}} 
        \allinethickness{0.254mm}\special{sh 0.3}\put(115,-145){\ellipse{4}{4}} 
        \allinethickness{0.254mm}\special{sh 0.3}\put(135,-195){\ellipse{4}{4}} 
        \allinethickness{0.254mm}\special{sh 0.3}\put(185,-195){\ellipse{4}{4}} 
        \allinethickness{0.254mm}\special{sh 0.3}\put(205,-145){\ellipse{4}{4}} 
        \allinethickness{0.254mm}\special{sh 0.3}\put(315,-110){\ellipse{4}{4}} 
        \allinethickness{0.254mm}\special{sh 0.3}\put(275,-135){\ellipse{4}{4}} 
        \allinethickness{0.254mm}\special{sh 0.3}\put(275,-175){\ellipse{4}{4}} 
        \allinethickness{0.254mm}\special{sh 0.3}\put(315,-200){\ellipse{4}{4}} 
        \allinethickness{0.254mm}\special{sh 0.3}\put(355,-175){\ellipse{4}{4}} 
        \allinethickness{0.254mm}\special{sh 0.3}\put(355,-135){\ellipse{4}{4}} 
\end{picture}
\end{center}
\medskip
\begin{center}
Figure 2.  Two graphs which admit no non-trivial solutions to (\ref{one}) with $\ga$ constant.
\end{center}
\medskip

Even with such simple examples, the equations are quite difficult to solve by hand.  Let us consider one way to approach the problem of computing the geometric spectrum.

As above, let $\Ga = (V,E)$ be a finite connected graph with $N$ vertices labeled $x_1, \ldots , x_N$.  For each $\ell = 2, \ldots N$, consider the following set of $N$ polynomials defined over the algebraically closed field $\CC$.  The variables are the values $\{ z_1, \ldots , z_N\}$ of a field on $\Ga$ with constraints $z_1=0$ and $z_{\ell} = 1$; we suppose the degree of vertex $j$ is $n(j)$ and that $z_{jk}\in \{ z_1, \ldots , z_N\}$ $(k = 1, \ldots , n(j))$ are the values of the field on the neighbours $x_{jk}$ of $x_j$.  The polynomials are then defined by   
$$
f_j{}^{\ell}:= \frac{\ga}{n(j)} \left( \sum_{k=1}^{n(j)}(z_j - z_{jk})\right)^2 - \sum_{k=1}^{n(j)}(z_j-z_{jk})^2 \qquad (z_1=0,\, z_{\ell} = 1)\,,
$$
in the $N-1$ \emph{complex} variables $\{ \ga , z_2,z_3, \ldots , \wh{z_{\ell}}, \ldots , z_N\}$.  Recall some facts and terminology from commutative algebra.  We are particularly interested in the techniques of Gr\"obner bases, for which we refer the reader to \cite{Ad-Lo, St}. 

For an ideal $I = <f_1, \ldots , f_N>$ in a polynomial ring $\CC [x_1, x_2, \ldots , x_M]$, we denote by $V(I)$ the corresponding variety given as the solution set of the equations $f_1=0, \, f_2=0,\, \ldots , f_N=0$.  Then $I$ is called \emph{zero-dimensional} if $V(I)$ is finite.  A \emph{Gr\"obner basis} for $I$ is a basis of polynomials which can be constructed from $f_1, \ldots , f_N$ using a particular algorithm, called the \emph{Buchberger algorithm}.  To employ this algorithm, one is required first to choose an order on monomials.  We shall only be concerned with \emph{lexicographical order} here, which means we first choose an ordering of the variables, say $x_1>x_2>\cdots >x_M$ and then order monomials $x^{\al}:=x_1{}^{\al_1}\cdots x_M{}^{\al_M}$, $x^{\be}:=x_1{}^{\be_1}\cdots x_M{}^{\be_M}$, by $x^{\al}<x^{\be}$ if and only if the first co\"ordinate $\al_i$ and $\be_i$ from the left which are different satisfy $\al_i<\be_i$.  With respect to the monomial order, every polynomial $f$ in $I$ has a leading term ${\rm lt}\,(f)$ which is the product ${\rm lt}\,(f)={\rm lc}\,(f){\rm lm}\,(f)$ of the leading coefficient with the leading monomial.  

A set of non-zero polynomials $G = \{ g_1, \ldots , g_P\}$ in $I$ is called a \emph{Gr\"obner basis for $I$} if and only if for all $f\in I$ such that $f\neq 0$, there is a $g_j$ in $G$ such that ${\rm lm}\,(g_j)$ divides ${\rm lm}\,(f)$.  The Gr\"obner basis is further called \emph{reduced} if for all $j$, ${\rm lc}\,(g_j)=1$ and $g_j$ is reduced with respect to $G\setminus \{ g_j\}$, that is, no non-zero term in $g_j$ is divisible by any ${\rm lm}\,(g_k)$ for any $k\neq j$.  A theorem of Buchberger states that every non-zero ideal has a unique reduced Gr\"ober basis with respect to a monomial order \cite{Bu}.  Gr\"obner bases are particularly useful for understanding the solution set of a system of polynomial equations.

Let $I$ be an ideal in the polynomial ring $\CC [x_1, x_2, \ldots , x_M]$ and let $G = \{ g_1, \ldots , g_P\}$ be the unique reduced Gr\"obner basis with respect to the lexicographical ordering induced by the order $x_1>x_2>\cdots >x_M$.  Then $V(I)$ is finite if and only if for each $j = 1, \ldots , M$, there exists a $g_k\in G$ such that ${\rm lm}\, g_k=x_j{}^{n_j}$ for some natural number $n_j$.  As a consequence, if $I$ is a zero-dimensional ideal, it follows that we can order $g_1, \ldots , g_P$ so that $g_1$ contains only the variable $x_M$, $g_2$ contains only $x_M, x_{M-1}$ and so on.  This is because the leading monomial of one element, $g_1$ say, of $G$ must be a power of $x_M$ and then no other term of $g_1$ can contain powers of any other variable (for such terms would be greater that any power of $x_M$ with respect to the monomial order), and so on for successive elements $g_2, g_3, \ldots $ of $G$.  We note also that $V(I)$ is empty if and only if $1\in G$.  

It is also the case that, with the above hypotheses, the polynomial $g_1$ is the \emph{least degree univariate polynomial in $x_M$ which belongs to} $I$ (any zero-dimensional ideal contains such a polynomial for every variable).  For if there was another univariate polynomial $p(x_M)$ with ${\rm deg}\, p < {\rm deg}\, g_1$, then ${\rm lm}\, p$ would divide ${\rm lm}\, g_1$ in a strict sense, which would contradict the fact that $G$ is a reduced Gr\"obner basis.  Let us now return to the case under consideration.

For each $\ell = 2, \ldots ,N$, consider the ideal $I_{\ell} = < f_1{}^{\ell}, \ldots , f_N{}^{\ell}>$.  Suppose that for each $\ell = 2, \ldots , N$ this admits a least degree univariate polynomial $p_{\ell}$ in $\ga$. This can be constructed by first choosing a lexicographical ordering of the variables with $\ga$ the smallest and then applying an algorithm (say the Buchberger algorithm) to construct the unique reduced Gr\"obner basis for $I_{\ell}$.  The first element of this basis gives $p_{\ell}$.  

\begin{definition} \label{def:gamma-poly}  We define the \emph{$\ga$-polynomial} $p = p_{\Ga}$ of the connected finite graph $\Ga = (V,E)$ to be the least common multiple of the least degree univariate polynomials $p_{\ell}$ ($\ell = 2, \ldots , N$) in $\ga$ associated to the equations {\rm (\ref{one})} for fields $(z_1, \ldots , z_N)$ on $\Ga$ with $z_1=0$ and $z_{\ell} = 1$:
$$
p:= {\rm lcm}\, (p_2, \ldots , p_N)\,,
$$
when each $p_{\ell}$ exists.
\end{definition}

The $\ga$-polynomial $p(\ga )$ is defined up to rational multiple and has rational coefficients.  This is because the initial polynomials $f_j{}^{\ell}$ used to define $p$ all have integer coefficients and the Buchberger algorithm then generates polynomials with rational coefficients--it involves at most division by coefficients--see \cite{Ad-Lo}. Clearly $p$ depends only on the isomorphism class of a graph and in the case when the equations (\ref{one}) admit no solutions for $\ga$ constant and complex, then $p \equiv 1$.  In this case we shall say that \emph{$p$ is trivial}.  The polynomials $p_{\ell}$ and so $p$ may still be well-defined even if the solution set of the equations is infinite (that is the corresponding ideal is no longer zero-dimensional).  In fact we know of no case when they are not well-defined.  

The elements of the geometric spectrum arise as real roots of $p$ (the problem of establishing the discreteness of the spectrum is clearly intimately related to knowing if $p$ is well-defined in all cases).  However, not all real roots may occur in the spectrum, for in general they must also solve the other equations determined by the Gr\"obner basis: $g_1=0, \ldots , g_P=0$.  Examples below illustrate this property.  We know of no two non-isomorphic connected graphs with non-trivial $\ga$-polynomial having the same $\ga$-polynomial.  However, the examples of Figure 2, give two non-isomorphic graphs having trivial $p$. 

The examples of the triangle $C_3$ (the cyclic graph on three vertices) and the bipartite graphs $K_{23}$ and $K_{33}$ are instructive.  We label the vertices as indicated and consider fields $\phi$ taking the values $\phi (x_j) = z_j$ at each vertex $x_j$. 

\medskip
\begin{center}
\setlength{\unitlength}{0.254mm}
\begin{picture}(376,117)(123,-166)
        \allinethickness{0.254mm}\path(130,-120)(200,-120) 
        \allinethickness{0.254mm}\path(200,-120)(160,-65) 
        \allinethickness{0.254mm}\path(160,-65)(125,-120) 
        \allinethickness{0.254mm}\path(140,-120)(195,-120) 
        \allinethickness{0.254mm}\path(125,-120)(150,-120) 
        \allinethickness{0.254mm}\path(255,-120)(225,-65) 
        \allinethickness{0.254mm}\path(225,-65)(305,-120) 
        \allinethickness{0.254mm}\path(305,-120)(330,-65) 
        \allinethickness{0.254mm}\path(330,-65)(285,-95) 
        \allinethickness{0.254mm}\path(275,-105)(255,-120) 
        \allinethickness{0.254mm}\path(280,-65)(290,-85) 
        \allinethickness{0.254mm}\path(305,-120)(295,-95) 
        \allinethickness{0.254mm}\path(280,-65)(270,-85) 
        \allinethickness{0.254mm}\path(255,-120)(265,-95) 
        \allinethickness{0.254mm}\path(365,-120)(365,-65) 
        \allinethickness{0.254mm}\path(365,-65)(420,-120) 
        \allinethickness{0.254mm}\path(420,-120)(475,-65) 
        \allinethickness{0.254mm}\path(475,-65)(475,-120) 
        \allinethickness{0.254mm}\path(420,-65)(420,-120) 
        \allinethickness{0.254mm}\path(420,-65)(395,-90) 
        \allinethickness{0.254mm}\path(390,-95)(365,-120) 
        \allinethickness{0.254mm}\path(420,-65)(445,-90) 
        \allinethickness{0.254mm}\path(450,-95)(475,-120) 
        \allinethickness{0.254mm}\path(365,-120)(395,-105) 
        \allinethickness{0.254mm}\path(475,-65)(445,-80) 
        \allinethickness{0.254mm}\path(435,-85)(425,-90) 
        \allinethickness{0.254mm}\path(415,-95)(405,-100) 
        \allinethickness{0.254mm}\path(365,-65)(395,-80) 
        \allinethickness{0.254mm}\path(475,-120)(445,-105) 
        \allinethickness{0.254mm}\path(435,-100)(425,-95) 
        \allinethickness{0.254mm}\path(415,-90)(405,-85) 
        \allinethickness{0.254mm}\special{sh 0.3}\put(200,-120){\ellipse{4}{4}} 
        \allinethickness{0.254mm}\special{sh 0.3}\put(160,-65){\ellipse{4}{4}} 
        \allinethickness{0.254mm}\special{sh 0.3}\put(125,-120){\ellipse{4}{4}} 
        \allinethickness{0.254mm}\special{sh 0.3}\put(255,-120){\ellipse{4}{4}} 
        \allinethickness{0.254mm}\special{sh 0.3}\put(225,-65){\ellipse{4}{4}} 
        \allinethickness{0.254mm}\special{sh 0.3}\put(280,-65){\ellipse{4}{4}} 
        \allinethickness{0.254mm}\special{sh 0.3}\put(330,-65){\ellipse{4}{4}} 
        \allinethickness{0.254mm}\special{sh 0.3}\put(305,-120){\ellipse{4}{4}} 
        \allinethickness{0.254mm}\special{sh 0.3}\put(365,-120){\ellipse{4}{4}} 
        \allinethickness{0.254mm}\special{sh 0.3}\put(365,-65){\ellipse{4}{4}} 
        \allinethickness{0.254mm}\special{sh 0.3}\put(420,-65){\ellipse{4}{4}} 
        \allinethickness{0.254mm}\special{sh 0.3}\put(475,-65){\ellipse{4}{4}} 
        \allinethickness{0.254mm}\special{sh 0.3}\put(475,-120){\ellipse{4}{4}} 
        \allinethickness{0.254mm}\special{sh 0.3}\put(420,-120){\ellipse{4}{4}} 
        \put(125,-136){\shortstack{$x_1$}} 
        \put(195,-136){\shortstack{$x_2$}} 
        \put(155,-61){\shortstack{$x_3$}} 
        \put(255,-136){\shortstack{$x_1$}} 
        \put(305,-136){\shortstack{$x_2$}} 
        \put(220,-61){\shortstack{$x_3$}} 
        \put(275,-61){\shortstack{$x_4$}} 
        \put(325,-61){\shortstack{$x_5$}} 
        \put(360,-136){\shortstack{$x_1$}} 
        \put(415,-136){\shortstack{$x_2$}} 
        \put(470,-136){\shortstack{$x_3$}} 
        \put(365,-61){\shortstack{$x_4$}} 
        \put(415,-61){\shortstack{$x_5$}} 
        \put(465,-61){\shortstack{$x_6$}} 
        \put(160,-166){\shortstack{$C_3$}} 
        \put(275,-166){\shortstack{$K_{23}$}} 
        \put(415,-166){\shortstack{$K_{33}$}} 
\end{picture}
\end{center}
\medskip

For the triangle, there are precisely two solutions to (\ref{one}) when we normalize so that $z_1=0, z_2=1$; specifically $z_3 = \frac{1}{2} \pm \ii \frac{\sqrt{3}}{2}$.  Then $p = p_2=p_3 = 3\ga - 2$ is the $\ga$ polynomial and the geometric spectrum is the unique root $\ga = 2/3$.  

For $K_{23}$, we find $p_2 = 1$ with no solution and $p_3 = \ga^2-2\ga + 1 = (\ga -1)^2$ with solution $z_1=0,\, z_3=1$, $z_2=0$, $z_5=\la$ arbitrary and $z_4=[1+\la \pm \sqrt{3}(1-\la )\ii ]/2$.  Then $p = \ga^2-2\ga + 1$ and the geometric spectrum is given by $\Si = \{ 1\}$.  

For $K_{33}$, we find $p_2 = 9\ga^2-26\ga + 17 = (\ga -1)(9\ga - 17)$ and $p_4 = 9\ga^3-35\ga^2+43\ga - 17 = (\ga -1)p_2$, so that the $\ga$-polynomial $p = 9\ga^3-35\ga^2+43\ga - 17$.  Although this has $\ga = 17/9$ as a root, the geometric spectrum $\Si = \{ 1\}$.  In fact for $\ga = 1, z_1=0, z_2 = 1$ we find a two complex parameter family of solutions as in Section \ref{sec:colourings}.  The next example shows that even for simple graphs, the $\ga$-polynomial can be quite complicated.

Consider the graph of constant degree three on six vertices whose edges form two concentric triangles as shown below. 
\medskip
\begin{center}
\setlength{\unitlength}{0.254mm}
\begin{picture}(134,99)(113,-162)
        \allinethickness{0.254mm}\path(145,-145)(180,-95) 
        \allinethickness{0.254mm}\path(180,-95)(215,-145) 
        \allinethickness{0.254mm}\path(180,-65)(245,-160) 
        \allinethickness{0.254mm}\path(245,-160)(115,-160) 
        \allinethickness{0.254mm}\path(115,-160)(180,-65) 
        \allinethickness{0.254mm}\path(215,-145)(145,-145) 
        \allinethickness{0.254mm}\path(180,-95)(180,-65) 
        \allinethickness{0.254mm}\path(245,-160)(215,-145) 
        \allinethickness{0.254mm}\path(145,-145)(115,-160) 
        \allinethickness{0.254mm}\special{sh 0.3}\put(180,-65){\ellipse{4}{4}} 
        \allinethickness{0.254mm}\special{sh 0.3}\put(180,-95){\ellipse{4}{4}} 
        \allinethickness{0.254mm}\special{sh 0.3}\put(245,-160){\ellipse{4}{4}} 
        \allinethickness{0.254mm}\special{sh 0.3}\put(215,-145){\ellipse{4}{4}} 
        \allinethickness{0.254mm}\special{sh 0.3}\put(145,-145){\ellipse{4}{4}} 
        \allinethickness{0.254mm}\special{sh 0.3}\put(115,-160){\ellipse{4}{4}} 
\end{picture}
\end{center}
\medskip

The $\ga$-polynomial is given by 
\begin{eqnarray*}
& 5859375\ga^{10}-67656250\ga^9+333521875\ga^8-926025000\ga^7 \\
& \qquad  +1603978830\ga^6 -1808486028\ga^5 +1339655598\ga^4 \\
& \qquad \qquad -639892872\ga^3+186760323\ga^2-29598858\ga+1883007\,.
\end{eqnarray*}
This has eight real roots, four of which are rational: $\ga = 3/5, 21/25, 1, 3$.  The value $3$ lies in the spectrum and corresponds to the obvious colouring of the vertices with two colours, by choosing identical colours for each triangle.  The value $1$ also lies in the spectrum and corresponds to the colouring of each triangle with the three colours $0,1, \frac{1}{2} + \ii \frac{\sqrt{3}}{2}$ corresponding the position function of an equilateral triangle in the plane; we do this so each vertex is joined to precisely one of the same colour.  We do not know which of the other roots lie in the spectrum since the computer program used for this example fails to solve the complete set of equations in a reasonable time.  There remain two conjugate complex roots of the $\ga$-polynomial.  

For the spectral values $\ga = 1$, we can compute the corresponding curvature as given by Definition \ref{def:N-curvature} and Proposition \ref{prop:3-curvature}, to obtain (in radians) $\delta = 2\pi$ at each vertex, to give a total curvature of $12\pi$.

 As we calculated above, the bipartite graph $K_{33}$ has $\ga$-polynomial given by
$$
9\ga^3-35\ga^2+43\ga-17 = (\ga - 1)^2(9\ga - 17)\,,
$$
In particular, these two graphs of constant degree three on six vertices cannot be isomorphic.

\section{An elementary universe}  \label{sec:particles}  Our objective is to construct an elementary universe populated entirely by graphs from which geometry and dynamics emerge.  The universe is based on a binary relation between objects:  are they connected by an edge or not?  It is only this relation that matters; the nature of the objects being irrelevant.  Our perspective is that, out of the graphs that are so formed, further implicit structure is present, given by the geometric spectrum and the corresponding fields.  This implicit structure comes into play when graphs correlate.  Thus we describe ways in which graphs can interact and so dynamics, that is \emph{change}, appears.  With a suitable definition of time, this change can be ordered to give a universe endowed with local geometry and time.  

The initial data for our universe is a graph $\Ga$ made up of a finite number of connected components $\Ga_1, \Ga_2, \ldots , \Ga_K$.  In addition to the binary relation between vertices (whether or not they are connected by an edge), there is an additional relation between them: whether or not they belong to the same connected component of $\Ga$.  A priori there is no reason to give greater emphasis to the property that two vertices be connected by an edge, rather than that they are not so connected.  The representation by drawing an edge just gives a convenient way to visualize the relation.  

A \emph{particle} is a connected component $\Ga_k$ of $\Ga$.  A \emph{state of the particle} is an equivalence class of solutions $(\phi , \ga )$ to equation (\ref{one}) on $\Ga_k$, where $\phi \sim \la \phi + \mu$ for $\la , \mu \in \CC$ and where we require that $\ga \leq 1$ at each vertex.  A member of an equivalence class will be called a \emph{representative state} and we shall write $[\phi ]$ for the equivalence class determined by the state $\phi$.  A state for which $\ga$ is constant will be called an \emph{isostate}.  We allow point particles, consisting of a single vertex with all elements of $\CC$ as representative states.  

A particle is not deemed to be in any state, but carries with it, its ensemble of states.  By analogy with quantum mechanics, when two particles correlate, each falls into a particular state; that is, states are chosen with a certain probability.  Isostates are favoured for empirical reasons relating to energy that we discuss at the end of this section.  After correlation, further states may be present in the combined graph, permitting new correlations with other graphs that were not possible prior to correlation.  An \emph{evolution} of the universe is a sequence $\Ga \mapsto \Ga^{'}  \mapsto \Ga^{''} \mapsto \cdots$ of graphs, whereby a subsequent graph is obtained from the previous one by specific rules to be defined.  The order of the sequence is dictated by the rules.  There are three \emph{changes} in our universe that we now specify:  correlation between particles; internal mutation of a particle; separation of a particle into two or more particles.

\medskip

\noindent (i)  \emph{Correlation.}  A \emph{correlation} between two particles $\Ga_1 = (V_1, E_1)$ and $\Ga_1 = (V_2, E_2)$ is a new particle $\Ga_1\star \Ga_2 = (V_1\cup V_2, E)$, where we require $E_1\cup E_2 \subset E$, i.e. vertices are preserved and the edge set is increased.  For a correlation to occur, we require $\Ga_1$ be in a state $[\phi_1]$, $\Ga_2$ be in a state $[\phi_2]$ and $\Ga$ be in a state $[\phi ]$ such that $\phi\vert_{V_1} \in [\phi_1]$ and $\phi\vert_{V_1}\in [\phi_2]$.  After correlation, the new particle carries its ensemble of states and is not considered to be in any particular state.  An example of correlation is shown in the figure below.

\medskip
\begin{center}
\setlength{\unitlength}{0.254mm}
\begin{picture}(342,212)(120,-256)
        \allinethickness{0.254mm}\path(160,-60)(135,-105) 
        \allinethickness{0.254mm}\path(135,-105)(160,-145) 
        \allinethickness{0.254mm}\path(160,-145)(185,-105) 
        \allinethickness{0.254mm}\path(185,-105)(160,-60) 
        \allinethickness{0.254mm}\path(185,-105)(135,-105) 
        \allinethickness{0.254mm}\path(320,-60)(295,-105) 
        \allinethickness{0.254mm}\path(320,-60)(345,-105) 
        \allinethickness{0.254mm}\path(345,-105)(320,-145) 
        \allinethickness{0.254mm}\path(320,-145)(295,-105) 
        \allinethickness{0.254mm}\path(295,-105)(345,-105) 
        \allinethickness{0.254mm}\path(240,-200)(215,-155) 
        \allinethickness{0.254mm}\path(215,-155)(190,-200) 
        \allinethickness{0.254mm}\path(190,-200)(215,-240) 
        \allinethickness{0.254mm}\path(215,-240)(240,-200) 
        \allinethickness{0.254mm}\path(245,-200)(270,-155) 
        \allinethickness{0.254mm}\path(270,-155)(295,-200) 
        \allinethickness{0.254mm}\path(295,-200)(270,-240) 
        \allinethickness{0.254mm}\path(245,-200)(270,-240) 
        \allinethickness{0.254mm}\path(215,-155)(270,-155) 
        \allinethickness{0.254mm}\path(270,-240)(215,-240) 
        \allinethickness{0.254mm}\path(240,-200)(190,-200) 
        \allinethickness{0.254mm}\path(245,-200)(295,-200) 
        \allinethickness{0.254mm}\special{sh 0.3}\put(160,-60){\ellipse{4}{4}} 
        \allinethickness{0.254mm}\special{sh 0.3}\put(135,-105){\ellipse{4}{4}} 
        \allinethickness{0.254mm}\special{sh 0.3}\put(185,-105){\ellipse{4}{4}} 
        \allinethickness{0.254mm}\special{sh 0.3}\put(160,-145){\ellipse{4}{4}} 
        \allinethickness{0.254mm}\special{sh 0.3}\put(320,-60){\ellipse{4}{4}} 
        \allinethickness{0.254mm}\special{sh 0.3}\put(295,-105){\ellipse{4}{4}} 
        \allinethickness{0.254mm}\special{sh 0.3}\put(345,-105){\ellipse{4}{4}} 
        \allinethickness{0.254mm}\special{sh 0.3}\put(320,-145){\ellipse{4}{4}} 
        \allinethickness{0.254mm}\special{sh 0.3}\put(270,-155){\ellipse{4}{4}} 
        \allinethickness{0.254mm}\special{sh 0.3}\put(215,-155){\ellipse{4}{4}} 
        \allinethickness{0.254mm}\special{sh 0.3}\put(190,-200){\ellipse{4}{4}} 
        \allinethickness{0.254mm}\special{sh 0.3}\put(215,-240){\ellipse{4}{4}} 
        \allinethickness{0.254mm}\special{sh 0.3}\put(240,-200){\ellipse{4}{4}} 
        \allinethickness{0.254mm}\special{sh 0.3}\put(245,-200){\ellipse{4}{4}} 
        \allinethickness{0.254mm}\special{sh 0.3}\put(295,-200){\ellipse{4}{4}} 
        \allinethickness{0.254mm}\special{sh 0.3}\put(270,-240){\ellipse{4}{4}} 
        \put(235,-111){\shortstack{$\star$}} 
        \put(165,-56){\shortstack{$\sqrt{3}i$}} 
        \put(325,-56){\shortstack{$\sqrt{3}i$}} 
        \put(200,-151){\shortstack{$\sqrt{3}i$}} 
        \put(250,-151){\shortstack{$2+\sqrt{3}i$}} 
        \put(110,-111){\shortstack{$-1$}} 
        \put(190,-111){\shortstack{$1$}} 
        \put(270,-111){\shortstack{$-1$}} 
        \put(350,-111){\shortstack{$1$}} 
        \put(150,-155){\shortstack{$-\sqrt{3}i$}} 
        \put(310,-155){\shortstack{$-\sqrt{3}i$}} 
        \put(240,-196){\shortstack{$1$}} 
        \put(110,-206){\shortstack{$=$}}
        \put(165,-206){\shortstack{$-1$}} 
        \put(300,-206){\shortstack{$3$}} 
        \put(200,-256){\shortstack{$-\sqrt{3}i$}} 
        \put(255,-256){\shortstack{$2-\sqrt{3}i$}} 
        \put(385,-66){\shortstack{$\gamma =2/3$}} 
        \put(390,-111){\shortstack{$\gamma =0$}} 
        \put(385,-151){\shortstack{$\gamma =2/3$}} 
        \put(390,-216){\shortstack{$\gamma =0$}} 
\end{picture}

\small{Fig. 1.  An example of correlation: the two particles at the top correlate to form a single particle in an isostate.}  
\end{center}
\medskip

We have drawn the resulting particle as an invariant framework in the plane, so the two central vertices appear superimposed at the point $1$; but in fact they are distinct unconnected vertices.
This correlation is to be favoured, since from two non-isostates, the outcome is an isostate.

\medskip

\noindent \emph{Separation.}  A \emph{separation} of a particle is a dissociation of $\Ga$ into two particles $\Ga_1=(V_1,E_1)$, $\Ga_2=(V_2,E_2)$ with $\Ga$ a correlation of $\Ga_1$ and $\Ga_2$.  For separation to occur, we require $\Ga$ be in a state $[\phi ]$, $\Ga_1$ be in a state $[\phi_1]$ and $\Ga_2$ be in a state $[\phi_2]$ with $\phi\vert_{V_1}\in [\phi_1]$ and $\phi\vert_{V_2}\in [\phi_2]$.  An example of separation is given by the last example of Section \ref{sec:spectrum}, when the two concentric triangles connected by edges as indicated, falls into the state corresponding to the spectral value $\ga = 1$.  The edges joining vertices on which the field has a common value are removed to give two disjoint triangles.

\medskip
\begin{center}
\setlength{\unitlength}{0.254mm}
\begin{picture}(410,147)(70,-221)
        \allinethickness{0.254mm}\path(160,-85)(80,-205) 
        \allinethickness{0.254mm}\path(80,-205)(240,-205) 
        \allinethickness{0.254mm}\path(240,-205)(160,-85) 
        \allinethickness{0.254mm}\path(160,-130)(125,-185) 
        \allinethickness{0.254mm}\path(125,-185)(195,-185) 
        \allinethickness{0.254mm}\path(195,-185)(160,-130) 
        \allinethickness{0.254mm}\path(160,-85)(160,-130) 
        \allinethickness{0.254mm}\path(195,-185)(240,-205) 
        \allinethickness{0.254mm}\path(125,-185)(80,-205) 
        \allinethickness{0.254mm}\path(400,-85)(320,-205) 
        \allinethickness{0.254mm}\path(320,-205)(480,-205) 
        \allinethickness{0.254mm}\path(480,-205)(400,-85) 
        \allinethickness{0.254mm}\path(400,-130)(365,-185) 
        \allinethickness{0.254mm}\path(365,-185)(435,-185) 
        \allinethickness{0.254mm}\path(435,-185)(400,-130) 
        \allinethickness{0.254mm}\path(250,-140)(310,-140)\special{sh 1}\path(310,-140)(304,-138)(304,-140)(304,-142)(310,-140) 
        \put(165,-131){\shortstack{$0$}} 
        \put(165,-86){\shortstack{$0$}} 
        \put(115,-186){\shortstack{$1$}} 
        \put(70,-216){\shortstack{$1$}} 
        \put(150,-198){\shortstack{$\frac{1}{2}+\frac{\sqrt{3}}{2}i$}} 
        \put(210,-221){\shortstack{$\frac{1}{2}+\frac{\sqrt{3}}{2}i$}} 
\end{picture}

\small{Fig.2.  An example of separation}
\end{center}
\medskip

\noindent \emph{Mutation.}  A \emph{mutation} of a particle $\Ga$ is a change $\Ga = (V,E) \ra \Si = (W,F)$, where $\Si$ is a new particle with $V=W$.  For mutation to occur, we require $\Ga$ be in a state $[\phi ]$, $\Si$ be in a state $[\psi ]$ with $\phi \in [\psi ]$.  Collapsing is a particular example of mutation provided it does not disconnect the particle, whereby we remove edges that connect vertices on which $\phi$ takes on the same value.  If collapsing disconnects the particle, it falls into the category of separation.  An example of mutation is illustrated in the following figure. 

\medskip
\begin{center}
\setlength{\unitlength}{0.254mm}
\begin{picture}(351,152)(50,-236)
        \allinethickness{0.254mm}\path(120,-100)(80,-160) 
        \allinethickness{0.254mm}\path(80,-160)(120,-220) 
        \allinethickness{0.254mm}\path(120,-220)(160,-160) 
        \allinethickness{0.254mm}\path(120,-100)(160,-160) 
        \allinethickness{0.254mm}\path(240,-160)(280,-100) 
        \allinethickness{0.254mm}\path(280,-100)(320,-160) 
        \allinethickness{0.254mm}\path(320,-160)(280,-220) 
        \allinethickness{0.254mm}\path(280,-220)(240,-160) 
        \allinethickness{0.254mm}\path(240,-160)(320,-160) 
        \allinethickness{0.254mm}\special{sh 0.3}\put(120,-100){\ellipse{4}{4}} 
        \allinethickness{0.254mm}\special{sh 0.3}\put(80,-160){\ellipse{4}{4}} 
        \allinethickness{0.254mm}\special{sh 0.3}\put(160,-160){\ellipse{4}{4}} 
        \allinethickness{0.254mm}\special{sh 0.3}\put(120,-220){\ellipse{4}{4}} 
        \allinethickness{0.254mm}\special{sh 0.3}\put(280,-100){\ellipse{4}{4}} 
        \allinethickness{0.254mm}\special{sh 0.3}\put(240,-160){\ellipse{4}{4}} 
        \allinethickness{0.254mm}\special{sh 0.3}\put(320,-160){\ellipse{4}{4}} 
        \allinethickness{0.254mm}\special{sh 0.3}\put(280,-220){\ellipse{4}{4}} 
        \allinethickness{0.254mm}\path(185,-160)(215,-160)\special{sh 1}\path(215,-160)(209,-158)(209,-160)(209,-162)(215,-160) 
        \put(115,-96){\shortstack{$1$}} 
        \put(115,-236){\shortstack{$0$}} 
        \put(30,-150){\shortstack{$\frac{1}{2}+\frac{\sqrt{3}}{2}i$}} 
        \put(160,-150){\shortstack{$\frac{1}{2}+\frac{\sqrt{3}}{2}i$}} 
        \put(140,-106){\shortstack{$\gamma =1$}} 
        \put(135,-226){\shortstack{$\gamma =1$}} 
        \put(160,-176){\shortstack{$\gamma = 2/3$}} 
        \put(335,-166){\shortstack{$\gamma =$1}} 
\end{picture}

\small{Fig. 3. The particle on the left mutates into an isostate by the addition of an edge.} 
\end{center}
\medskip

Another illustration is given by the particle of Figure 1, which, after mutation, produces the $1$-skeleton of a cube.   

\medskip
\begin{center}
\setlength{\unitlength}{0.254mm}
\begin{picture}(365,95)(90,-175)
        \allinethickness{0.254mm}\path(120,-80)(90,-130) 
        \allinethickness{0.254mm}\path(90,-130)(120,-175) 
        \allinethickness{0.254mm}\path(120,-175)(150,-130) 
        \allinethickness{0.254mm}\path(150,-130)(120,-80) 
        \allinethickness{0.254mm}\path(155,-130)(185,-80) 
        \allinethickness{0.254mm}\path(185,-175)(155,-130) 
        \allinethickness{0.254mm}\path(150,-130)(90,-130) 
        \allinethickness{0.254mm}\path(185,-80)(215,-130) 
        \allinethickness{0.254mm}\path(215,-130)(185,-175) 
        \allinethickness{0.254mm}\path(215,-130)(155,-130) 
        \allinethickness{0.254mm}\path(120,-80)(185,-80) 
        \allinethickness{0.254mm}\path(185,-175)(120,-175) 
        \allinethickness{0.254mm}\path(360,-80)(325,-130) 
        \allinethickness{0.254mm}\path(325,-130)(360,-175) 
        \allinethickness{0.254mm}\path(360,-80)(425,-80) 
        \allinethickness{0.254mm}\path(360,-175)(425,-175) 
        \allinethickness{0.254mm}\path(425,-175)(455,-130) 
        \allinethickness{0.254mm}\path(455,-130)(425,-80) 
        \allinethickness{0.254mm}\path(390,-130)(325,-130) 
        \allinethickness{0.254mm}\path(390,-130)(425,-80) 
        \allinethickness{0.254mm}\path(390,-130)(425,-175) 
        \allinethickness{0.254mm}\path(395,-130)(455,-130) 
        \allinethickness{0.254mm}\path(360,-175)(385,-140) 
        \allinethickness{0.254mm}\path(360,-80)(385,-120) 
        \allinethickness{0.254mm}\path(255,-130)(290,-130)\special{sh 1}\path(290,-130)(284,-128)(284,-130)(284,-132)(290,-130) 
\end{picture}

\small{Fig. 4. Mutation occurs when two edges ``flip" to connect the middle outer vertices to different central vertices.} 
\end{center}
\medskip

This change of state embeds the particle in Euclidean space as a $3$-dimensional object.  The particles before and after mutation both correspond to holomorphic states and so we view this mutation as neutral.

\medskip

\begin{example} \label{ex:evolution} {\rm An evolution of a simple universe is as follows. First take the universe $\Ga$ consisting of two copies of the left-hand particle of Figure 3.  These then mutate to form $\Ga^{'}$ consisting of two copies of the right-hand particle of Figure 3.  A correlation then occurs as in Figure 1 to form $\Ga^{''}$.  Finally a mutation takes place as in Figure 4 to form $\Ga^{'''}$.  We view this evolution as irreversible, in the sense that if we begin with the $1$-skeleton of the cube, the probability that it fall into a state $\phi$ which has identical values on two diagonally opposite vertices to enable the reciprocal mutation of Figure 4, would be negligible. However, this is speculative, since we have not given a rule for deciding the probability of transition.}
\end{example}  

A desirable objective would be to produce a complex universe from a simple initial state.  One way to accomplish this is to suppose the existence of virtual point particles that are susceptible to correlate with existing particles.  We now explore this possibility in more detail.  

Consider a graph $\Ga = (V,E)$ together with a function $\phi : V \ra \CC$ not necessarily a solution to (\ref{one}).  For ease of representation, suppose that each vertex $x\in V$ be placed at its corresponding position $\phi (x)$ in the complex plane.  We fix our attention on a particular vertex, $x_0$ say, which by translation, we suppose placed at the origin.  Suppose $x_0$ has $k$ neighbours placed at $z_1, \ldots , z_k$.
We now wish to add a new vertex placed at $w$ and to join it to $x_0$ in such a way as to satisfy (\ref{one}) at $x_0$.  Specifically, we wish to consider the locus of points $w$ which can be placed in this way.  Since the mapping
\begin{equation} \label{locus}
w\mapsto \frac{z_1{}^2+ \cdots + z_k{}^2 + w^2}{(z_1+ \cdots + z_k+w)^2}\,,
\end{equation}
is in general holomorphic in $w$, we expect a $1$-parameter family of values of $w$ for which the right-hand side is real.      

\begin{example} {\rm Consider the graph on three vertices as indicated below.  It may be that the extremal vertices placed at $re^{\ii\theta}$ and $1$ are joined to other vertices, but for the moment we are just interested in satisfying (\ref{one}) at the origin.
\medskip
\begin{center}
\setlength{\unitlength}{0.254mm}
\begin{picture}(208,68)(103,-147)
        \allinethickness{0.254mm}\path(155,-145)(225,-145) 
        \allinethickness{0.254mm}\path(155,-145)(240,-90) 
        \allinethickness{0.254mm}\dashline{5}[5](155,-145)(105,-100) 
        \allinethickness{0.254mm}\special{sh 0.3}\put(105,-100){\ellipse{4}{4}} 
        \allinethickness{0.254mm}\special{sh 0.3}\put(240,-90){\ellipse{4}{4}} 
        \allinethickness{0.254mm}\special{sh 0.3}\put(225,-145){\ellipse{4}{4}} 
        \allinethickness{0.254mm}\special{sh 0.3}\put(155,-145){\ellipse{4}{4}} 
        \put(245,-91){\shortstack{$re^{i\theta}$}} 
        \put(230,-146){\shortstack{$1$}} 
        \put(115,-101){\shortstack{$w$}} 
        \put(150,-136){\shortstack{$0$}} 
\end{picture}
\end{center}
\medskip
If we set $w = u +\ii v$, then it is a routine computation to show that the identity $\ga (1+re^{\ii\theta} + w)^2= 1+r^2e^{2\ii\theta} + w^2$ has $\ga$ real if and only if the following algebraic equation of degree three in $u$ and $v$ is satisfied:
\begin{equation} \label{algebraic}
\begin{array}{l} (ur\sin \theta - v(1+r\cos \theta ))(u^2+v^2) + r\sin \theta (u^2-v^2) - 2uvr\cos \theta  \\
 + r\sin \theta (1-r^2-2r\cos \theta )u + (1+r\cos \theta + r^3\cos \theta + r^2 \cos 2\theta )v \\
  \qquad \qquad + r(1-r^2)\sin \theta = 0\,. \end{array}
\end{equation} 
There are two cases when the solution set can be explicitly written down:

\noindent (i) $\theta = \pi /2$.  Equation (\ref{algebraic}) now becomes:
\begin{equation} \label{implicit-degree3}
(ru-v)(u^2+v^2)+r(u^2-v^2)+r(1-r^2)u + (1-r^2)v + r(1-r^2) = 0\,.
\end{equation}
For each $r \neq 0, 1$, this is a smooth curve except at the singular point $(u,v) = (-1,-r)$.  For the case $r=2$, the curve is as indicated in the graph below.  When $r=1$ (so the original graph is holomorphic at the origin), this gives the algebraic set:
$$
(u-v)(u^2+v^2+u+v)=0\,,
$$
consisting of the union of the line $u=v$ and the circle $(u+\frac{1}{2})^2+(v+\frac{1}{2})^2 = \frac{1}{2}$\,.

\begin{figure}
\begin{center}
\includegraphics[width=10cm, height=10cm]{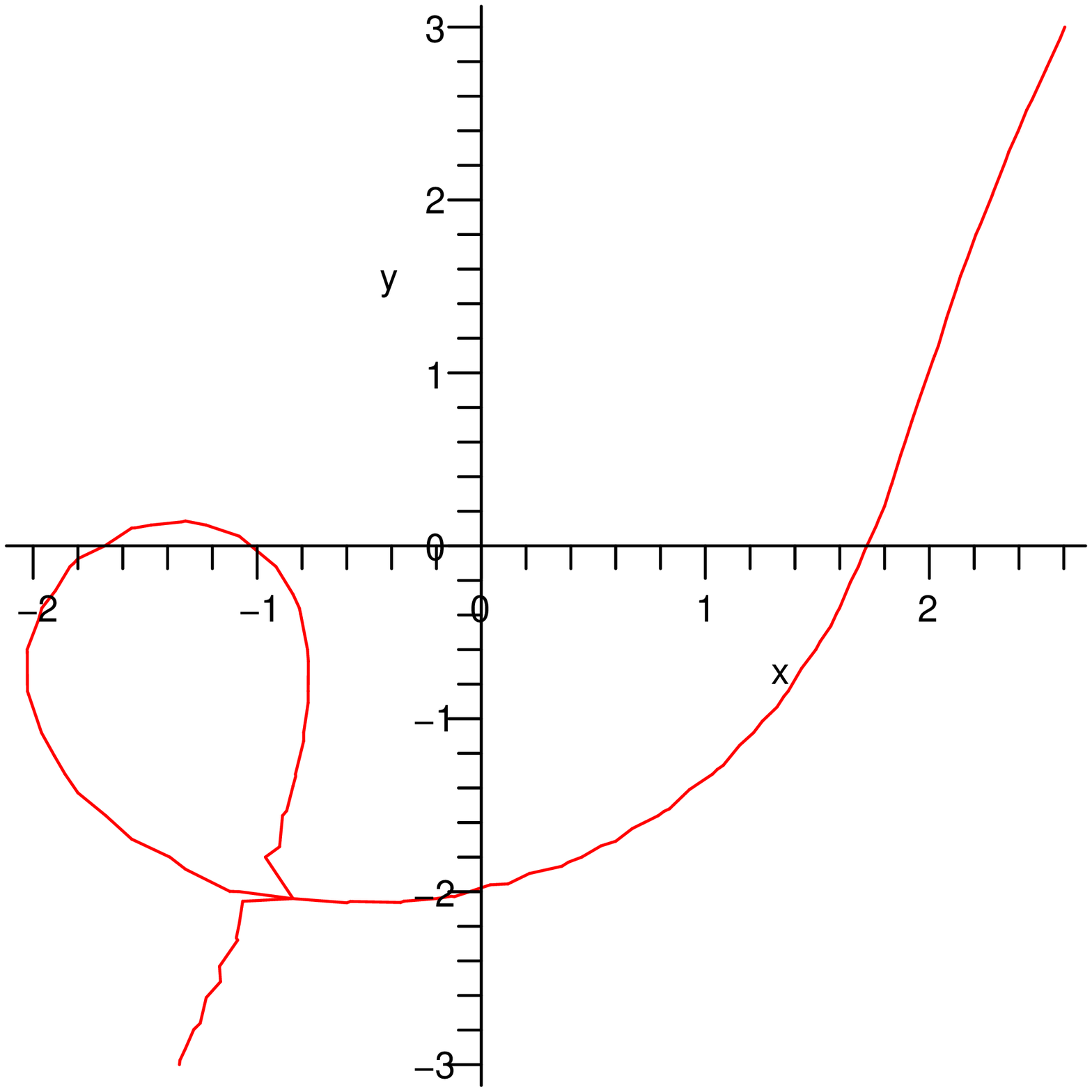}
\end{center}
\end{figure}
}
\end{example}

A variant on the above procedure is to consider two particles $\Ga$ and $\Si$ and to add a new vertex $x$ which correlates with both particles, joining it to $x_0$ in $\Ga$ and $y_0$ in $\Si$, say.  In order to correlate, suppose $\Ga$ falls into state $[\phi ]$ and $\Si$ falls into state $[\psi ]$.  Now we require that (\ref{one}) be satisfied at the new vertex $x$.  From Section \ref{sec:cyclic}, this is the case if and only if $|\phi (x)-\phi (x_0)|=|\psi (y)-\psi (y_0)|$, for representative states.  This can be further generalized by creating a new vertex $x$ and attempting to join several vertices to $x$.  Such correlations can lead to discrete phenomena when we combine the various constraints.   That is, the various loci determined by the real solutions to (\ref{locus}) will in general intersect in a discrete set of points.

The possibility that point particles may attach themselves to existing (more complex) particles, can lead to duplication and eventually a complex universe.  The following sequence of correlation, mutation and separation gives an example of duplication.

\medskip
\begin{center}
\setlength{\unitlength}{0.254mm}
\begin{picture}(444,59)(78,-117)
        \allinethickness{0.254mm}\path(120,-60)(80,-115) 
        \allinethickness{0.254mm}\path(80,-115)(160,-115) 
        \allinethickness{0.254mm}\path(160,-115)(120,-60) 
        \allinethickness{0.254mm}\path(240,-60)(200,-115) 
        \allinethickness{0.254mm}\path(200,-115)(280,-115) 
        \allinethickness{0.254mm}\path(280,-115)(240,-60) 
        \allinethickness{0.254mm}\path(240,-60)(240,-80) 
        \allinethickness{0.254mm}\path(200,-115)(220,-105) 
        \allinethickness{0.254mm}\path(280,-115)(260,-105) 
        \allinethickness{0.254mm}\path(360,-60)(320,-115) 
        \allinethickness{0.254mm}\path(320,-115)(400,-115) 
        \allinethickness{0.254mm}\path(400,-115)(360,-60) 
        \allinethickness{0.254mm}\path(360,-60)(360,-80) 
        \allinethickness{0.254mm}\path(320,-115)(340,-105) 
        \allinethickness{0.254mm}\path(400,-115)(380,-105) 
        \allinethickness{0.254mm}\path(380,-105)(360,-80) 
        \allinethickness{0.254mm}\path(360,-80)(340,-105) 
        \allinethickness{0.254mm}\path(340,-105)(380,-105) 
        \allinethickness{0.254mm}\path(480,-60)(440,-115) 
        \allinethickness{0.254mm}\path(440,-115)(520,-115) 
        \allinethickness{0.254mm}\path(520,-115)(480,-60) 
        \allinethickness{0.254mm}\path(480,-80)(460,-105) 
        \allinethickness{0.254mm}\path(460,-105)(500,-105) 
        \allinethickness{0.254mm}\path(500,-105)(480,-80) 
        \allinethickness{0.254mm}\special{sh 0.3}\put(120,-60){\ellipse{4}{4}} 
        \allinethickness{0.254mm}\special{sh 0.3}\put(80,-115){\ellipse{4}{4}} 
        \allinethickness{0.254mm}\special{sh 0.3}\put(160,-115){\ellipse{4}{4}} 
        \allinethickness{0.254mm}\special{sh 0.3}\put(240,-60){\ellipse{4}{4}} 
        \allinethickness{0.254mm}\special{sh 0.3}\put(240,-80){\ellipse{4}{4}} 
        \allinethickness{0.254mm}\special{sh 0.3}\put(220,-105){\ellipse{4}{4}} 
        \allinethickness{0.254mm}\special{sh 0.3}\put(200,-115){\ellipse{4}{4}} 
        \allinethickness{0.254mm}\special{sh 0.3}\put(260,-105){\ellipse{4}{4}} 
        \allinethickness{0.254mm}\special{sh 0.3}\put(280,-115){\ellipse{4}{4}} 
        \allinethickness{0.254mm}\special{sh 0.3}\put(360,-60){\ellipse{4}{4}} 
        \allinethickness{0.254mm}\special{sh 0.3}\put(360,-80){\ellipse{4}{4}} 
        \allinethickness{0.254mm}\special{sh 0.3}\put(380,-105){\ellipse{4}{4}} 
        \allinethickness{0.254mm}\special{sh 0.3}\put(340,-105){\ellipse{4}{4}} 
        \allinethickness{0.254mm}\special{sh 0.3}\put(320,-115){\ellipse{4}{4}} 
        \allinethickness{0.254mm}\special{sh 0.3}\put(400,-115){\ellipse{4}{4}} 
        \allinethickness{0.254mm}\special{sh 0.3}\put(480,-60){\ellipse{4}{4}} 
        \allinethickness{0.254mm}\special{sh 0.3}\put(480,-80){\ellipse{4}{4}} 
        \allinethickness{0.254mm}\special{sh 0.3}\put(460,-105){\ellipse{4}{4}} 
        \allinethickness{0.254mm}\special{sh 0.3}\put(500,-105){\ellipse{4}{4}} 
        \allinethickness{0.254mm}\special{sh 0.3}\put(520,-115){\ellipse{4}{4}} 
        \allinethickness{0.254mm}\special{sh 0.3}\put(440,-115){\ellipse{4}{4}} 
\end{picture}
\end{center}
\medskip 
We begin with a triangle on the left-hand side.  Three isolated vertices then attach themselves in a symmetric way.  This must be done to preserve the property that equation (\ref{one}) remain satisfied.  Symmetry is prefered since this leads to an isostate.  In fact, if the vertices of the left-hand triangle have representative field values $0, 1, \frac{1}{2} + \frac{\sqrt{3}}{2}\ii$, then the point particle connected to $0$ should have representative field value $\frac{3}{2} + \frac{\sqrt{3}}{2}\ii$, with the other point particles similarly assigned values to give an isostate with $\ga = 1$ (so the figure is misleading if we view the field as the position function, but avoids crossing edges).  A mutation now occurs whereby the new vertices are joined by edges in the way shown.  This produces an non-isostate particle with $\ga = 7/9$ at the new vertices (with $\ga$ still equal to $1$ at the original vertices), which then falls into an isostate given by the example on the left-hand side of Figure 2, with spectral value $\ga = 1$, once more.  Finally, separation occurs as in Figure 2.  

Further examples of how particles may correlate are given in Appendix \ref{sec:invariant-structures}.  This principally concerns unstable double cones which correlate to form a rich array of stable geometric structures. 

We now wish to discuss how we may associate an energy to a state and the role of isostates.  A natural quantity that occurs, which we may call \emph{energy}, is given in the notation of Section \ref{sec:distance} by $||Z||^2 = {\rm trace}\, Z^tZ$ (see equation (\ref{ZtZ})).  

Recall that at each vertex, the vector $Z$ picks out the configured star that projects to a state normalized to be zero at the vertex in question; furthermore, it has the minimum Frobenius norm amongst solutions to the system (\ref{system-2}).  When $N=3$, it is precisely this solution that satisfies the constraint (\ref{system-2-constraint}).  In dimension $N>3$, the solution $X=Z+Y$ which satisfies the constraint is no longer in general minimizing, whereas $Z$ is still minimizing amongst solutions to (\ref{system-2}).  In dimension $N=2$, at a vertex of degree $2$ as discussed prior to Theorem \ref{thm:lift}, the solution to the lifting problem is completely determined.  If we are to regard this energy as a criterion for stability in our universe, then this may provide reasons why dimension three should emerge as a favoured dimension.  In defining energy, we also need to take into account the normalizing freedom.

\begin{definition}  \label{def:energy} Let $\Ga = (V,E)$ be a connected graph endowed with a non-constant solution $\phi$ to {\rm (\ref{one})} with $\ga (x)\leq 1$ for all $x\in V$ (a state).  Consider a particular vertex $x\in V$ and let $y_1, \ldots , y_n$ be the neighbours of $x$.  Suppose $\phi$ is not constant on the star with internal vertex $x$.  Set $z_{\ell} = \phi (y_{\ell})-\phi (x)$ for $\ell = 1, \ldots n$, so that {\rm (\ref{gz})} is satisfied for some $\ga$.  Then for $N\geq 3$, provided $\ga <1$, we define the \emph{energy of $\phi$ at $x$} to be the quantity:
$$
\Ee (\phi ,x):= \frac{n||Z||^2}{\sum_{\ell}|z_{\ell}|^2} = \frac{1}{2(1-\ga )}\left\{ n - \ga (3-2\ga )\frac{|\sum_{\ell}z_{\ell}|^2}{\sum_{\ell}|z_{\ell}|^2}\right\}\,;
$$
for $n=N=2$, we define the energy by 
$$
\Ee (\phi , x) = 1+\cos \theta\,,
$$
where $\theta$ is the exterior angle at $x$.
If $\phi$ is constant on the star with internal vertex $x$, we take the energy to be zero.  The \emph{total energy} is defined to be the sum over the vertices of the energies at each vertex: $\Ee (\phi ) = \sum_{x\in V}\Ee (\phi ,x)$.
\end{definition}     
There are various ways to normalize; we have chosen to divide by the average length of the complex numbers $z_{\ell}$.  This is most convenient when $n=N=2$, when the energy has the above concise expression.  For $N\geq 3$, the formula in the definition is deduced from (\ref{ZtZ}):
$$
||Z||^2 = {\rm trace}\,Z^tZ = \rho + \si (1-u_1{}^2 - u_2{}^2)\,,
$$
and the expressions for $u_1$ and $u_2$ given by (\ref{identity-1}), where we recall that $\si = \ga \rho / (1-\ga )$.  For $n=2$, we have $u_1{}^2+u_2{}^2 = 1$, so that it is reasonable to replace $||Z||^2$ above by $\rho$.  From Section \ref{sec:cyclic}, specifically equation (\ref{spec-value}), when $n=2$, any solution to (\ref{one}) with $\ga <1$, necessarily corresponds to a configuration in the plane in which the two edges connecting the neighbours of $x$ have the same length, $r$ say.  But now if we refer to the figure above equation(\ref{spec-value}), then
$$
\Ee (\phi ,x) = \frac{\rho}{r}  =  1 - \frac{\cos \theta}{2(\cos \theta - 1)}|1-e^{-\ii\theta}|^2 
  =  1 + \cos \theta\,.
$$
In this case, we can allow states for which $\ga = 1$ which correspond to $\theta = \pm \pi$; these are characterized as having zero energy.  It is also important to note that the energy doesn't depend on the sign of the exterior angle.  We now wish to look more closely at the case when $n=N=2$ in order to understand the importance of isostates.

Consider a framework in the plane whose underlying graph is cyclic with $K$ edges.  Suppose the length of each edge is identical, so that the position function $\phi$ of the framework defines a solution to (\ref{one}).  We are interested in critical configurations for the energy:
$$
\Ee (\phi ) = K + \sum_{j = 1}^K\cos \theta_j\,,
$$
where $\theta_j$ is the exterior angle at vertex $x_j$.  In Appendix \ref{sec:trig}, we show that the regular configurations, that is the isostates, are critical.

Up to normalization, a three sided figure is completely determined and corresponds to an isostate with $\ga = 2/3$.  Four sided figures are determined up to two branches by one of the exterior angles $\theta$.  The two branches occur depending upon whether we choose $+\theta$ or $-\theta$ for the opposite exterior angle.  In the former case, the total energy is $4$, whatever the exterior angle $\theta \in (0, \pi )$; in the latter case it is $2+2\cos \theta$.  The two branches coalesce when $\theta = 0$.  The absolute minimum $\Ee =0$ occurs when the four-sided figure is completely folded up, that is when all exterior angles are $\pm \pi$.  This corresponds to an isostate with $\ga = 1$.  

As $K$ increases, the situation becomes more complicated.  Up to normalization, the configuration space of the framework is parametrized by $K-3$ exterior angles, however there are various branches that can occur which become more numerous as $K$ gets larger.   When $K$ is even, the absolute minimum of $\Ee$ is again zero and occurs when the figure folds up so all edges are superimposed and the exterior angles are all $\pm \pi$.  When $K$ is odd, the framework can no longer fold up in this way, but as we discuss in Appendix \ref{sec:trig}, the evidence suggests that the absolute minimum of $\Ee$ occurs when the figure folds up as best it can, that is with all exterior angles as close to $\pm\pi$ as possible, in a regular configuration.  For example, in the case of a five sided figure, one can easily check that the regular pentagon gives a local maximum for $\Ee$, whereas the regular star pentagon gives a local minimum.  These are precisely the isostates for a five sided figure (see Example \ref{ex:cyclic-5vertices}).          
 
 \medskip
 
 \noindent \emph{Time.}  Time is an ordering on a sequence of universes: $(\Ga, \Ga^{'}, \ldots )$.  The ordering must be compatible with the rules for change.  Thus $\Ga^{(j+1)}$ must derive from $\Ga^{(j)}$ by correlation, separation, mutation, or correlation with virtual point particles.  There are two ways to decide such an ordering:

(i)  The rules for change:  these may determine an irreversible process, such as that given in Example \ref{ex:evolution}.  The order $(\Ga , \Ga^{'}, \Ga^{''}, \Ga^{'''})$ is determined by the irreversibility of $\Ga^{''}\ra \Ga^{'''}$.

(ii)  A statistical parameter.  The \emph{thermal time hypothesis} has been developed by Connes and Rovelli \cite{Co-Ro}.  This is based on the Tomita flow associated to a von Neumann algebra.  In quantum field theory, the appropriate von Neumann algebra is the closure of the algebra of observables.  Then, given a state of a system over this algebra there is always a flow by which the state evolves and we may call this the ``flow of time" (see also \cite{Ro} \S 5.5.1).  In our context, we don't have an obvious von Neumann algebra that we can exploit.  However, there are various parameters that we may consider.  

Graph entropy is a well-know concept based on a probability distribution associated to the vertices \cite{Ko}.  Entropy is of course intimately related to the second law of thermodynamics.  Intrinsic curvature is also a natural parameter that occurs in smooth Riemannian geometry and curvature flow provides a way by which a manifold may evolve into one of uniform structure.

Curvature can provide a measure of local concentrations of structure.  In general, our various notions of curvature described in Section \ref{sec:curvature} depend on the field $\phi$ satisfying (\ref{one}).  However, given a particle in a particular state, we could envisage processes whereby the graph could evolve to uniformize the curvature -- say the vertex curvature, edge curvature or scalar curvature.

A simple curvature which depends only on the combinatorial structure is given, for a graph $\Ga = (V, E)$ with degree function $n:V \ra \NN$, by the function $n(x)-2$.  
That this can be considered as a measure of curvature appears to have first been suggested in \cite{Ur-1}.  In \cite{Ba-Ti}, an algorithm is given whereby the quantity $\sum_{x\in V}(n(x)-2)^2$ may be minimized subject to $\sum_{x\in V}n(x)$ remaining constant, by a process of sliding edges.  This procedure preserves connectedness and may be viewed as a discrete analogue of the scalar curvature flow in Riemannian geometry.  This leads to the parameter 
$$
t(\Ga ) := \frac{\sqrt{}\left\{ \sum_{x\in V} (n(x)-2)^2\right\}}{2|E|}\,,
$$
as a possible measure of thermal time.  In Example \ref{ex:evolution}, the passage from $\Ga\ra \Ga^{'}$ increases $t$, whereas the passage from $\Ga^{'} \ra \Ga^{''}$ decreases $t$, but we don't preclude local increases in $t$.  Indeed, time should be a statistically \emph{dominant} parameter that appears at a macroscopic level. 
 
\appendix

\section{Invariant structures}  \label{sec:invariant-structures} For our purposes, an \emph{elementary invariant framework} is a framework corresponding to the $1$-skeleton of a either a regular polytope, or the invariant double cones of Section \ref{sec:inv-polytopes}, possibly with the two apexes connected with an edge.  An \emph{invariant structure} is an invariant framework made up of elementary components, the components connected by edges in an appropriate way.  Motivated by the last section, we also require that the corresponding function $\ga$ be constant.  Thus, invariant structures are objects that may populate our elementary universe.

Consider an invariant star in $\RR^N$ with internal vertex located at the origin and with $n$ external vertices.  Let $\vec{x}\in \RR^N$ be the centre of mass of the external vertices.  Suppose that $\vec{x}\neq \vec{0}$.  Then we call the ray through the origin generated by $\vec{x}$ the \emph{axis of the star}.  Let $b>0$ denote the distance of the centre of mass from the origin along this axis.

\begin{lemma} \label{lem:extension-star}  The addition of a new external vertex at any point other than $-nb$ along the axis of the star produces a new invariant star.  Furthemore, if $\ga$ denotes the invariant of the original star and $x\in \RR$ is the position along the axis of the new vertex, then the new star invariant is given by 
\begin{equation} \label{ga-ext}
\wt{\ga} = \frac{(n+1)(x^2+nb^2\ga )}{(x+nb)^2}\,.
\end{equation}
\end{lemma}

\begin{proof}  Without loss of generality, we may suppose that the centre of mass of the star lies along the $y_N$-axis.  In particular, if $\vec{v}_1, \ldots , \vec{v}_n$ denote the external vertices, then
$$
\sum_{\ell = 1}^n\vec{v}_{\ell} = nb\vec{e}_N\,.
$$
We now add a new vertex at the point $x\vec{e}_N$, for some $x\in \RR$.  Thus the new star matrix is given by
$$
\left( \vec{v}_1|\cdots |\vec{v}_n|x\vec{e}_N\right)\,.
$$
As usual, let $A = (a_{jk})$ be an arbitrary orthogonal transformation of $\RR^N$ and let $P:\RR^N\ra \CC$ be the projection $P(y_1, \ldots , y_N) = y_1 + \ii y_2$.  Set $z_{\ell} = P\circ A(\vec{v}_{\ell})$ for $\ell = 1, \ldots , n$ and $z_{n+1} = P\circ A(a\vec{e}_N)$.  Then
$$
z_{\ell} = \sum_{j=1}^N(a_{1j} + \ii a_{2j})v_{\ell j}\,, \qquad z_{n+1} = x(a_{1N} + \ii a_{2N})\,.
$$
Furthermore, $\sum_{\ell = 1}^nz_{\ell} = nb(a_{1N} + \ii a_{2N})$, so that
$$
\sum_{\ell = 1}^nz_{\ell}{}^2 = \frac{\ga}{n} \left(\sum_{\ell =1}^nz_{\ell}\right)^2 = \ga nb^2(a_{1N} + \ii a_{2N})^2\,,
$$
where $\ga$ is the invariant of the original star.  We require that there is a real number $\wt{\ga}$ such that
$$
\frac{\wt{\ga}}{n+1}\left( z_{n+1} + \sum_{\ell = 1}^nz_{\ell}\right)^2 = z_{n+1}{}^2 + \sum_{\ell = 1}^nz_{\ell}{}^2\,.
$$
But this is uniquely given by (\ref{ga-ext}). 
\end{proof}

We note that as $x$ approaches $-nb$, then $\wt{\ga}$ becomes arbitrary large.  Indeed, when $x = -nb$, then we have harmonicity at the internal vertex of the new star, so that $\wt{\ga}$ is not well-defined in this case.

Invariant structures now arise by connecting invariant frameworks with edges in an appropriate way.  We can use the invariant double cones of Section \ref{sec:inv-polytopes}, as well as regular polytopes to produce new structures.  There are various ways in which this can be done; as we don't have an exhaustive classification, we will consider some examples of geometric interest. 

Consider first the case of a double cone on a regular polygon.  The corresponding framework in $\RR^3$ satisfies (\ref{one}) with $\ga$ in general having a different value at the apexes to the value at the lateral vertices.  We extend the double cone by adding a new edge to each apex along the axis of the cone, attaching new double cones to each of these to produce an infinite family of double cones along a common axis.  Suppose the double cone is aligned along the $x_3$-axis.  Let $x$ denote a variable along this axis which is zero at the topmost apex and which increases towards the centre of mass of the vertex figure.  Place a new vertex at position $x$ along the axis.  Then from (\ref{ga-ext}) and (\ref{ga-apex}), the new value of $\ga$ at the apex becomes:
$$
\wt{\ga}_{\rm apex} = \frac{(n+1)(x^2+n\sin^2\frac{2\pi}{n} - \frac{n}{2})}{(x+n\sin\frac{2\pi}{n})^2}\,.
$$
We require this to equal the lateral value of $\ga$ given by (\ref{ga-lat}).  This determines the quadratic equation in $x$:
\begin{equation} \label{quadratic}
\frac{2x^2-n\cos\frac{4\pi}{n}}{(x+n\sin\frac{2\pi}{n})^2} = \frac{2-2\cos\frac{2\pi}{n} + \cos \frac{4\pi}{n}}{(2-\cos\frac{2\pi}{n})^2}\,.
\end{equation}
When $n=3$, this has two distinct roots given by $x=\sqrt{3}$ and $x=-\sqrt{3}/4$.  When we take the root $x = \sqrt{3}$, then $x$ lies precisely at the bottom vertex and the edge joining $x$ to the top vertex connects the two apexes.  In fact, since the triangle is a complete graph on three vertices, we have recovered the case of Corollary \ref{cor:inv-complete-graph} with $n=3$.  However, we can continue to add additional copies of the double cone to obtain a curious structure, as follows.

We begin with one double cone $C_1$ and attach it to a new one $C_2$ so the top apex of $C_1$ is joined to the bottom apex of $C_2$ at a distance $\sqrt{3}$ along the axis; thus the bottom apex of $C_2$ is situated at the bottom apex of $C_1$.  The two cones are therefore superimposed, with a new edge running down the centre.  The value of $\ga$ at the top apex of $C_1$ is equal to that at the bottom apex of $C_2$, which is equal to $\ga_{\rm lat}$.  

However, when we join the cones, we are at liberty to perform an arbitrary rotation of $C_2$ with respect to $C_1$ without affecting invariance.  If we perform a relative rotation through an irrational multiple of $2\pi$ and perform the same relative rotation on joining a new cone $C_3$ to $C_2$ and so on, we obtain countably many double cones with common central axis.  The closure of this set consists of the solid region enclosed by two (round) double cones, together with the segment $A$ of the central axis from $\sqrt{3}$ to $0$.  We note that even though $A$ constitutes a common edge joining successive cones, as a framework, each edge is distinct.  

For $n=4$, equation (\ref{quadratic}) has no real solutions. This also turns out to be the case for $n=5$, however for $n\geq 6$, there are two distinct real solutions, which enable us to connect double cones with edges along the axis of symmetry to obtain invariant structures.  For example, for $n=6$, we obtain the quadratic:
$$
118x^2-48\sqrt{3}\,x - 27 = 0\,,
$$
with roots, one positive, one negative, given by
$$
x = \frac{24\sqrt{3} \pm \sqrt{546}}{118}\,.
$$
Taking the positive root, on connecting successive cones we can perform a rotation to obtain interlaced frameworks.

As a final construction, we reconsider double cones on regular polytopes and attach new double cones both laterally and vertically, interpolating the two distances in such a way that $\ga$ is constant and the resulting structures are invariant.  In order to do this, we need to be able to periodically position the regular polytopes in $\RR^{N-1}$ in an appropriate way.  For example, for regular polygons, this amounts to finding an appropriate periodic tiling in the plane.  Three examples are illustrated below.

\medskip
\begin{center}
\setlength{\unitlength}{0.200mm}
\begin{picture}(580,200)(10,-230)
        \allinethickness{0.254mm}\path(80,-180)(120,-180) 
        \allinethickness{0.254mm}\path(120,-180)(145,-155) 
        \allinethickness{0.254mm}\path(120,-180)(150,-150) 
        \allinethickness{0.254mm}\path(150,-150)(150,-110) 
        \allinethickness{0.254mm}\path(150,-110)(120,-80) 
        \allinethickness{0.254mm}\path(120,-80)(80,-80) 
        \allinethickness{0.254mm}\path(80,-80)(50,-110) 
        \allinethickness{0.254mm}\path(50,-110)(50,-150) 
        \allinethickness{0.254mm}\path(50,-150)(80,-180) 
        \allinethickness{0.254mm}\path(80,-80)(80,-40) 
        \allinethickness{0.254mm}\path(80,-40)(120,-40) 
        \allinethickness{0.254mm}\path(120,-40)(120,-80) 
        \allinethickness{0.254mm}\path(50,-110)(10,-110) 
        \allinethickness{0.254mm}\path(10,-110)(10,-150) 
        \allinethickness{0.254mm}\path(10,-150)(50,-150) 
        \allinethickness{0.254mm}\path(80,-180)(80,-220) 
        \allinethickness{0.254mm}\path(80,-220)(120,-220) 
        \allinethickness{0.254mm}\path(120,-220)(120,-180) 
        \allinethickness{0.254mm}\path(150,-150)(190,-150) 
        \allinethickness{0.254mm}\path(190,-150)(190,-110) 
        \allinethickness{0.254mm}\path(190,-110)(150,-110) 
        \allinethickness{0.254mm}\path(295,-185)(325,-185) 
        \allinethickness{0.254mm}\path(365,-145)(365,-115) 
        \allinethickness{0.254mm}\path(325,-75)(295,-75) 
        \allinethickness{0.254mm}\path(255,-115)(255,-145) 
        \allinethickness{0.254mm}\path(325,-185)(350,-170) 
        \allinethickness{0.254mm}\path(350,-170)(365,-145) 
        \allinethickness{0.254mm}\path(365,-115)(350,-90) 
        \allinethickness{0.254mm}\path(350,-90)(325,-75) 
        \allinethickness{0.254mm}\path(255,-115)(270,-90) 
        \allinethickness{0.254mm}\path(295,-75)(270,-90) 
        \allinethickness{0.254mm}\path(255,-145)(270,-170) 
        \allinethickness{0.254mm}\path(270,-170)(295,-185) 
        \allinethickness{0.254mm}\path(325,-75)(310,-50) 
        \allinethickness{0.254mm}\path(310,-50)(295,-75) 
        \allinethickness{0.254mm}\path(270,-90)(240,-90) 
        \allinethickness{0.254mm}\path(240,-90)(255,-115) 
        \allinethickness{0.254mm}\path(350,-90)(380,-90) 
        \allinethickness{0.254mm}\path(380,-90)(365,-115) 
        \allinethickness{0.254mm}\path(350,-170)(380,-170) 
        \allinethickness{0.254mm}\path(380,-170)(365,-145) 
        \allinethickness{0.254mm}\path(270,-170)(240,-170) 
        \allinethickness{0.254mm}\path(240,-170)(255,-145) 
        \allinethickness{0.254mm}\path(295,-185)(310,-210) 
        \allinethickness{0.254mm}\path(310,-210)(325,-185) 
        \allinethickness{0.254mm}\path(500,-155)(530,-155) 
        \allinethickness{0.254mm}\path(500,-105)(530,-105) 
        \allinethickness{0.254mm}\path(530,-155)(545,-130) 
        \allinethickness{0.254mm}\path(530,-105)(545,-130) 
        \allinethickness{0.254mm}\path(500,-105)(485,-130) 
        \allinethickness{0.254mm}\path(485,-130)(500,-155) 
        \allinethickness{0.254mm}\path(500,-105)(485,-80) 
        \allinethickness{0.254mm}\path(485,-80)(500,-55) 
        \allinethickness{0.254mm}\path(500,-55)(530,-55) 
        \allinethickness{0.254mm}\path(530,-55)(545,-80) 
        \allinethickness{0.254mm}\path(545,-80)(530,-105) 
        \allinethickness{0.254mm}\path(500,-155)(485,-180) 
        \allinethickness{0.254mm}\path(485,-180)(500,-205) 
        \allinethickness{0.254mm}\path(500,-205)(530,-205) 
        \allinethickness{0.254mm}\path(530,-205)(545,-180) 
        \allinethickness{0.254mm}\path(545,-180)(530,-155) 
        \allinethickness{0.254mm}\path(485,-80)(455,-80) 
        \allinethickness{0.254mm}\path(455,-80)(440,-105) 
        \allinethickness{0.254mm}\path(440,-105)(455,-130) 
        \allinethickness{0.254mm}\path(455,-130)(485,-130) 
        \allinethickness{0.254mm}\path(455,-130)(440,-155) 
        \allinethickness{0.254mm}\path(440,-155)(455,-180) 
        \allinethickness{0.254mm}\path(455,-180)(485,-180) 
        \allinethickness{0.254mm}\path(545,-80)(575,-80) 
        \allinethickness{0.254mm}\path(575,-80)(590,-105) 
        \allinethickness{0.254mm}\path(590,-105)(575,-130) 
        \allinethickness{0.254mm}\path(575,-130)(545,-130) 
        \allinethickness{0.254mm}\path(575,-130)(590,-155) 
        \allinethickness{0.254mm}\path(590,-155)(575,-180) 
        \allinethickness{0.254mm}\path(575,-180)(545,-180) 
        \allinethickness{0.254mm}\path(310,-50)(310,-30) 
        \allinethickness{0.254mm}\path(240,-90)(220,-75) 
        \allinethickness{0.254mm}\path(380,-90)(400,-75) 
        \allinethickness{0.254mm}\path(380,-170)(400,-185) 
        \allinethickness{0.254mm}\path(310,-210)(310,-230) 
        \allinethickness{0.254mm}\path(240,-170)(220,-185) 
\end{picture}
\end{center}
\medskip

In the first figure, we consider the square, bisect its vertex figure at each vertex to form a complementary octagon and so tile the plane with squares and octagons.  For the second figure, we take a triangle, bisect its vertex figure at each vertex to form a complementary dodecagon ($12$-sided figure).  Finally, the complementary figure of a hexagon is the hexagon itself.  Our aim is to form an infinite structure in $\RR^3$, by first taking the double cone on each of these polygons (square, triangle and hexagon); then taking an identical copy and attaching it to the first by vertical edges of appropriate length which join the apexes, so forming succesive layers.

The first observation, is that the height of the double cone as given by Theorem \ref{thm:double-cone} is not affected by the addition of a new edge and vertex bisecting the vertex figure. Indeed, using the notation of the proof of Theorem \ref{thm:double-cone} and taking the more general situation of that theorem of a double cone on a regular polytope, the star matrix (\ref{star-double-cone}) is adjusted to 
$$
S= \left( \begin{array}{c|c|c|c|c|r|r} \vec{v}_1 & \vec{v}_2 & \cdots & \vec{v}_n & \vec{0} & \vec{0} & \vec{0} \\
 c & c & \cdots & c & a & a & -d \\ 0 & 0 & \cdots & 0 & b & -b & 0 \end{array} \right)
 $$
 where $d$ denotes the distance of the new vertex along the axis of symmetry (measured away from the centre of the vertex figure).  In addition to the projections $z_1, \ldots , z_n, z_{n+1}, z_{n+2}$, we now have an additional vertex which projects to
 $$
 z_{n+3} = -d(a_{1N} + \ii a_{2N})\,.
 $$
 The affect of this is to give the new sums:
 $$
 \sum_{\ell = 1}^{n+3}z_{\ell} = (nc+2a-d)(a_{1N} + \ii a_{2N})\,,
 $$
 and
 $$
 \sum_{\ell = 1}^{n+3} z_{\ell}{}^2 = (nc^2+2a^2+d^2-\rho )(a_{1N} + \ii a_{2N})^2 + (2b^2-\rho )(a_{1,N+1} + \ii a_{2,N+1})^2\,.
 $$
 Thus, as before, invariance requires $b = \sqrt{\rho /2}$ and the new value of $\ga_{\rm lat}$ is given by 
 $$
 \wt{\ga}_{\rm lat} = \frac{(n+3)(nc^2+2a^2+d^2-\rho )}{(nc + 2a - d)^2}\,.
 $$
 
 If we take as an example the double cone on a regular triangle whose vertices are at a distance $1$ from its centre, then we have
 $$
 a=1,\quad c = 2\sin^2\frac{2\pi}{3} = \frac{3}{2}, \quad \rho = 2\sin^2\frac{2\pi}{3} = \frac{3}{2}\,.
 $$
 If we impose the tiling given by triangles and dodecagons, then $d = \sqrt{3}$, so that 
 $$
 \wt{\ga}_{\rm lat} = \frac{40}{(5-\sqrt{3})^2}\,.
 $$
 On the other hand, we can calculate $\ga$ at the apex, after we add on another edge according to Lemma \ref{lem:extension-star}.  To do this we apply (\ref{double-cone-ga-apex}) and (\ref{ga-ext}) to give
$$
\wt{\ga}_{\rm apex} = \frac{(m+1)(x^2+mb^2- \rho_{\Pp})}{(x+mb)^2} = \frac{8x^2+6}{(\sqrt{2}\,x - 3\sqrt{3/2})^2}\,,
$$
where we recall $m=3$ is the cardinality of the regular polytope $\Pp$ (in this case the triangle).  Equating $\wt{\ga}_{\rm lat}$ and $\wt{\ga}_{\rm apex}$ yields the two (real) roots of the quadratic:
$$
4(9-5\sqrt{3})x^2+60\sqrt{3}\,x -93-15\sqrt{3}\,.
$$
We now stack layers at either of these distances in order to obtain an invariant structure.

One can generalize this procedure to higher dimension, whenever we can fill out the corresponding Euclidean space in an appropriate way.  This can be done in any dimension for the hypercube and the cross-polytope, where we attach edges along the axes of symmetry of the various vertex figures.  In $\RR^3$, we can also take the tetrahedron with vertices placed at the points given by (\ref{tetrahedron-coords}).  Note that these lie at those vertices of a cube (of edge length $2$) diagonally opposite across each face.  We then fill out $\RR^3$ with identical cubes, whose vertices are placed on a lattice with odd integer components.  Begin by placing a tetrahedron in one of these cubes; then the axes of symmetry of the vertex figures cross adjacent cubes along diagonals which connect vertices which are opposite with respect to the centre of the cube.  We obtain the $3$-dimensional analogue of the tiling of the plane by squares and octagons illustrated above.  By taking double cones on these tetrahedra and connecting them at the apexes with edges of appropriate length, we then obtain an invariant structure in $\RR^4$.  We omit the detailed calculations of the edge lengths, which proceed as in the example above.

\section{Some planar trigonometry} \label{sec:trig}  In this section we parametrize the configuration space of the framework corresponding to an $M$-sided planar polygonal figure with edges of common length, and show that the regular figures are extremal with respect to the energy functional $\Ee$ defined in Section \ref{sec:particles}.  

Consider such a framework with $M\geq 5$ with edges of unit length labeled as in the figure below.  The first bar has endpoints $0$ and $1$ in the complex plane, then the next $M-3$ exterior angles are labeled by $\si_1, \si_2 , \ldots , \si_{M-3}\in [-\pi , \pi ]$ and the last three exterior angles by $\theta_1, \theta_2, \theta_3$.  We define the quantities
$$
\begin{array}{rcl}
s & = & |1+e^{\ii\si_1} + e^{\ii (\si_1+\si_2)} + \cdots + e^{\ii (\si_1 + \si_2 + \cdots + \si_{M-3})}| \\
r & = & |1+e^{\ii\si_1} + e^{\ii (\si_1+\si_2)} + \cdots + e^{\ii (\si_1 + \si_2 + \cdots + \si_{M-4})}|  \\
t & = & |e^{\ii\si_1} + e^{\ii (\si_1+\si_2)} + \cdots + e^{\ii (\si_1 + \si_2 + \cdots + \si_{M-3})}| \\
 & = & |1+e^{\ii\si_2} + e^{\ii (\si_2+\si_3)} + \cdots + e^{\ii (\si_2 + \si_2 + \cdots + \si_{M-3})}| 
\end{array}
$$

\medskip
\begin{center}
\setlength{\unitlength}{0.254mm}
\begin{picture}(396,245)(65,-260)
        \allinethickness{0.254mm}\path(160,-220)(280,-220) 
        \allinethickness{0.254mm}\path(280,-220)(385,-160) 
        \allinethickness{0.254mm}\path(385,-160)(430,-50) 
        \allinethickness{0.254mm}\path(160,-220)(100,-120) 
        \allinethickness{0.254mm}\path(100,-120)(190,-30) 
        \allinethickness{0.254mm}\path(190,-30)(310,-30) 
        \allinethickness{0.254mm}\dottedline{5}(430,-50)(215,-15) 
        \allinethickness{0.254mm}\dottedline{5}(190,-30)(140,-30) 
        \allinethickness{0.254mm}\dottedline{5}(100,-120)(65,-155) 
        \allinethickness{0.254mm}\dottedline{5}(160,-220)(185,-260) 
        \allinethickness{0.254mm}\dottedline{5}(280,-220)(330,-220) 
        \allinethickness{0.254mm}\dottedline{5}(385,-160)(455,-120) 
        \allinethickness{0.254mm}\dashline{5}[5](190,-30)(160,-220) 
        \allinethickness{0.254mm}\dashline{5}[5](310,-30)(160,-220) 
        \allinethickness{0.254mm}\dashline{5}[5](190,-30)(230,-115) 
        \allinethickness{0.254mm}\dashline{5}[5](280,-220)(240,-135) 
        \put(260,-166){\shortstack{$t$}} 
        \put(275,-91){\shortstack{$r$}} 
        \put(185,-121){\shortstack{$s$}} 
        \put(175,-236){\shortstack{$\theta_3$}} 
        \put(310,-216){\shortstack{$\sigma_1$}} 
        \put(397,-141){\shortstack{$\sigma_2$}} 
        \put(215,-27){\shortstack{$\sigma_{M-3}$}} 
        \put(155,-46){\shortstack{$\theta_1$}} 
        \put(95,-146){\shortstack{$\theta_2$}} 
        \put(170,-61){\shortstack{$\frac{\theta_2}{2}$}} 
        \put(150,-196){\shortstack{$\frac{\theta_2}{2}$}} 
        \put(170,-190){\shortstack{$\alpha$}} 
\end{picture}
\end{center}
\medskip   

\begin{proposition} \label{prop:5-sided}  The angles $\si_1, \ldots , \si_{M-3}$ parametrize a closed polygonal bar framework if and only if $s\leq 2$, in which case $\cos \theta_2$ is determined and provided $s\neq 0$, $\cos\theta_1$ and $\cos\theta_3$ have a two-fold ambiguity; the energy
\begin{equation} \label{energy-polygon}
\begin{array}{l}
\Ee := M+\left(\sum_{j=1}^{M-3}\cos\si_j\right) + \cos\theta_1 + \cos\theta_2 + \cos\theta_3 = \\
M-\frac{3}{2} + \left(\sum_{j=1}^{M-3}\cos\si_j\right) + \frac{r^2+t^2}{4} \\
+ \frac{\sqrt{4-s^2}}{2s} \left\{ \sum_{j=1}^{M-3}[\sin (\si_1+ \si_2+ \cdots + \si_j) +  \sin (\si_j + \si_{j+1} + \cdots + \si_{M-3})] \right\}
\end{array}
\end{equation} 
where the choice of sign of the square root corresponds to two possible configurations of the framework.  If $s=0$, there are infinitely many possible configurations given by $\theta_2= \pm\pi$ with $\theta_1$ arbitrary. 
\end{proposition} 
We postpone the proof of this proposition until the end of this section.

Note that the sum of the exterior angles may jump from $2\pi$ to $4\pi$ with a continuous deformation of the framework, however, $\Ee$ varies continuously, as indicated in the sketch below for $M=5$.

\medskip
\begin{center}
\setlength{\unitlength}{0.254mm}
\begin{picture}(240,65)(80,-150)
        \allinethickness{0.254mm}\path(160,-140)(125,-85) 
        \allinethickness{0.254mm}\path(115,-85)(80,-140) 
        \allinethickness{0.254mm}\path(115,-85)(120,-145) 
        \allinethickness{0.254mm}\path(120,-145)(125,-85) 
        \allinethickness{0.254mm}\path(160,-140)(125,-140) 
        \allinethickness{0.254mm}\path(115,-140)(80,-140) 
        \allinethickness{0.254mm}\path(270,-95)(275,-100) 
        \allinethickness{0.254mm}\path(290,-95)(280,-150) 
        \allinethickness{0.254mm}\path(280,-150)(270,-95) 
        \allinethickness{0.254mm}\path(240,-140)(270,-140) 
        \allinethickness{0.254mm}\path(320,-140)(290,-140) 
        \allinethickness{0.254mm}\path(240,-140)(265,-115) 
        \allinethickness{0.254mm}\path(290,-95)(275,-110) 
        \allinethickness{0.254mm}\path(320,-140)(290,-115) 
        \allinethickness{0.254mm}\path(185,-110)(220,-110)\special{sh 1}\path(220,-110)(214,-108)(214,-110)(214,-112)(220,-110) 
\end{picture}
\end{center}
\medskip
For the framework on the left, the sum of the exterior angles is $2\pi$, whereas for the right-hand figure, it is $4\pi$.  At the point of transition, we have $\si_1 = \si_2 = 2\pi /3$.  Then both $s$ and the curly bracket (which equals $2(\sin \frac{2\pi}{3} + \sin \frac{4\pi}{3})$) vanish in the expression for $\Ee$.  However, the singularity is removable and $\Ee = 3-\sqrt{3}$ is well-defined and continuous at this point.   

\begin{corollary}  \label{cor:polygons}  The regular polygon and star polygons represent extrema for the functional $\Ee$.
\end{corollary}

\begin{proof}  The angles $\si_{\ell}$ are not all on an equal footing, so in order to determine critical points, we establish a recursive formula on the derivatives.  First note that
\begin{eqnarray*}
s & = & |1+e^{\ii\si_1} + e^{\ii (\si_1+\si_2)} + \cdots + e^{\ii (\si_1 + \si_2 + \cdots + \si_{M-3})}| \\
 & = & \surd \{ (1+\cos \si_1 + \cos (\si_1+\si_2) + \cdots + \cos (\si_1 + \si_2 + \cdots + \si_{M-3}))^2 \\
& & \qquad + (\sin \si_1 + \sin (\si_1+\si_2) + \cdots + \sin (\si_1 + \si_2 + \cdots + \si_{M-3}))^2 \} \\
  & = &  \surd\Big\{ N-2 + 2 \sum_{k=1}^{M-3}\sum_{j = k}^{M-3} \cos (\si_k + \si_{k+1} + \cdots + \si_j)\Big\}
\end{eqnarray*} 
where the latter equality follows from using $\cos \si_1\cos (\si_1+\si_2) + \sin \si_1\sin (\si_1+\si_2) = \cos \si_2$ and so on.  Therefore, for each $\ell = 1, \ldots , M-3$, 
$$
\frac{\pa s}{\pa \si_{\ell}} = - \frac{1}{s} \sum_{k = 1}^{\ell}\sum_{j = \ell}^{M-3} \sin (\si_k+ \si_{k+1} + \cdots \si_j)\,,
$$
and we have the recursive formula
$$
\frac{\pa s}{\pa \si_{\ell}} = \frac{\pa s}{\pa\si_{\ell -1}} + \frac{1}{s} \sum_{k=1}^{\ell - 1} \sin (\si_k + \cdots + \si_{\ell -1}) - \frac{1}{s}\sum_{j = \ell}^{M-3}\sin (\si_{\ell} + \cdots + \si_j)\,,
$$
where, for $\ell = 1$, we set the first two terms on the right-hand side equal to zero. 
By the same reasoning, for $\ell = 1 , \ldots , M-4$,
$$
\frac{\pa r}{\pa \si_{\ell}} = \frac{\pa r}{\pa\si_{\ell -1}} + \frac{1}{r} \sum_{k=1}^{\ell - 1} \sin (\si_k + \cdots + \si_{\ell -1}) - \frac{1}{r}\sum_{j = \ell}^{M-4}\sin (\si_{\ell} + \cdots + \si_j)\,,
$$
with $\pa r / \pa \si_{M-3} = 0$, and for $\ell = 2, \ldots M-3$,
$$
\frac{\pa t}{\pa \si_{\ell}} = \frac{\pa t}{\pa\si_{\ell -1}} + \frac{1}{t} \sum_{k=2}^{\ell - 1} \sin (\si_k + \cdots + \si_{\ell -1}) - \frac{1}{t}\sum_{j = \ell}^{M-3}\sin (\si_{\ell} + \cdots + \si_j)\,,
$$
with $\pa t/\pa \si_1 = 0$.  From (\ref{energy-polygon}), we now obtain for $\ell = 2, \ldots , M-3$,
\begin{equation} \label{ee-rec}
\begin{array}{l}
\ds \frac{\pa\Ee}{\pa\si_{\ell}}  =  \frac{\pa\Ee}{\pa \si_{\ell -1}} + \sin \si_{\ell -1} - \sin \si_{\ell} - \frac{1}{2}\sin (\si_1+ \cdots + \si_{\ell -1}) \\
 + \frac{1}{2}\sin (\si_{\ell} + \cdots + \si_{M-3}) 
  - \frac{\sqrt{4-s^2}}{2s} [\cos (\si_1 + \cdots + \si_{\ell -1}) - \cos (\si_{\ell} + \cdots + \si_{M-3})] \\
+ \left( 1 - \frac{2T}{s^3\sqrt{4-s^2}}\right) \left( \sum_{k=1}^{\ell -1}\sin (\si_k + \cdots + \si_{\ell -1}) - \sum_{j=\ell}^{M-3}\sin (\si_{\ell} + \cdots + \si_j)\right) \,,
\end{array}
\end{equation}
where we have set
$$
T:= \sum_{j=1}^{M-3}[\sin (\si_1 + \cdots + \si_j) + \sin (\si_j + \cdots + \si_{M-3})]\,.
$$
We now claim that under the hypothesis that all angles $\si_{\ell} = \ta$ are equal with $M\ta = 2k\pi$, for some integer $k$ satisfying $0<k\leq [M/2]$, then all derivatives $\pa \Ee / \pa \si_{\ell}$ vanish.

The case $\ta = \pm\pi$ can only occur when $M$ is even, in which case it represents an absolute minimum for $\Ee$ and is certainly critical; so henceforth, suppose that $\ta \neq \pm \pi$.  With all angles equal, we must take the positive sign for the square root in (\ref{energy-polygon}); indeed, this is necessary to obtain $M-M\cos \ta$ for the value of $\Ee$.  The following identities are useful:
$$
\begin{array}{rl}
 & 1+e^{\ii\ta} + e^{2\ii\ta} + \cdots + e^{(M-1)\ii\ta} = 0 \\
\Rightarrow & \sum_{j = 1}^{M-3} \sin (j\ta ) = - \sin (M-2)\ta - \sin (M-1)\ta = \sin \ta (1+2\cos \ta ) \\
{\rm similarly} & \sum_{j=1}^{M-3}\cos (j\ta ) = - \cos\ta (1+2\cos\ta ) \\
 \Rightarrow & T = 2\sin\ta (1+2\cos\ta )\,,
\end{array}
$$
and
$$
s=2\cos \frac{\ta}{2} \qquad r = t = |1+2\cos\ta |\,.
$$
On noting that $\pa t/\pa \si_1=0$, for the derivative with respect to $\si_1$, we obtain
\begin{eqnarray*}
\frac{\pa\Ee}{\pa \si_1} & = & - \sin \si_1 + \frac{r}{2} \frac{\pa r}{\pa \si_1} 
 -\frac{4}{s^2\sqrt{4-s^2}} \frac{\pa s}{\pa \si_1}  \sum_{j=1}^{M-3}\sin (j\ta ) \\
 & & + \frac{\sqrt{4-s^2}}{2s}\Big\{ \cos (M-3)\ta + \sum_{j=1}^{M-3} \cos (j\ta ) \Big\} \\
 & = & 0\,, 
\end{eqnarray*} 
where the last equality follows after substitution and routine calculations.  Now apply recurrence on $\pa \Ee / \pa \si_{\ell}$, so suppose that $\pa \Ee / \pa \si_{\ell - 1} = 0$.  Then from (\ref{ee-rec}) we have,
\begin{eqnarray*}
\frac{\pa \Ee}{\pa \si_\ell} & = & -\frac{1}{2} \sin (\ell - 1)\ta - \frac{1}{2} \sin (2+\ell ) \ta + \frac{\cos\ta}{1+\cos\ta}[\sin \ell \ta + \sin (\ell + 1)\ta ] \\
 & & \quad + \frac{\sin\ta}{2(1+\cos\ta )} [ - \cos (\ell -1)\ta + \cos (2+\ell )\ta ] \,.
\end{eqnarray*}
But now multiplication through by $2(1+\cos\ta )$ shows that this vanishes.  By recurrence, all derivatives $\pa \Ee / \pa \si_{\ell}$ vanish for $\ell = 1, \ldots , M-3$.       
\end{proof}

Included in the list of ``star polygons" are ones that correspond to several circuits of a regular polygon.  For example, when $M=6$ and in the notation of the above proof $\ta = 2\pi /3$, the solution corresponds to two circuits of a regular triangle as illustrated in Example \ref{ex:hexagon}.  Such cases occur when $M$ has non-trivial factors.  When $M$ is even, we always have the solution $\ta = \pm \pi$, where the polygonal chain ``folds up" to cover the segment from $0$ to $1$.  Then the energy is zero and is at an absolute minimum.  We conjecture that the regular convex polygon given by $\ta = 2\pi /M$ is a global maximum for $\Ee$ and for $M$ odd, the star polygon given by $\ta = 2[M/2]\pi / M$ is a global minimum; but we don't have proofs of this in general.  Also, we do not know if there are other critical configurations for $\Ee$ other than the regular ones.

When $M=5$, there are just two possible regular configurations.  For $\si_1 = \si_2 = 2\pi /5$, we have $\Ee = 5(3+\sqrt{5})/4$ and for $\si_1 = \si_2 = 4\pi /5$, we have $\Ee = 5(3-\sqrt{5})/4$.  It is now possible to calculate explicitly the second derivatives of $\Ee$ at these critical points and to verify directly that these configurations correspond respectively to a local maximum and local minimum of $\Ee$.  However, this approach for general $M$ presents formidable computational difficulties.  

Our introduction of the functional $\Ee$ suggests the intriguing possibility of setting up a gradient flow on the configuration space of positions of the framework parametrized by the $\si_{\ell}$, whereby the figure evolves into one of the regular configurations (depending upon whether we take the gradient flow or its inverse).  The combinatorial solution to such problems of regularizing polygonal frameworks, has been derived in the references \cite{Sa} and \cite{Co-De-Ro}.  But we are not aware of a functional analytic approach.  An evolution from the pentagon to the star pentagon with reflectional symmetry, is obtained by setting $\si_1 = \si_2$ and allowing $\si_1$ to vary continuously from $2\pi /5$ to $4\pi /5$.  Two of the positions of this evolution are illustrated in the figure above.

\medskip

\noindent \emph{Proof of Proposition.}\  Without loss of generality we can take the length of each edge to be unity.  Since $s = 2\cos \frac{\ta_2}{2}$, $\cos \ta_2$ is completely determined by the data $\si_1, \ldots , \si_{M-3}$ and is given by
$$
\cos \ta_2 = \frac{s^2-2}{2}\,.
$$
even though $\ta_2$ is only determined up to sign.  We now find expressions for $\ta_1$ and $\ta_3$.  Suppose $s\neq 0$.  Note that this implies that $r\neq 0$ as well.  Let $\al$ be the angle indicated in the figure, being the interior angle of the triangle with adjacent edges of lengths $r$, $s$ and opposite edge of length $1$.  Set $A=\ta_1+\frac{\ta_2}{2}$.  Then application of the sine rule gives the equalities:
$$
\sin \al = \frac{\sin (A - \al )}{s} = \frac{\sin A}{r}\,.
$$
One can eliminate the terms involving $\al$ from this expression to obtain
$$
\cos A = \frac{r^2-s^2-1}{2s}\,.
$$
Once more this is unambiguously defined, even though $A$ itself has a sign ambiguity.  This is underlined by the expression for $\cot A$, which has the ambiguity attached:
$$
\cot A = \pm \frac{r^2-s^2-1}{\sqrt{4s^2-(r^2-s^2-1)^2}}\,.
$$
However,
$$
\cos\ta_1  =  \cos (A-\frac{\ta_2}{2}) =  (1+\cot A \cot \frac{\ta_2}{2})\sin A\sin \frac{\ta_2}{2}\,,
$$
from which we deduce that
\begin{eqnarray*}
\cos \ta_1 & = &  \frac{1 + \cot A\cot \frac{\ta_2}{2}}{\sqrt{(1+\cot^2A)(1+\cot^2\frac{\ta_2}{2})}} \\
 & = & \frac{1}{4s}\{ \sqrt{4-s^2}\sqrt{4s^2-(r^2-s^2-1)^2}+ s(r^2-s^2-1)\} \,.
\end{eqnarray*}
But now there is a reflectional symmetry of notation by traversing the polygon in the opposite sense, so we may interchange $\si_j$ with $\si_{M-2-j}$, $r$ with $t$ and $\ta_1$ with $\ta_3$, to obtain
$$
\cos \ta_3 = \frac{1}{4s}\{ \sqrt{4-s^2}\sqrt{4s^2-(t^2-s^2-1)^2}+ s(t^2-s^2-1)\} \,.
$$
Simplification of the various terms now gives the formula of the proposition.  When $s=0$, then $\ta_2=\pm\pi$ and we can pivot the two superimposed edges that result about a common vertex, giving infinitely many solutions obtained by varying $\ta_1$ arbitrarily.    
\hfill $\Box$

\section{The linearized and weak forms of the equations} \label{sec:lin}
We develop the relevant functional analytic framework on a finite graph in order to deduce weak forms of the equations (\ref{one}).  This provides heuristic arguments as to why a pair $(\Ga , \phi )$ consisting of a finite connected graph and a solution $\phi$ to (\ref{one}) should model an elementary particle.  

Let $\Ga = (V,E)$ be a finite graph.  For $x\in V$ define the \emph{tangent space to $\Ga$ at $x$} to be the set of oriented edges with base point $x$ : $T_x\Ga = \{ \vec{xy}: y\sim x\}$.  Define the \emph{tangent bundle to $\Ga$} to be the union: $T\Ga = \cup_{x\in V}T_x\Ga$.  Then a \emph{$1$-form} on $\Ga$ is a map $\om : T\Ga \ra \CC$ such that $\om (\vec{xy}) = - \om (\vec{yx})$.  To a function $\phi : V \ra \CC$, we can naturally associate a $1$-form, the \emph{derivate} $\dd \phi$, by $\dd \phi (\vec{xy}) = \phi (y) - \phi (x)$.  

For two $1$-forms $\om , \eta$, define their \emph{pointwise symmetric product at $x\in V$} by 
$$
\inn{\om , \eta }_x = \sum_{y\sim x}\om (\vec{xy})\eta (\vec{xy})\,,
$$
and their \emph{(global) symmetric product} by
$$
(\om , \eta ) = \sum_{e\in E}\om (e)\eta (e) = \frac{1}{2}\sum_{x\in V}\sum_{y\sim x}\om (\vec{xy})\eta (\vec{xy})\,.
$$
Note that in the first sum the $1$-forms act on unoriented edges so that only their product is well-defined; the factor of one half occurs in the second sum, since there, unoriented edges are counted twice. 

For functions $\phi , \psi : V \ra \CC$, define their \emph{(global) symmetric product} by 
$$
(\phi , \psi ) = \sum_{x\in V}n(x)\phi (x)\psi (x)\,,
$$
where $n(x)$ is the degree of vertex $x$.

The above definitions are the complex symmetric analogues of standard $L^2$ products that arise in functional analytic theory on a graph; in the latter situation they are replaced by Hermitian products rather than symmetric products \cite{Ch}. 

Given a function $\xi : V \ra \CC$ and a $1$-form $\om : T\Ga \ra \CC$, we can define a new $1$-form $\xi \om$ by
$$
(\xi \om )(\vec{xy}) = \frac{1}{2}(\xi (x) + \xi (y))\om (\vec{xy})\,.
$$
Then it is easily checked that
$$
\dd (\phi \psi ) = \phi \dd \psi + \psi \dd \phi\,.
$$

Given a $1$-form $\om$, define its \emph{co-deriviative} $\dd^*\om$ to be the function which at each vertex $x\in V$ is given by
$$
\dd^*\om (x) = - \frac{1}{n(x)}\sum_{y\sim x}\om (\vec{xy})\,.
$$
Then for a function $\phi : V \ra \CC$, 
$$
\dd^*\dd\phi = - \frac{1}{n(x)}\sum_{y\sim x}\dd\phi (\vec{xy}) = - \frac{1}{n(x)}\sum_{y\sim x}(\phi (y)-\phi (x)) = - \Delta \phi\,.
$$

\begin{lemma} \label{lem:functional-formulae} Let $\phi , \psi : V \ra \CC$ be functions and $\om : T\Ga \ra \CC$ a $1$-form.  Then the following formulae hold:

{\rm (i)}  $(\dd\phi , \om ) = ( \phi , \dd^*\om )$\,;

{\rm (ii)} $(\Delta \phi , \psi ) = - ( \dd\phi , \dd \psi )$\,;

{\rm (iii)}  $(\Delta \phi , \psi ) = (\phi , \Delta \psi )$\,;

{\rm (iv)}  $\dd^*(\phi \om )(x) = \phi (x) \dd^*\om (x) - \frac{1}{2n(x)}\inn{\dd\phi , \om}_x$ for each $x\in V$.
\end{lemma}
\begin{proof}  To prove (i), we notice that for $x\sim y$, the sum
$$
(\dd\phi , \om )  =  \frac{1}{2}\sum_{x\in V}\sum_{y\sim x} (\phi (y) - \phi (x)) \om (\vec{xy}) 
$$
contributes $(\phi (y) - \phi (x)) \om (\vec{xy}) = - \phi (x) \om (\vec{xy}) - \phi (y)\om (\vec{yx})$ (the term being symmetric in $x$ and $y$), which equates to the corresponding terms in the sum
$$
(\phi , \dd^*\om ) = - \sum_{x\in V} \phi (x) \sum_{y\sim x} \om (\vec{xy})\,.
$$
The identity (ii) now follows from the fact that $\dd^*\dd\phi = - \Delta \phi$.  Identity (iii) follows from (ii), by symmetry.  Finally
\begin{eqnarray*}
\dd^*(\phi \om )(x) & = & - \frac{1}{n(x)}\sum_{y\sim x}(\phi \om )(\vec{xy}) \\
 & = & - \frac{1}{2n(x)} \sum_{y\sim x} (\phi (x) + \phi (y)) \om (\vec{xy}) \\
  & = & \phi (x) \dd^*\om (x) - \frac{1}{2n(x)} \sum_{y\sim x}(\phi (y) - \phi (x))\om (\vec{xy}) \\
  & = & \phi (x) \dd^*\om (x) - \frac{1}{2n(x)}\inn{\dd \phi , \om}_x\,,
  \end{eqnarray*}
  which gives (iv).
\end{proof}

With the above notation and formulae established, we can give a weak form of equation (\ref{one}).

\begin{proposition} \label{prop:weak} The equation {\rm (\ref{one})} holds if and only if
\begin{equation} \label{weak}
(\Delta \phi ,  \xi \ga \Delta \phi ) - 2(\dd\phi , n\xi \dd \phi ) = 0\,,
\end{equation}
for any function $\xi : V \ra \CC$.
\end{proposition}
\begin{proof} Let $\xi : V \ra \CC$.  Then if (\ref{one}) holds, we have:
$$
\sum_{x\in V}n(x) \xi (x)\left(\ga (x) (\Delta \phi (x))^2- \dd\phi (x)^2\right) = 0\,,
$$
equivalently
$$
(\Delta \phi , \xi \ga \Delta \phi ) - (\dd\phi^2, \xi ) = 0\,,
$$
where we recall that $\dd\phi^2(x) = \sum_{y\sim x} (\phi (y) - \phi (x))^2 = \inn{\dd\phi , \dd\phi }_x$.  But we claim that $(\dd\phi^2, \xi ) = 2(\dd\phi , \xi \dd\phi )$.  Indeed, 
$$
(\dd\phi^2, \xi ) = \sum_{x\in V}n(x) \xi (x) \left( \sum_{y\sim x}(\phi (y) - \phi (x))^2\right)\,,
$$
which for each $x\sim y$, contributes the term $(n(x)\xi (x) + n(y)\xi (y))(\phi (y)-\phi (x))^2$.  But precisely this term occurs with in
$$
2(\dd \phi , n\xi \dd\phi ) = \frac{1}{2}\sum_{x\in V} \sum_{y\sim x} (n(x)\xi (x) + n(y)\xi (y))(\phi (y) - \phi (x))^2\,.
$$
(on noting the symmetry of the expression to be summed on the right-hand side).  

Conversely, if we fix a vertex $x\in V$ and consider the function $\xi : V \ra \CC$ given by $\xi (x) = 1$ and $\xi (y) = 0$ for all $y \neq x$, then (\ref{weak}) gives $\ga (x) \Delta \phi (x)^2 - \dd\phi^2(x) = 0$, so that (\ref{one}) holds.
\end{proof}

On taking the function $\xi$ to be identically equal to $1$, we obtain the following consequence.

\begin{corollary} \label{cor:weak}  Let $\phi : V \ra \CC$ be a solution to equation {\rm (\ref{one})} with $\ga$ constant.  Then
$$
(\ga \Delta (\Delta \phi ) + 2n\Delta \phi +\frac{1}{n}\inn{\dd n, \dd\phi}, \phi ) = 0\,.
$$
\end{corollary}

This corollary suggests heuristic arguments as to why we might consider a pair $(\Ga , \phi )$ consisting of a connected graph endowed with a solution $\phi$ to (\ref{one}) with $\ga$ constant as a particle with mass inversely proportional to $|\ga |$; in the case when $\ga = 0$, we will view the pair as representing a massless particle. 

Firstly, we do not admit any fixed background with respect to define parameters of equations: the particle in a sense creates its own background, so we view an equation of the form $(\Pp (\phi ), \phi )$ as appropriate, where $\Pp$ is some (discrete) differential operator.  In the case when $n$ is constant, we now note the relation between the operator $\Pp (\phi ) = \ga \Delta (\Delta \phi ) + 2n\Delta \phi$ and the operator on the left-hand side of the time-independent Schr\"odinger equation on a fixed smooth background:
$$
\left( -\frac{\hbar^2}{2m}\vec{\na}^2 + V(x)\right) \psi = E\psi (x)\,,
$$  
when $\psi$ is identified with $\Delta \phi$.  The case of mass zero (when $\ga \equiv 0$) is justified in some detail in \cite{Ba-We}. 

In order to define the linearized equation, we consider a family $\{ \phi_t\}$ of functions such that $\phi_0 = \phi$ solves (\ref{one}) with $\ga$ independent of $t$.  On writing $\xi (x) = \frac{\pa \phi (x)}{\pa t}\vert_{t=0}$, we obtain the equation linear in $\xi$:
\begin{equation} \label{linearized}
\ga (x) \Delta \phi (x) \Delta \xi (x) = \inn{\dd\phi (x), \dd \xi (x)}_x\,.
\end{equation}
In the case when $\ga \equiv 0$, that is $\phi$ is holomorphic, this takes the particularly simple form:
\begin{equation} \label{lin-holo}
\inn{\dd\phi (x), \dd \xi (x)}_x=0\,,
\end{equation}
for all $x\in V$.  In all case, we see that $\xi = \la \phi + \mu$ solves the linearized equation ($\la , \mu \in \CC$ constant); this reflects the normalisation freedom $\phi \mapsto \la \phi + \mu$. 

The natural class of mappings between graphs which preserve equation (\ref{one}) are the so-called \emph{holomorphic} mappings.  These were introduced for simple graphs under the name \emph{semi-conformal} mappings by Urakawa \cite{Ur-1, Ur-2}, as the class of maps which preserve local harmonic functions (harmonic at a vertex).  The notion was later extended to non-simple graphs by Baker and Norine \cite{Ba-No-1, Ba-No-2}, who used the term holomorphic mapping.  In \cite{Ba-We}, it was shown that the holomorphic mappings are precisely the class of mappings which preserve local \emph{holomorphic} functions, as we have defined them by (\ref{holo}). 

The definition requires that we restrict to mappings of the vertices that respect the relation of adjacency.  Thus we define a \emph{mapping $f:\Ga = (V,E) \ra \Si = (W,F)$ between graphs} as a mapping of the vertices such that $x\sim y$ implies \emph{either} $f(x) = f(y)$ \emph{or} $f(x) \sim f(y)$.  

\begin{definition} Let $f: \Ga = (V,E) \ra \Si = (W,F)$ be a mapping between graphs.  Then $f$ is \emph{holomorphic} if there exists a function $\la : V \ra \NN$ such that for all $x\in V$ and for all $z'\sim z = f(x)$, we have
$$
\la (x) = \la (x, z') = \sharp \{ x'\in V : x'\sim x, f(x') = z'\}\,,
$$
is independent of the choice of $z'$; we set $\la (x)=0$ if $f(x')=z$ for all $x'\sim x$.  Call $\la$ the \emph{dilation of} $f$.
\end{definition}

\begin{proposition} \label{prop:holo}  Let $f: \Ga = (V,E) \ra \Si = (W,F)$ be a holomorphic mapping between graphs of dilation $\la : V \ra \NN$.  Suppose $\psi : W \ra \CC$ satisfies the equation
$$
\mu (\Delta \psi )^2 = (\dd\psi )^2\,,
$$
for some $\mu : W \ra \RR$.  Then for each $x\in V$ such that $\la (x) \neq 0$, the function $\phi = \psi \circ f$ satisfies {\rm (\ref{one})} at $x$ with
$$
\ga (x) = \frac{n(x)\mu (f(x))}{\la (x) m(f(x))}\,,
$$
where $m(f(x))$ denotes the degree of the vertex $f(x)\in W$.
\end{proposition}

\begin{proof}  Let $f: \Ga = (V,E) \ra \Si = (W,F)$ be a holomorphic mapping between graphs of dilation $\la : V \ra \NN$.  Let $x\in V$ and set $z = f(x)$.  Then
$$
\frac{\mu (z)}{m(z)} \left(\sum_{z'\sim z}(\psi (z')-\psi (z)\right)^2  =  \sum_{z'\sim z}(\psi (z')-\psi (z))^2.
$$
Since $f$ is holomorphic
$$
\sum_{x'\sim x}[(\psi \circ f)(x')-(\psi \circ f)(x)] = \la (x)\sum_{z'\sim z}(\psi (z')-\psi (z))\,.
$$
Suppose that $\la (x) \neq 0$.  Then
\begin{eqnarray*}
\sum_{x'\sim x}[(\psi \circ f)(x')-(\psi \circ f)(x)]^2 & = & \la (x)\sum_{z'\sim z}(\psi (z')-\psi (z))^2 \\
 & = & \frac{\la (x)\mu (z)}{m(z)} \left( \sum_{z'\sim z}(\psi (z')-\psi (z))\right)^2 \\
  & = & \frac{\mu (z)}{\la (x)m(z)} \left( \sum_{x'\sim x}[(\psi \circ f)(x')-(\psi \circ f)(x)]\right)^2\,,
  \end{eqnarray*}
 from which the formula follows. 
If on the other hand $\la (x) = 0$, then $f(x') = f(x)$ for all $x'\sim x$ and both sides of (\ref{one}) vanish.
\end{proof}

\end{document}